\documentclass[sigconf, authorversion,nonacm, screen]{acmart}

\setcopyright{none}

\newfont{\mycrnotice}{ptmr8t at 7pt}
\newfont{\myconfname}{ptmri8t at 7pt}
\let \originalleft \left
\let\originalright\right

\renewcommand{\left}{\mathopen{}\mathclose\bgroup\originalleft}
\renewcommand{\right}{\aftergroup\egroup\originalright}

\usepackage{bm}
\usepackage{paralist}
\usepackage{booktabs}

\usepackage{amsthm}

\usepackage{color}
\usepackage{arydshln}
\usepackage{enumitem}
\usepackage{listings}
\usepackage{thm-restate}
\usepackage{hyperref}
\usepackage{xspace}
\usepackage{algorithmicx}
\usepackage{subcaption,graphicx,url,multirow}
\usepackage{algorithm}
\usepackage[noend]{algpseudocode}

\newtheorem{theorem}{Theorem}

\newtheorem{corollary}{Corollary}

\newcommand{\fullver}[2]{{{\ifx\submissionversion\undefined#1\else#2\fi}}}

\newcommand{\framework}{\textsc{ConnectIt}\xspace}
\newcommand{\oursystem}{\textsc{ConnectIt}\xspace}

\newcommand{\makeset}{\textsc{MakeSet}}
\newcommand{\union}{\textsc{Union}}
\newcommand{\find}{\textsc{Find}}

\newcommand{\maxpathlen}{Max Path Length}
\newcommand{\totalpathlen}{Total Path Length}
\newcommand{\unionfind}{union-find}
\newcommand{\components}{connected components}
\newcommand{\minbased}{min-based}
\newcommand{\otherminbased}{other min-based}
\newcommand{\labelvar}{\ensuremath{C}}
\newcommand{\largestcomp}{\ensuremath{L_{\max}}}
\newcommand{\samplecorrect}{correct}

\newcommand{\koutsample}{$k$-out Sampling}
\newcommand{\bfssample}{BFS Sampling}
\newcommand{\lddsample}{LDD Sampling}

\newcommand{\koutpure}{$\mathsf{kout\mhyphen pure}$}
\newcommand{\kouthybrid}{$\mathsf{kout\mhyphen hybrid}$}
\newcommand{\koutafforest}{$\mathsf{kout\mhyphen afforest}$}
\newcommand{\koutmaxdeg}{$\mathsf{kout\mhyphen maxdeg}$}

\definecolor{mygreen}{rgb}{0.0, 0.5, 0.0}

\newcommand{\stinger}{STINGER}
\newcommand{\mpc}{\ensuremath{\mathsf{MPC}}}

\newcommand{\connect}{\ensuremath{\mathsf{Connect}}}
\newcommand{\parentconnect}{\ensuremath{\mathsf{ParentConnect}}}
\newcommand{\extendedconnect}{\ensuremath{\mathsf{ExtendedConnect}}}
\newcommand{\rootupdate}{\ensuremath{\mathsf{RootUp}}}
\newcommand{\shortcut}{\ensuremath{\mathsf{Shortcut}}}
\newcommand{\alter}{\ensuremath{\mathsf{Alter}}}
\newcommand{\fullshortcut}{\ensuremath{\mathsf{FullShortcut}}}

\mathchardef\mhyphen="2D
\newcommand{\unionasync}{\ensuremath{\mathsf{UF\mhyphen Async}}}
\newcommand{\unionhook}{\ensuremath{\mathsf{UF\mhyphen Hooks}}}
\newcommand{\unionearly}{\ensuremath{\mathsf{UF\mhyphen Early}}}
\newcommand{\unionremlock}{\ensuremath{\mathsf{UF\mhyphen Rem\mhyphen Lock}}}
\newcommand{\unionremcas}{\ensuremath{\mathsf{UF\mhyphen Rem\mhyphen CAS}}}
\newcommand{\ufmetaalgorithm}{\ensuremath{\mathsf{Connectivity (Union\mhyphen Find)}}}
\newcommand{\jayanti}{\ensuremath{\mathsf{UF\mhyphen JTB}}}
\newcommand{\liutarjan}{\ensuremath{\mathsf{Liu\mhyphen Tarjan}}}
\newcommand{\shiloachvishkin}{\ensuremath{\mathsf{SV}}}
\newcommand{\labelpropagation}{\ensuremath{\mathsf{Label\mhyphen Prop}}}
\newcommand{\bfscc}{BFSCC}
\newcommand{\gbbscc}{WorkeffCC}
\newcommand{\labelprop}{label propagation}

\definecolor{amaranth}{rgb}{0.9, 0.17, 0.31}
\definecolor{revise-color}{rgb}{0.55, 0.71, 0.0} 
\newcommand{\revised}[1]{#1}
\newcommand{\sndrevised}[1]{#1}

\newcommand{\splitone}{\ensuremath{\mathsf{SplitAtomicOne}}}
\newcommand{\halveone}{\ensuremath{\mathsf{HalveAtomicOne}}}
\newcommand{\splice}{\ensuremath{\mathsf{SpliceAtomic}}}
\newcommand{\findnaive}{\ensuremath{\mathsf{FindNaive}}}
\newcommand{\findcompress}{\ensuremath{\mathsf{FindCompress}}}
\newcommand{\findhalve}{\ensuremath{\mathsf{FindAtomicHalve}}}
\newcommand{\findsplit}{\ensuremath{\mathsf{FindAtomicSplit}}}
\newcommand{\findtwotry}{\ensuremath{\mathsf{FindTwoTrySplit}}}

\newcommand{\mapedges}{\ensuremath{\textsc{MapEdges}}}
\newcommand{\gatheredges}{\ensuremath{\textsc{GatherEdges}}}

\newcommand{\cas}{CAS}

\newcommand{\writemin}{writeMin}

\newcommand{\codevar}[1]{\mathit{#1}}

\newcommand{\boruvka}{Bor\r{u}vka}

\usepackage[most]{tcolorbox}
\tcbuselibrary{breakable}

\theoremstyle{definition}
\newtheorem{definition}{Definition}[section]

\newcounter{myalgctr}

\newtcolorbox{OuterBox}[1][]{%
    breakable,
    enhanced,
    frame hidden,
    interior hidden,
    left=-5pt,
    right=-5pt,
    top=-5pt,
    float=p,
    boxsep=0pt,
    arc=0pt
#1}%

\newtcolorbox{InnerBox}[1][]{%
    enforce breakable,
    enhanced,
    colback=gray,
    colframe=white,
#1}%

\newenvironment{tbox}{
\vspace{0.2cm}
\begin{tcolorbox}[width=0.45\textwidth,
                  enhanced,
                  boxsep=1pt,
                  left=1pt,
                  right=1pt,
                  top=0.5pt,
                  boxrule=1pt,
                  arc=0pt,
                  colback=white,
                  colframe=black,
                  unbreakable
                  ]
}{
\end{tcolorbox}
}

\newenvironment{mytbox}{
\vspace{0.2cm}
\begin{tcolorbox}[width=\textwidth,
                  enhanced,
                  boxsep=1pt,
                  left=1pt,
                  right=1pt,
                  top=4pt,
                  boxrule=1pt,
                  arc=0pt,
                  colback=white,
                  colframe=black,
                  unbreakable
                  ]
}{
\end{tcolorbox}
}

\newcommand{\tboxhrule}[0]{\vspace{0.1cm} {\color{black} \hrule} \vspace{0.2cm}}

\newenvironment{titledtbox}[1]{\begin{tbox}#1 \tboxhrule}{\end{tbox}}

\newenvironment{mytitledtbox}[1]{\begin{mytbox}#1 \tboxhrule}{\end{mytbox}}

\newcommand{\defn}[1]{\emph{\textbf{#1}}}

\newcommand{\myparagraph}[1]{\vspace{1.5pt}\noindent {\bf #1.}}
\newcommand{\id}[1]{\ifmmode\mathit{#1}\else\textit{#1}\fi}
\newcommand{\const}[1]{\ifmmode\mbox{\textc{#1}}\else\textsc{#1}\fi}

\usepackage[font=small,aboveskip=0pt,belowskip=0pt]{caption}
\usepackage{microtype}

\begin{document}
\date{}

\title{ConnectIt: A Framework for Static and Incremental Parallel Graph Connectivity Algorithms
}
\titlenote{This is an extended version of a paper in PVLDB (to be presented at VLDB’21).}


\author{Laxman Dhulipala}
\affiliation{%
  \institution{MIT CSAIL}
}
\email{laxman@mit.edu}

\author{Changwan Hong}
\affiliation{%
  \institution{MIT CSAIL}
}
\email{changwan@mit.edu}

\author{Julian Shun}
\affiliation{%
  \institution{MIT CSAIL}
}
\email{jshun@mit.edu}

\fancyhead{}

\begin{abstract}
Connected components is a fundamental kernel in graph applications.
The
fastest existing  multicore algorithms for solving graph connectivity are
based on some form of edge sampling and/or linking and compressing
trees. However, many combinations of these design choices have been
left unexplored. In this paper, we design the \framework{} framework,
which provides different sampling strategies as well as various tree
linking and compression schemes.
\framework{} enables us to obtain
several hundred new variants of connectivity algorithms, most of which
extend to computing spanning forest. In addition to static graphs, we
also extend \framework{} to support mixes of insertions and
connectivity queries in the concurrent setting.

We present an experimental evaluation of \framework{} on a 72-core
machine, which we believe is the most comprehensive evaluation of
parallel connectivity algorithms to date. Compared to a collection of
state-of-the-art static multicore algorithms, we obtain an average
speedup of 12.4x (2.36x average speedup over the fastest existing
implementation for each graph). Using \framework{}, we are able to
compute connectivity on the largest publicly-available graph (with
over 3.5 billion vertices and 128 billion edges) in under 10 seconds
using a 72-core machine, providing a 3.1x speedup over the fastest
existing connectivity result for this graph, in any computational
setting. For our incremental algorithms, we show that our algorithms
can ingest graph updates at up to several billion edges per second.
To guide the user in selecting the best variants in
\framework{} for different situations, we provide a detailed analysis
of the different strategies.
Finally, we show how the techniques in \framework can be used to speed up two important graph applications: approximate minimum spanning forest and SCAN clustering.

\end{abstract}

\maketitle


\section{Introduction}\label{sec:intro}
Computing the \components{} (connectivity) of an undirected
graph is a fundamental problem for which numerous algorithms have been
designed. In the \components{} problem, we are given an
undirected graph and the goal is to assign labels to the vertices such
that two vertices reachable from one another have the same label, and
otherwise have different labels~\cite{CLRS}. A recent paper by Sahu et
al.~\cite{sahu2017ubiquity} surveying industrial uses of graph
algorithms shows that connectivity is the most frequently performed
graph computation out of a list of 13 fundamental graph routines
including shortest paths, centrality computations, triangle
counting, and others. Computing \components{} is also used to solve
many other graph problems, for example, to solve biconnectivity and
higher-order connectivity~\cite{Tarjan85}, as well as a subroutine in
popular clustering algorithms~\cite{Ester1996, xu2007scan, Patwary12,
Gan2015,Fang2016, wen2017efficient}.

In the sequential setting, \components{} can be easily solved
using breadth-first search, depth-first search, or union-find.
However, it is important to have fast parallel algorithms for
the problem in order to achieve high performance. Many parallel algorithms for \components{}
have been proposed in the literature (see,
e.g.,~\cite{Awerbuch1987,Chin1982,Chong95,Cole1991,Han1990,Hirschberg1979,Iwama1994,Johnson,Karger1999,Koubek85,Kruskal90,NathM82,Phillips89,Reif85,ShiloachV82,Gazit1991,Vishkin1984,SDB14,DhBlSh18,Greiner94,HalperinZ94,Slota14,BFGS12,Bader2005a,BaderJ96,MadduriB09,Bader2005b,Bus01,Caceres04,Krishnamurthy94,PatwaryRM12,Hambrusch,Hsu97,Nguyen2013,Hawick2010,Soman,Banerjee2011,ShunB2013,Sutton2018,Jaiganesh2018,LiuT19,Stergiou2018,Jain2017,Kiveris2014,Iverson2015,Feng2018,Cong2014,Meng2019,andoni2018parallel,behnezhad2019massively,behnezhad2019near,behnezhad2020massively,zhang2020fastsv,Pai2016,Ben-Nun2017,Rastogi2013,Yan2014},
among many others).  Recent state-of-the-art parallel implementations
are based on graph traversal~\cite{Slota14,DhBlSh18,SDB14,ShunB2013},
label propagation~\cite{Slota14,Nguyen2013,ShunB2013},
union-find~\cite{BFGS12,PatwaryRM12,Jaiganesh2018}, or the
hook-compress
paradigm~\cite{LiuT19,ShiloachV82,Bader2005a,Awerbuch1987,Cong2014,BeamerAP15,Greiner94,Soman,Meng2019,Wang2017,Pai2016,Ben-Nun2017,zhang2020fastsv}.
Recent work by Sutton et al.~\cite{Sutton2018} uses sampling to find
the \components{} on a subset of the edges, which can be used
to reduce the number of edge inspections when running
connectivity on the remaining edges. However, most prior work has
provided only one, or a few implementations of a specific approach for
a particular architecture, and there are many variants of these
algorithmic approaches that have been left unexplored.

In this paper, we design the \framework framework for multicore CPUs,
which enables many possible implementation choices of the algorithmic
paradigms for parallel connectivity from the literature.
Furthermore, as many real-world graphs are frequently updated under
insertion-heavy workloads (e.g., there are about 6,000 tweets
per second on Twitter, but only a few percent of tweets are
deleted~\cite{almuhimedi2013tweets}), \framework provides algorithms
that can maintain connectivity under incremental updates (edge
insertions) to the graph. A subset of the \framework{} implementations
also support computing the spanning forest of a graph in
both the static and incremental settings. We focus on the multicore
setting as the largest publicly-available real-world graphs can fit in
the memory of a single machine~\cite{DhBlSh18, dhulipala2020semi}. We
also compare \framework{}'s results with reported results for the
distributed-memory setting, showing that our multicore solutions are
significantly faster and much more cost-efficient.
We have recently extended our techniques to the GPU setting~\cite{HongDS20}.
\subsection*{\framework Overview}

\myparagraph{Algorithms}
\framework{} is designed for \emph{\minbased{} connectivity
algorithms}, which are based on vertices propagating labels to other
vertices that they are connected to, and updating labels based on the
minimum label received. All of the algorithms conceptually view the
label of a vertex $v$ as a directed edge from $v$ to the vertex
corresponding to $v$'s label. Thus, these directed edges form a set of
directed trees. All of the algorithms that we study maintain acyclicity in
this forest (ignoring self-loops at the roots of trees). The
\minbased{} algorithms that we study include both
\emph{root-based algorithms}, which include a broad class of union-find
algorithms and several other algorithms that only modify the labels
of roots of trees in the forest, as well as \emph{\otherminbased{}
algorithms} which deviate from this rule and can modify the labels of
non-root vertices.

\myparagraph{Sampling and Two-phase Execution} Inspired by the
Afforest algorithm~\cite{Sutton2018}, \framework produces algorithms
with two phases: the \defn{sampling} phase and the \defn{finish}
phase. In the sampling phase, we run \components{} on a subset of the
edges in the graph, which assigns a temporary label to each vertex. We
then find the most frequent label, $L_{\max}$, which corresponds to the
ID of the largest component (not necessarily maximal) found so far.
In the finish phase, we only need to run \components{} on the incident
edges of vertices with a label not equal to $L_{\max}$, which can
significantly reduce the number of edge traversals. We observe that
this optimization is similar in spirit to the direction-optimization
in breadth-first search~\cite{Beamer12}, which skips over incoming
edges for vertices, once they have already been visited during dense
iterations. If sampling is not used, then the finish phase is
equivalent to running a min-based algorithms on all vertices and
edges.

Our main observation is that these sampling techniques are general
strategies that reduce the number of edge traversals in the finish
phase, and thus accelerate the overall algorithm.  In \framework, any
of the \minbased{} algorithms that we consider can be used for the finish
phase in combination with any of the three sampling schemes: $k$-out, breadth-first search, and low-diameter decomposition sampling.
For the implementations where the finish phase uses a root-based
algorithm, \framework also supports spanning forest computation.  Our
generalized sampling paradigm and integration into \framework enables
us to express over 232 combinations of parallel algorithms of parallel
graph connectivity and 192 implementations for parallel spanning
forest (only our root-based algorithms support
spanning forest).

\myparagraph{Incremental Connectivity} Due to the frequency of updates
to graphs, various parallel streaming algorithms for connected
components have been
developed~\cite{McColl13,Ediger12,Simsiri2017,Sengupta17,Acar2019,dhulipala2020parallel}.
Motivated by this, the \framework framework supports a combination of
edge insertions and connectivity queries in the graph.  Both the
union-find and root-based algorithm implementations in \framework
support batched edge insertions/queries.  Additionally, all of the
union-find implementations, other than the variants of Rem's algorithm combined with
the splice compression scheme, support asynchronous updates and
queries, and most of them are lock-free or wait-free.

\setlength{\tabcolsep}{2.2pt}
\begin{table}[!t]\footnotesize
\centering

\scalebox{0.95}{
\begin{tabular}[!t]{llrrrr}
\toprule
{\bf System} & {\bf Graph} & {\bf Mem. (TB)} & {\bf Threads} & {\bf Nodes} & {\bf Time (s)} \\
\midrule
\multirow{1}{*}{Mosaic~\cite{maass2017mosaic}}
& Hyperlink2014 & 0.768 & 1000 & 1 & 708  \\

\multirow{1}{*}{FlashGraph~\cite{da2015flashgraph}}
& Hyperlink2012 & .512 & 64 & 1 & 461  \\

\multirow{1}{*}{GBBS~\cite{DhBlSh18}}
& Hyperlink2012 & 1 & 144 & 1 & 25.8  \\

\multirow{1}{*}{GBBS (NVRAM)~\cite{dhulipala2020semi}}
& Hyperlink2012 & 0.376 & 96 & 1 & 36.2  \\

\multirow{1}{*}{Galois (NVRAM)~\cite{Gill2020}}
& Hyperlink2012 & 0.376 & 96 & 1 & 76.0  \\

\multirow{1}{*}{Slota et al.~\cite{Slota2016}}
& Hyperlink2012 & 16.3 & 8192 & 256 & 63  \\

\multirow{1}{*}{Stergiou et al.~\cite{Stergiou2018}}
& Hyperlink2012 & 128 & 24000 & 1000 & 341 \\

\multirow{1}{*}{Gluon~\cite{Dathathri2018}}
& Hyperlink2012 & 24 & 69632 & 256 & 75.3  \\

\multirow{1}{*}{Zhang et al.~\cite{zhang2020fastsv}}
& Hyperlink2012 & $\geq 256$ & 262,000 & 4096 & 30  \\

\midrule
\multirow{2}{*}{\framework{}}
& Hyperlink2014 & 1 & 144 & 1 & {\bf \color{mygreen}2.83} \\
& Hyperlink2012 & 1 & 144 & 1 & {\bf \color{mygreen}8.20} \\
\bottomrule

\end{tabular}
}
\captionof{table}{\small
  System configurations, including memory (terabytes), num.
  hyper-threads and nodes, and running times (seconds) of connectivity
  results on the Hyperlink graphs.
  The last rows show the fastest \framework{} times. The
  fastest time per graph is shown in green. }
  \label{table:big_comparison}
\end{table}

\myparagraph{Experimental Evaluation}
We conduct a comprehensive experimental evaluation of
all of the connectivity and spanning forest implementations in \framework on a 72-core
multicore machine, in both the static and incremental setting.
Compared to existing work, the \framework implementations significantly
outperform state-of-the-art implementations using the new sampling
techniques proposed in this paper, and often outperform prior work
even without applying sampling. Our fastest algorithms using sampling
are always faster than the fastest currently available implementation.
As an example of our high performance, \framework is able to produce
implementations that can process the Hyperlink2012
graph~\cite{meusel15hyperlink}, which is the largest publicly-available
real-world graph and has over 3.5 billion vertices and 225 billion
directed edges,\footnote{We symmetrize the graphs to obtain an undirected graph for connectivity, and when reporting the number of edges we count each edge once per direction.} on a single 72-core machine with a terabyte of RAM. We
show existing results on this graph, and the smaller Hyperlink2014
graph (1.7 billion vertices and 128 billion directed edges) in
Table~\ref{table:big_comparison}. Our running times are between
3.65--41.5x faster than existing distributed-memory results that
report results for Hyperlink2012, while using orders of magnitude
fewer computing resources.
Finally, we show that \framework can be used to speed up two graph
applications: approximate minimum spanning forest and index-based SCAN
clustering.

\myparagraph{Contributions}
The contributions of this paper are as follows.

\begin{enumerate}[label=(\textbf{\arabic*}),topsep=1pt,itemsep=0pt,parsep=0pt,leftmargin=15pt]
\item We introduce \framework{}, which provides several hundred
different multicore implementations of connectivity, spanning forest,
and incremental connectivity, most of which are new.

\item \framework provides different choices of sampling methods based
  on provably-efficient graph algorithms, which can be used in two-phase
  execution to reduce the number of edge traversals.

\item The fastest implementation in \framework achieves an average
  speedup of 2.3x (and ranges from 1.5--4.02x speedup) over the fastest
  existing static multicore connectivity algorithms.

\item For incremental connectivity, the multicore implementations in
  \framework achieve a speedup of 1,461--28,364x over existing
  multicore solutions and a throughput (in terms of directed edges) of
  between 4.6 million insertions per second for very small batch sizes
  to 7.1 billion insertions per second for large batches, on graphs of
  varying sizes.

\item We present a detailed experimental analysis of the different
  implementations in \framework to guide the user into finding the
  most efficient implementation for each situation.

\item \revised{We show  that \framework can be used to speed up approximate minimum spanning forest by 2.03--5.36x  and index-based SCAN clustering by 42.5--50.5x.}

\item The \framework source code is available at \url{https://github.com/ParAlg/gbbs/tree/master/benchmarks/Connectivity/ConnectIt}.
\end{enumerate}

\begin{figure*}
\begin{center}
\vspace{-3em}
\includegraphics[width=0.85\textwidth]{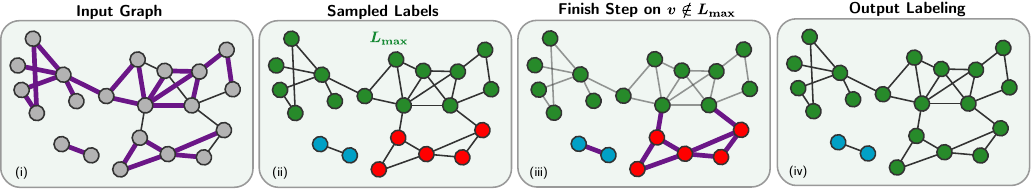}
\vspace{0.25em}
\caption{\label{fig:samplefig}
This figure illustrates the \framework{} framework for connectivity using
\koutsample{} on a small input graph. {\bf (i)} illustrates the input graph. The bolded purple
edges are those selected by the $k$-out sample with $k=2$. {\bf(ii)}
shows the partial connectivity labeling after computing connectivity
on the sampled edges (e.g., using a union-find algorithm).
\largestcomp{} indicates the vertices in the largest component (shown in
green).  {\bf(iii)} shows the edges (bolded in purple) which still must be processed in
the finish step, namely all edges incident to vertices $v \notin
\largestcomp{}$. Lastly, {\bf(iv)} shows the output connectivity
labeling.
}
\end{center}
\end{figure*}

\section{Preliminaries}\label{sec:prelims}

\myparagraph{Graph Notation and Formats} We denote an unweighted graph
by $G(V, E)$, where $V$ is the set of vertices and $E$ is the set of
edges in the graph. We use $n=|V|$ to refer to the number of vertices
and $m=|E|$ to refer to the number of directed edges.
In our graph
representations, vertices are indexed from $0$ to $n-1$, and we consider
two graph formats---\defn{compressed sparse row (CSR)} and
\defn{edge/coordinate list (COO)}.  In CSR, we are given two arrays,
$I$ and $A$, where the incident edges of a vertex $v$ are stored in
$\{A[I[v]],\ldots,A[I[v+1]-1]\}$ (we assume $A[n]=m$). In
COO, we are given an array of pairs $(u,v)$ corresponding to edge
endpoints. Unless otherwise mentioned, we store graphs in the CSR

\myparagraph{Compare-and-Swap}
A \defn{compare-and-swap (CAS)} takes three arguments: a memory
location \emph{x}, an old value \emph{oldV}, and a new value
\emph{newV}. If the value stored at \emph{x} is equal to \emph{oldV},
the CAS atomically updates the value at \emph{x} to be \emph{newV} and
returns \emph{true}; otherwise the CAS returns \emph{false}. CAS is
supported by most modern processors.

\fullver{
\myparagraph{Work-Depth Model}
We analyze our parallel algorithms using the \defn{work-depth
measure}~\cite{JaJa92, CLRS}. The \defn{work} is the number of
operations used by the algorithm and the \defn{depth} is the length of
the longest sequential dependence.  The work-depth measure is a
fundamental tool in analyzing parallel algorithms, e.g., see
~\cite{Blelloch2016, DhBlSh18, gu15semisort, SDB14,
sun2019supporting, Sun2018, wang2019pardbscan, dhulipala2020semi} for
a sample of recent practical uses of this model.
}{}

\myparagraph{Linearizability}
Linearizability~\cite{Herlihy90, HS} is the standard correctness
criteria for concurrent algorithms.  A set of operations are
\defn{linearizable} if the result of a concurrent execution is the
same as if the operations were applied at a distinct point in time
(the \defn{linearization point}) between the operation's invocation
and response. \fullver{The sequential ordering of the operations based
on their linearization points is its \defn{linearization order}.}{}

\myparagraph{Graph Connectivity and Related Problems}
A \defn{connected component (CC)} in $G$ is a maximal set of vertices,
such that there is a path between any two vertices in the set.
An algorithm for computing connected components returns a
\defn{connectivity labeling} \labelvar{} for each vertex, such that
$\labelvar{}(u)=\labelvar{}(v)$ if and only if vertices $u$ and $v$
are in the same connected component. We represent
connectivity labelings using an array of $n$ integers. A \defn{partial
connectivity labeling} is a labeling $\labelvar{}$ such that
$\labelvar{}(u) = \labelvar{}(v)$ implies that $u$ and $v$ are in the
same component. A \defn{connectivity query} takes as input two vertex
identifiers and returns \emph{true} if and only if their labels are
the same. A \defn{spanning forest (SF)} in $G = (V,E)$ contains one
tree for each connected component in $G$ containing all vertices of
that component.  A \defn{breadth-first search (BFS)} algorithm takes
as input a graph $G = (V,E)$ and a source vertex $\codevar{src} \in V$, and
traverses the vertices reachable from $\codevar{src}$ in increasing order of
their distance from $\codevar{src}$.  A \defn{low-diameter decomposition (LDD)}
of a graph parameterized by $0<\beta<1$ and $d$ is a partition of the
vertices into $V_1,\ldots,V_k$, such that the shortest path between
any two vertices in the same partition using only intra-partition
edges is at most $d$, and the number of inter-partition edges is at
most $\beta m$~\cite{MillerPX2013}.

\fullver{
A \defn{union-find} (or disjoint-set) data structure maintains a
collection of sets where all elements in each set are associated with
the same label. The data structure is represented as a forest where
each element has a parent pointer (with the root of each tree pointing
to itself), and the label of an element is the root of the tree it
belongs to.  A union-find data structure supports the \makeset{},
\union{}, and \find{} operations~\cite{CLRS}. \makeset{($e$)} creates
a new set containing just the element $e$ with a parent pointer to
itself. \union{($e$,$e'$)} merges together the trees of $e$ and $e'$
into a single tree, if they are not already in the same
tree. \find{($e$)} returns the label of element $e$ by looking up the
root of the tree that $e$ belongs to. The \find{} and \union{}
operations can compress paths in the tree to speed up future
operations.  When used in a CC algorithm, the elements are vertices,
and at the end of the algorithm, \find{($v$)} returns
$\labelvar{}(v)$.  There are various versions of union-find based on
how the \union{} and \find{} operations are implemented, and this
paper explores a large number of union-find implementations for the
concurrent setting.

A \defn{\minbased{}} connectivity algorithm also maintains a set data
structure, like union-find algorithms. A \minbased{} algorithm only
updates the label of an element if the new label is smaller than the
previous one. All of the connectivity algorithms studied in this paper
are \minbased{} algorithms.
A \defn{root-based} connectivity algorithm is a special type of
\minbased{} algorithm. Specifically, a root-based algorithm only links
sets together by adding a link from the root of one tree, to a node in
another tree.
}{}

\section{\framework Framework}\label{sec:framework}
In this section, we define the components of the \framework{}
framework and the specific ways in which they can be combined. We start by
presenting our algorithm for computing connectivity. The algorithm
supports combining multiple sampling methods and finish methods, which
we describe in detail in Sections~\ref{subsec:sampling} and
\ref{subsec:finish}. Figure~\ref{fig:samplefig} illustrates how the
\framework{} framework works for a $k$-out sampling algorithm on an
example graph.

Next, in Sections~\ref{subsec:spanningforest} and
\ref{subsec:streaming}, we describe modifications required to adapt
our framework for spanning forest, and for the
incremental setting. We present correctness proofs for our
algorithms in all of the settings in
\fullver{Appendix~\ref{apx:framework}}{the full version of our
  paper~\cite{connectit}}.  Finally, in
Section~\ref{subsec:framework_implementation} we describe details
regarding our implementation.

\subsection{Connectivity}\label{subsec:connectivity}

\begin{algorithm}
\caption{\framework{} Framework: Connectivity} \label{alg:framework_connectivity}
\small
\begin{algorithmic}[1]
\Procedure{Connectivity}{$G(V, E), \codevar{sample\_f}, \codevar{finish\_f}$}
  \State $\codevar{sampling} \gets \textsc{GetSamplingAlgorithm}(\codevar{sample\_opt})$
  \State $\codevar{finish} \gets \textsc{GetFinishAlgorithm}(\codevar{finish\_opt})$
  \State $\codevar{labels} \gets \{i \rightarrow i\ |\ i \in [|V|]\}$\label{line:init}
  \State $\codevar{labels} \gets \codevar{sampling}.\textsc{SampleComponents}(G, \codevar{labels})$\label{line:sample}
  \State $\largestcomp{} \gets \textsc{IdentifyFrequent}(\codevar{labels})$\label{line:identifyfreq}
  \State $\codevar{labels} \gets \codevar{finish}.\textsc{FinishComponents}(G, \codevar{labels},\largestcomp{})$\label{line:finish}
  \State \algorithmicreturn{}$\ \codevar{labels}$
\EndProcedure
\end{algorithmic}
\end{algorithm}

A connectivity algorithm in \framework{} is instantiated by supplying
$\codevar{sampling}$ and $\codevar{finish}$ methods.
Algorithm~\ref{alg:framework_connectivity} presents the generic
\framework{} connectivity algorithm, parameterized by these
user-defined functions.
\revised{The algorithm first initializes a \defn{connectivity
labeling}, which is represented as an array of $n$ integers, by
setting each vertex's label to be its own ID (Line~\ref{line:init}).}
It then performs a sampling step using the provided
$\codevar{sampling}$ algorithm (Line~\ref{line:sample}). The sampling
step results in partial connectivity information being computed and
stored in the $\codevar{labels}$ array. Next, it identifies
$\bm{\largestcomp{}}$, the most frequently occurring component ID in
the $\codevar{labels}$ array (Line~\ref{line:identifyfreq}). The
identified component is then supplied to the $\codevar{finish}$
algorithm, which finishes computing the connected components of $G$.
The finish method potentially saves significant work by avoiding
processing vertices in the component with an ID of $\largestcomp{}$.

\myparagraph{Properties of Sampling Methods}
\sndrevised{
Before discussing the sampling and finish methods, we discuss the
properties that we require in \framework{} to obtain a correct
parallel connectivity algorithm.

\begin{definition}\label{def:sampling_produces_trees}
Consider a graph $G$. Let $\labelvar{}$ be the labeling produced by a sampling
method $S$ and let the labeling $\labelvar{}' = \textsc{Connectivity}$ $(G[\labelvar{}])$
where $G[\labelvar{}]$ is the graph induced by contracting $G$ using
$\labelvar{}$.\footnote{Contracting $G$ with respect to $\labelvar{}$ creates a new graph
by merging all vertices $v$ with the same label into a single vertex,
and only preserving edges $(u,v)$ such that $\labelvar{}(u) \neq
\labelvar{}(v)$, removing duplicate edges.}
We say that a sampling method $S$ is \defn{\samplecorrect{}} if:

$\labelvar{}'' = \{ \labelvar{}'[\labelvar{}[v]]\ |\ v \in V\}$ is a correct connectivity labeling.
\end{definition}

In other words the definition states that the connectivity label found
by composing the sampled label $(\labelvar{}[v])$ with the
connectivity label of the contracted vertex on the contracted graph,
i.e., $\labelvar{}'[\labelvar{}[v]]$, yields a correct connectivity
labeling.
}

\sndrevised{
\myparagraph{Properties of Finish Methods}
Next, we define the correctness properties for finish methods in
\framework{}. The definition uses the interpretation of labels as
rooted trees, which we discussed above.

\begin{definition}\label{def:monotone}
Let $\labelvar{}$ be a connectivity labeling, which initially maps
every vertex to its own node ID. We call a connectivity algorithm
\defn{monotone} if the algorithm updates the labels such that the
updated labeling can be represented as the union of two trees in the
previous labeling.
\end{definition}

In other words, once a vertex is connected to its parent in a tree, it
will always be in the same tree as its parent.
}
Finally, we define linearizable monotonicity, which is the relaxed
correctness property possessed by most finish algorithms in our framework.

\begin{definition}\label{def:linearizable_monotonicity}
  Given an undirected graph $G(V, E)$, we say that a connectivity
  algorithm operating on a connectivity labeling $\labelvar{}$ is \defn{linearizably monotone} if
  \begin{enumerate}[label=(\textbf{\arabic*}),topsep=1pt,itemsep=0pt,parsep=0pt,leftmargin=15pt]
    \item Its operations are linearizable.
    \item Every operation in the linearization preserves monotonicity.\label{def:linearize_mono_preserves}
  \end{enumerate}
\end{definition}

The finish methods considered in \framework{} are linearizably
monotone, with only a few exceptions. The first is a subset of the
variants of Rem's algorithms (all variants using the \emph{splice}
rule),
which we analyze separately\fullver{
(Theorem~\ref{thm:rem_correct})}{}, and the family of other min-based
connectivity algorithms, which includes a subset of Liu-Tarjan
algorithms, Stergiou's algorithm, Shiloach-Vishkin, and
\labelprop{}\fullver{
(Theorem~\ref{thm:min_based_correctness})}{}.  \fullver{}{
\revised{In the full version of our
paper~\cite{connectit}, we provide
proofs showing that both of these exceptions, and the large class of
linearizably monotone finish methods, are correct when composed with
our sampling schemes.
The key idea for the proof for linearizably monotone algorithms is to
apply induction over the linearization order of the operations, and
apply the fact that each operation is monotone and thus preserves the
partial connectivity information.
The proof for the other min-based algorithms is by contradiction: at
the end of the algorithm, two vertices in the same component must have
the same label, as otherwise the finish algorithm will not have
terminated, even if we ignore the largest sampled component.
}
}

\subsection{Sampling Algorithms}\label{subsec:sampling}
This section introduces the \samplecorrect{} sampling methods used in
\framework{}.  \fullver{We provide pseudocode for our sampling methods
in Appendix~\ref{apx:subsec:sampling_pseudocode}.}{Due to space
constraints, we provide the pseudocode for our sampling methods in the
full version of our paper~\cite{connectit}.}

\myparagraph{\koutsample{}}
The $k$-out method takes a positive integer parameter $k$, selects $k$
edges out of each vertex uniformly at random, and computes the
connected components of this sampled graph. An important result shown
in a recent paper by Holm et al.~\cite{holm2019kout} is that if $nk$
edges are sampled in this way for sufficiently large $k$, only
$O(n/k)$ inter-component edges remain after contraction, in
expectation.
Holm et al.\ conjecture that this fact about the
number of inter-component edges holds for any $k \geq 2$, although the
proof holds for $k = \Omega(\log n)$.
Sutton et al.~\cite{Sutton2018} describe a sampling scheme that
selects the \emph{first $k$ edges} incident to the vertex, which does
not use randomization. However, on some graphs with poor orderings,
this method can result in only a small fraction of the components
being discovered (leaving up to several orders of magnitude more
inter-component edges), which results in a costly sampling step that
provides little benefit.
To improve our results for these poorly ordered graphs while achieving
good performance on graphs where this heuristic performs well, we
select the first edge incident to each vertex, and select the
remaining $k-1$ edges randomly. To the best of our knowledge, using
randomness in this experimental setting has not been explored before.
We provide a full evaluation of different options in
\fullver{Appendix~\ref{sec:sampling_eval}}{the full version of our
paper~\cite{connectit}}.

\myparagraph{Breadth-First Search Sampling}
The breadth-first search (BFS) sampling method is a simple heuristic
based on using a breadth-first search from a randomly chosen source
vertex.  Assuming that the graph contains a massive component,
containing, say, at least a constant fraction of the vertices, running
a BFS from a randomly selected vertex will discover this component
with constant probability.  To handle the case where we are unlucky
and pick a bad vertex, we can apply this process $c$ times.  Setting
$c = \Theta(\log n)$ would ensure that we find the massive component
with high probability. In practice we set $c=3$, and in our
experiments we found that one try was sufficient to discover the
massive component on all of the real-world graphs we test on.
\framework terminates the sampling algorithm when a component
containing more than 10\% of the vertices is found or after $c$
rounds, whichever happens first.

\myparagraph{Low-Diameter Decomposition Sampling}
As discussed in Section~\ref{sec:prelims}, running an LDD algorithm on a graph
with parameter $\beta < 1$ partitions the graph into clusters such
that the strong diameter (the shortest path using only edges inside
the cluster) of each cluster is $O(\log n/\beta)$, and cuts $O(\beta
m)$ edges in expectation~\cite{MillerPX2013}. Shun et al.~\cite{SDB14}
give a simple and practical work-efficient (linear-work) parallel connectivity
algorithm based on recursively applying LDD and performing graph
contraction to recurse on a contracted graph.

In this paper, we consider applying just a \emph{single} round of the
LDD algorithm of Miller, Peng, and Xu~\cite{MillerPX2013}, without
actually contracting the graph after performing the LDD. On
low-diameter graphs, the hope is that much of the largest connected
component will be contained in the most frequent cluster identified.
Our approach is practically motivated by studying the behavior of the
work-efficient connectivity algorithm of Shun et al.~\cite{SDB14} on
low-diameter real-world graphs and observing that after one
application of LDD, the resulting clustering contains a single massive
cluster, and the number of distinct clusters in the contracted graph
is extremely small. We provide an empirical analysis of this sampling
method for different values of $\beta$ in
\fullver{Appendix~\ref{sec:sampling_eval}}{the full version of our paper~\cite{connectit}.}

In \fullver{Appendix~\ref{apx:samplepf}}{the full version of the
paper}, we prove that \koutsample{}, \bfssample{}, and \lddsample{}
are all \samplecorrect{}, i.e., they all produce connectivity labeling
satisfying Definition~\ref{def:sampling_produces_trees}.

\begin{figure}[!t]
\begin{center}
\vspace{-1em}
\includegraphics[scale=0.44,trim={0 0.75cm 0 0.25cm}]{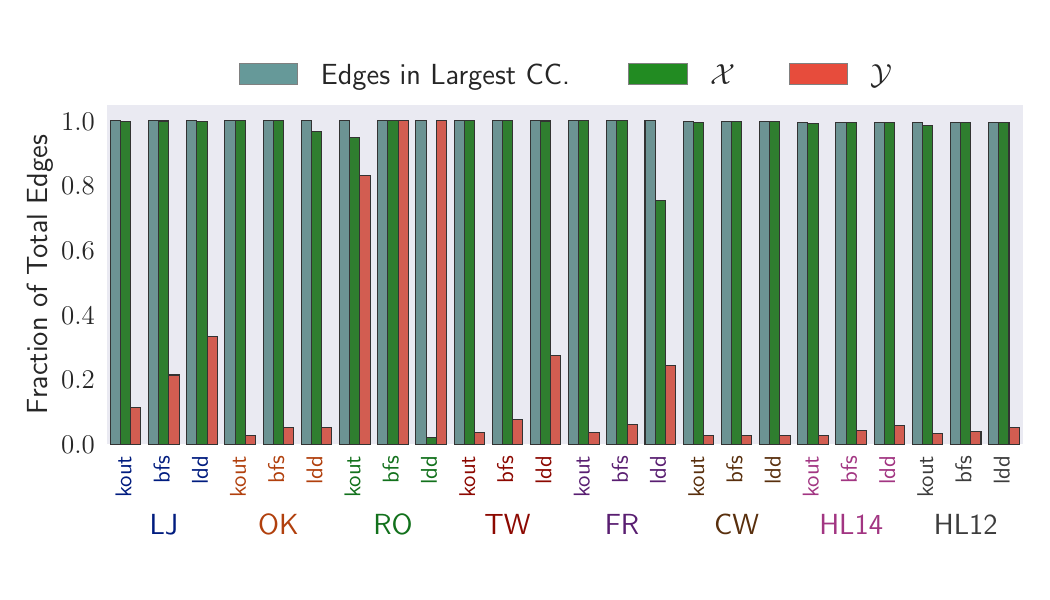}
\caption{\label{fig:sampling_efficiency}
\revised{
Bar plot showing the performance of different sampling strategies on
eight large real-world graphs in terms of the number of edges in the
largest connected component, the number of edges in the most frequent
sampled component ($\mathcal{X}$) and the number of edges processed by
the sampling strategy ($\mathcal{Y}$). All quantities are shown as a
fraction of the total number of edges, $m$.
}
}
\end{center}
\end{figure}

\revised{
\myparagraph{The Potential Benefits of Sampling}
Intuitively, the main advantage of our sampling schemes is that if
the graph has a single large component containing a significant
fraction of the edges, applying a sampling method can allow us to skip
processing most of these edges.
Specifically, suppose applying a \samplecorrect{} sampling scheme
uncovers a frequent component $L_{\max}$ containing $\mathcal{X}$
edges while processing only $\mathcal{Y}$ edges. Then, by skipping
processing vertices in $L_{\max}$ in the finish phase, the total
number of edges that are processed by the algorithm is $m -\mathcal X
+ \mathcal Y$.

Figure~\ref{fig:sampling_efficiency} illustrates the potential
benefits from applying our sampling schemes on a suite of large
real-world graphs (see Section~\ref{sec:eval} for graph details). We
observe that $\mathcal{X}$ is usually a large fraction of both the
number of edges in the largest component and of $m$, indicating that
the sampling methods can help us skip nearly all of the edges in the
graph in the finish phase. The one exception is the road\_usa (RO)
graph, where LDD and BFS Sampling both suffer due to the graph's large
diameter, which we discuss in more detail in
Section~\ref{subsec:static_cpu_nosample}. However, both LDD and BFS
Sampling perform well on the other low-diameter graphs since
direction-optimization enables them to complete while only examining a
small number of edges. Note that applying LDD on high-diameter graphs
results in a large number of small clusters, and thus this scheme is
better suited for low-diameter graphs~\cite{MillerPX2013}. We observe
that \koutsample{} results in a large value of $\mathcal{X}$ in all
cases. Finally, $\mathcal{Y}$ is typically significantly smaller than
$\mathcal{X}$ indicating that our approach enables us to process
$\mathcal{Y} + (m - \mathcal{X})$ edges in total, which is only a
small fraction of $m$.
}

\subsection{Finish Algorithms}\label{subsec:finish}
We now describe different \minbased{} methods that can be used as
\defn{finish} methods in our framework. All of the finish methods that
we describe can be combined with any of the sampling methods described
in Section~\ref{subsec:sampling}. We provide implementations of
several different algorithm classes, which internally have many
options that can be combined to generate different instantiations of
the algorithm. The \minbased{} algorithms that we consider as part of
the framework are \unionfind{} (many different variants, described
below), Shiloach-Vishkin~\cite{ShiloachV82}, Liu-Tarjan~\cite{LiuT19},
Stergiou~\cite{Stergiou2018}, and \labelprop{}. An important
feature of \framework is in modifying these finish methods to
\emph{avoid traversing} the vertices with the most frequently occurring
ID as identified from sampling in
Algorithm~\ref{alg:framework_connectivity}.  \fullver{We provide
pseudocode for  all of our implementations
 described in this section in Appendix~\ref{apx:framework}}
{We provide pseudocode for our
implementations of all of our implementations described in this section in
the full version of our paper~\cite{connectit}}.

\subsubsection{Union-Find}\label{sec:uf}
We consider several different concurrent (asynchronous) union-find
algorithms, which are all \minbased{}. All of these algorithms are
linearizably monotone for a set of concurrent union and find
operations, with the exception of the concurrent Rem's
algorithm variants using the splice rule, which are linearizable only
for a set of concurrent union operations or a set of concurrent find
operations, but not for mixed operations.  All of these algorithms can
be combined with sampling by simply skipping traversing the vertices
with label equal to $\largestcomp{}$ after sampling.

\myparagraph{Asynchronous Union-Find} The first
class of algorithms are inspired by a recent paper exploring
concurrent union-find implementations by Jayanti and
Tarjan~\cite{Jayanti2016}. We implement all of the variants from their
paper, as well as a full path compression technique (also considered
in~\cite{alistarh2019search}) which works better in practice in some
cases. We refer to this union-find algorithm as {$\bm{\unionasync{}}$}
since it is the classic union-find algorithm directly adapted for an
asynchronous shared-memory setting. The algorithm links from
lower-indexed to higher-indexed vertices to avoid cycles, and only
performs links on roots (thus implying that the algorithm is
monotone). This algorithm can be combined with the following
implementations of the find operation: $\bm{\mathsf{FindNaive}}$,
which performs no compression during the operation;
$\bm{\mathsf{FindSplit}}$ and $\bm{\mathsf{FindHalve}}$, which perform
path-splitting and path-halving, respectively; and
$\bm{\mathsf{FindCompress}}$, which fully compresses the find path.
Jayanti and Tarjan show that this class of algorithms is linearizable
for a set of concurrent union and find operations (they do not consider $\mathsf{FindCompress}$, but it is relatively easy to show that it is linearizable). The fact that these
algorithms are linearizably monotone for a set of concurrent union
and find operations follows from the observation that they only link
roots, and thus all label changes made by the operations are the
result of taking the union of trees.

We also consider two similar variants of the \unionasync{} algorithm:
{\unionhook{}} and {\unionearly{}}.  $\bm{\unionhook{}}$ is closely
related to \unionasync{}, with the only difference between the
algorithms being that instead of performing a CAS directly on the array
storing the connectivity labeling, we perform a CAS on an auxiliary
$\codevar{hooks}$ array, and perform an uncontended write on the
$\codevar{parents}$ array. $\bm{\unionearly{}}$ is also similar to
\unionasync{}, except that the algorithm traverses the paths from both vertices together, and tries to eagerly check and
hook a vertex once it is a root.
The algorithm can optionally perform a find on the
endpoints of the edge after the union operation finishes, which has
the effect of compressing the find path. The linearizability proof for
\unionasync{} by Jayanti and Tarjan~\cite{Jayanti2016} applies
 to \unionhook{} and \unionearly{}, and shows that both
algorithms are linearizable for a set of concurrent find and union operations.
Thus, these algorithms are also linearizably monotone for a set of
concurrent union and find operations as they only link roots.

\myparagraph{Randomized Two-Try Splitting}
Next, we incorporate a more sophisticated randomized algorithm by
Jayanti, Tarjan, and Boix{-}Adser{\`{a}}~\cite{JayantiTB19}, and
refer to this algorithm as $\bm{\jayanti{}}$.
The algorithm either performs finds naively, without using any path compression
($\bm{\findnaive{}}$), or uses a strategy called
$\bm{\mathsf{FindTwoTrySplit}}$, which guarantees provably-efficient
bounds for their algorithm, assuming a source of random bits. We
refer to~\cite{JayantiTB19}  for the pseudocode and proofs of correctness.
\fullver{In particular, they show that the algorithm has low total
expected work, and low cost per operation with high probability.
Since the \jayanti{} algorithm is linearizable, and only links roots,
the algorithm is also linearizably monotone.  Extending Theorem 4.1
in~\cite{JayantiTB19} to the work-depth setting, we have the following
corollary:
\vspace{-3pt}
\begin{corollary}
  The \jayanti{} algorithm
  solves connectivity in $O(m \cdot (\alpha(n, m/(np)) + \log (1 +
  np/m)))$ expected work and $O(\log n)$ depth with high probability.
\end{corollary}
\vspace{-3pt}}{}

\myparagraph{Concurrent Rem's Algorithm} We implement two concurrent
versions of Rem's algorithm (Rem's algorithm was first published in
Dijkstra's book~\cite{dijkstra1976discipline}): a lock-based version
by Patwary et al.~\cite{PatwaryRM12} ($\bm{\unionremlock{}}$) and a
lock-free compare-and-swap based implementation
($\bm{\unionremcas{}}$).  Our implementations of Rem's algorithm can
be combined with the same rules for path compression described for
$\mathsf{Union}$ above, with one exception which we discuss below.  In
addition to path compression strategies, our implementations of Rem's
algorithm take an extra \defn{splice} strategy, which is used when a
step of the union algorithm operates at a non-root vertex.
Specifically, our algorithms support the $\bm{\halveone{}}$,
$\bm{\splitone{}}$, and $\bm{\splice{}}$ rules.  The first two rules
perform a single path-halving or path-splitting.  The third rule
performs the splicing operation described in Rem's
algorithm, which atomically swaps the parent of the higher index
vertex in the path to the lower index vertex.
Combining the \findcompress{} option with the \splice{} rule
results in an incorrect algorithm, and so we exclude this single
combination.
All variants of Rem's algorithm that do not perform $\splice{}$ are
linearizable for a set of concurrent union and find operations by
simply following the proof of Jayanti and Tarjan~\cite{Jayanti2016}.
The implementations of Rem's algorithm combined with \splice{} do not satisfy
linearizability of both concurrent find and union operations.  In
\fullver{Appendix~\ref{apx:rem_correct}
(Theorem~\ref{thm:rem_correct})}{the full version of our paper}, we
show that the algorithm is correct in the \emph{phase-concurrent}
setting, where union and find operations are separated by a barrier.

We note that a recent paper by Alistarh et
al.~\cite{alistarh2019search} performed a careful performance
evaluation of concurrent union-find implementations for multicores on
much smaller graphs. They introduce a lock-free version of concurrent
Rem's algorithm with path-splitting similar to our implementation, but
do not prove that the algorithm is correct, or consider other
path-compaction methods in their algorithm.  Compared to their
lock-free Rem's implementation, our implementation is more general,
allowing the algorithm to be combined with path-halving and splicing
in addition to path-splitting.

\subsubsection{Other Min-Based Algorithms}
Lastly, we overview the other \minbased{} algorithms supported by
\framework{}.

\myparagraph{Liu-Tarjan's Algorithms}
Recently, Liu and Tarjan present a framework for simple concurrent
connectivity algorithms based on several rules that manipulate an
array of parent pointers using edges to transmit the connectivity
information~\cite{LiuT19}. These algorithms are not
true concurrent algorithms, but really parallel algorithms designed
for the synchronous Massively Parallel Computation (\mpc{}) setting.
We implement the framework proposed in their paper as part of
\framework{}.  Their framework ensures that the parent array is a
\emph{minimum labeling}, where each vertex's parent is the minimum
value among candidates that it has observed.

Conceptually, each round of an algorithm in the framework processes
all remaining edges and performs several rules on each edge. On each
round, each vertex observes a number of candidates and updates its
parent at the end of the round to the minimum of its current parent,
and all candidates. Each round performs a \emph{connect phase}, a
\emph{shortcut phase}, and possibly an \emph{alter phase}.
The connect phase updates the parents of edges based on
different operations, the shortcut phase performs  path
compression, and the optional  alter phase updates the
endpoints of an edge to be the current labels of its endpoints.
We provide details of the different options for implementing each phase
in \fullver{Appendix~\ref{apx:pseudocode}}{the full version of
our paper~\cite{connectit}}.

Liu and Tarjan prove that all of the algorithm combinations generated
from their framework are correct, but only analyze five particular
algorithms in terms of their parallel round complexity.  In addition
to the five original algorithms considered by Liu and Tarjan, we
consider a number of algorithm combinations that were not explored in
the original paper. We evaluate all algorithm
combinations that are expressible in the Liu-Tarjan framework, which
we list in \fullver{Appendix~\ref{apx:pseudocode}}{the full version of
our paper~\cite{connectit}}. Note that only the root-based
algorithms in the Liu-Tarjan framework are linearizably monotone.  The
remaining algorithms are not monotone since a non-root
vertex can be moved to a different subtree (one where the previous
tree and new trees are disconnected). The non-monotone algorithms
result in correct connectivity algorithms due to the fact that edges
which were previously applied continue to be applied in subsequent
rounds of the algorithm~\cite{LiuT19}.

\myparagraph{Stergiou et al.'s Algorithm}
Stergiou et al.~\cite{Stergiou2018} recently proposed a \minbased{}
connectivity algorithm for the massively parallel computation setting,
which is not monotone. We implement the algorithm as part of the
Liu-Tarjan framework, within which it can be viewed as a particular
instantiation of the Liu-Tarjan rules~\cite{LiuT19}.

\myparagraph{Shiloach-Vishkin and Label Propagation} \framework{} also
includes the classic Shiloach-Vishkin (\shiloachvishkin{}) algorithm,
which is linearizably monotone, and the folklore \labelprop{}
(\labelpropagation{}) algorithm, which is not monotone, both of which
we discuss in Appendix~\ref{sec:shiloachvishkin}.

\myparagraph{Sampling for Other Min-Based Algorithms}
Lastly, we describe how to combine the other \minbased{} algorithms
above which are \emph{not monotone} with the sampling algorithms in
\framework{}.
If the largest component after sampling, which has ID
$\largestcomp{}$, is relabeled such that all vertices in this
component have the smallest possible ID, then these vertices will
never change components, and we thus we never have to inspect edges
oriented out of these vertices. We show that this modification
produces correct algorithms by viewing the largest component
as a single contracted vertex that only preserves its inter-cluster
edges, and then applying the correctness proof for a connectivity
algorithm in the Liu-Tarjan framework~\cite{LiuT19}. We provide a
detailed proof in \fullver{Theorem~\ref{thm:min_based_correctness} in
the Appendix~\ref{apx:min_based_correctness}}{the full version of the
paper~\cite{connectit}}.

\subsection{Spanning Forest}\label{subsec:spanningforest}
We also extend \framework{} to generate a spanning forest of the
graph. We show that the class of root-based algorithms can be
converted in a black-box manner from parallel connectivity algorithms
to parallel spanning forest algorithms. This class consists of every
finish algorithm discussed in Section~\ref{subsec:finish}, with the
exception of the Liu-Tarjan algorithms that are not root-based, and
Stergiou's algorithm. We defer our description to
\fullver{Appendix~\ref{apx:spanning_forest}}{the full version of our
paper~\cite{connectit}}.

\subsection{Streaming}\label{subsec:streaming}
We now discuss how \framework{} supports streaming graph connectivity
in the parallel batch-incremental and wait-free asynchronous settings.
\revised{Formally, an algorithm in the
\defn{parallel batch-incremental} streaming setting receives a
sequence of \emph{batches} of operations, where each batch consists of
\defn{\textsc{Insert}$(u,v)$} operations and \defn{\textsc{IsConnected}$(u,v)$} queries.
The algorithm must process the batches one after the other, but
operations within a batch are not ordered, and so the streaming
algorithm can use parallelism to accelerate processing a batch. We
also consider the stronger asynchronous setting. Formally, in the
\defn{wait-free asynchronous} streaming setting, operations are not
presented to the algorithm as batches, but instead the \textsc{Insert}
and \textsc{IsConnected} operations can be called concurrently by a
set of asynchronous threads~\cite{HS}. Note that any wait-free
asynchronous algorithm can be easily extended to a parallel
batch-incremental algorithm by simply invoking the concurrent
implementation on every operation in parallel.
\framework{} supports the following types of algorithms in the
streaming setting:
\begin{enumerate}[label=(\textbf{\arabic*}),topsep=1pt,itemsep=0pt,parsep=0pt,leftmargin=15pt]
  \item The union-find algorithms in Section~\ref{sec:uf}, excluding
  Rem's algorithms with the \splice{} method. These algorithms are
  linearizably monotone in both the parallel batch-incremental
  and wait-free asynchronous settings. \label{lab:waitfreealgs}

  \item Shiloach-Vishkin (\shiloachvishkin{}) and the root-based
  \liutarjan{} algorithms.  These algorithms are linearizably monotone in the
  parallel batch-incremental setting, although we show that
  \textsc{IsConnected} queries can be applied concurrently in the
  wait-free asynchronous setting.
However, in these algorithms insertions cause edges to be processed multiple times until convergence, so they must be processed in batches.
  \label{lab:syncupdates}

  \item \unionremcas{} and \unionremlock{} using \splice{}. For this
  class of algorithms, we consider the \emph{phase-concurrent} setting
  which lies between the wait-free asynchronous and parallel
  batch-incremental settings. Essentially, in this setting, we have a
  synchronous barrier between insertion and query phases, but within a
  phase, operations can be called concurrently, like in the wait-free
  asynchronous setting. \label{lab:phaseconcurrentalgs}
\end{enumerate}
}

Due to space constraints, we provide details about our streaming
algorithm in \framework{} in
\fullver{Appendix~\ref{apx:framework}}{our full
paper~\cite{connectit}}.

\subsection{Implementation}\label{subsec:framework_implementation}
Our implementation of \framework is written in C++, and uses template
specialization to generate high-performance implementations while
ensuring that the framework code is high-level and general. Using
\framework, we can instantiate any of the supported connectivity
algorithm combinations using one line of code. Implementing a new
sampling algorithm is done by creating a new structure that implements
the sampling method for connectivity (and if applicable, a specialized
implementation for spanning forest).  The sampling code emits an array containing the partial connectivity
information. Additionally, for spanning forest the code emits a subset
of the spanning forest edges corresponding to the partial connectivity
information.  Implementing a new finish algorithm is done by
implementing a structure providing a
\textsc{FinishComponents} method.  If the finish algorithm supports
spanning forest, it also implements a \textsc{FinishForest} method.
Finally, if it supports streaming, the structure implements a
\textsc{ProcessBatch} method, taking a batch of updates, and returning
results for the queries in the batch. We note that our code can be
easily extended to the wait-free asynchronous setting.

We use the compression techniques provided by Ligra+~\cite{SDB2015,
DhBlSh18} to process the large graphs used in our experiments.
Storing the largest graph used in our experiments in the uncompressed
format would require well over 900GB of space to store the edges
alone. However, the graph requires only 330GB when encoded using byte
codes. The compression scheme uses difference-encoding for each
vertex's adjacency list, storing the differences using variable-length
byte codes. This compression scheme is supported by the Graph Based
Benchmark Suite (GBBS)~\cite{DhBlSh18, dhulipala20grades}, and we used
this code base as the basis for implementing \framework.

\section{Evaluation}\label{sec:eval}
\myparagraph{Overview}
We show the following results in this section:

\begin{itemize}[topsep=0pt,itemsep=0pt,parsep=0pt,leftmargin=8pt]
  \item Without sampling, the \unionremcas{} algorithm using the
  options
  $\{\splitone{}, \halveone{}\}$ is the fastest \framework{}
  algorithm across all graphs
  (Section~\ref{subsec:static_cpu_nosample}).
\item Performance analysis of the union-find variants
  (Section~\ref{subsec:union_find_eval}).
\item With sampling, the \unionremcas{} algorithm using the options $\{\splitone{}, \halveone{}\}$
  is consistently the fastest algorithm in \framework{}
  (Section~\ref{subsec:static_cpu_sample}).
\item The fastest \framework{}
  algorithms using sampling significantly outperform existing
  state-of-the-art results
  (Section~\ref{subsec:static_cpu_comparison}).

  \item \framework{} streaming algorithms achieve throughputs
    between 108M--7.16B
    directed edge insertions per second across all inputs.  Our
    algorithms achieve high throughput even at small batch sizes, and
    have consistent latency (Section~\ref{sec:streaming_cpu}).

  \item \framework{}'s streaming algorithms outperform the streaming
  algorithm from STINGER, an existing state-of-the-art graph streaming
  system, by between 1,461--28,364x (Section~\ref{sec:streaming_cpu}).

  \item \revised{An evaluation of \framework{}'s fastest algorithms on
  synthetic networks, and guidelines for selecting sampling and finish
  methods based on graph propertries (Section~\ref{sec:discussion}).}

\end{itemize}

\noindent We show the additional experimental results in \fullver{the appendix}{our full paper~\cite{connectit}}:
\begin{itemize}[topsep=0pt,itemsep=0pt,parsep=0pt,leftmargin=8pt]
  \item An analysis of the three sampling schemes considered in this
    paper, showing how these schemes behave in practice as a function
    of their parameters, and different implementation choices.
    \fullver{One of
    our main observations is that in practice, the number of
    inter-cluster edges left after \koutsample{} is significantly less
    than $n/k$ (Appendix~\ref{sec:sampling_eval}).}{}

  \item The fastest \framework{} algorithms with sampling are as fast as
    or faster than basic graph routines that perform one indirect
    read per edge.

  \item We compare \framework{}'s performance on very large graphs
    (Table~\ref{table:big_comparison}), to the performance of
    state-of-the-art external-memory, distributed-memory, and
    shared-memory systems.

\item The trends for spanning forest are similar to connectivity,
and on average the additional overhead needed to obtain the spanning
forest compared to connectivity is 23.7\%.
\end{itemize}

\setlength{\tabcolsep}{2pt}
\begin{table}[!t]
\footnotesize
  \centering
\begin{tabular}{l|r|r|r|r|r|c|c}
\toprule
{\bf Graph}  & {\bf $n$} & {\bf $m$} & {\bf Diam.} & {\bf Num C.} & {\bf Largest C.} & \revised{\bf LT-DC (s)} & \revised{\bf LT (s)}\\
        \midrule
{RO}       & 23.9M    & 57.7M                      & 6,809                     & 1            & 23.9M     & \revised{0.108}  &  0.241  \\ 
{LJ}       & 4.8M     & 85.7M                      & 16                        & 1,876        & 4.8M      & \revised{0.101}  &  0.226  \\ 
{CO}       & 3.1M     & 234.4M                     & 9                         & 1            & 3.1M      & \revised{0.094}  &  0.520  \\ 
{TW}       & 41.7M    & 2.4B                       & 23*                       & 1            &41.7M      & \revised{0.115}  &  2.80   \\ 
{FR}       & 65.6M    & 3.6B                       & 32                        & 1            &65.6M      & \revised{0.182}  &  6.07   \\ 
{CW}       & 978.4M   & 74.7B                      & 132*                      & 23.7M   &950.5M          & \revised{0.534}  &  54.2   \\ 
{HL14}     & 1.7B     & 124.1B                     & 207*                      & 129M  &1.57B             & \revised{1.02}  &  101.3  \\ 
{HL12}     & 3.6B     & 225.8B                     & 331*                      & 144M  &3.35B             & \revised{1.64}  &  192.5  \\ 
\end{tabular}
\caption[a]{Graph inputs, including vertices, directed edges, graph
diameter, the number of components (Num C.), the number of vertices in
the largest component (Largest C.), and the time required in seconds
to load the graph input in the binary CSR format used in GBBS,
assuming that the graph is already in the disk-cache (LT-DC). The
loading time (LT) is the time to fully read the input graph from a 1TB
\emph{Samsung-SM961} SSD. (The focus of this paper is not on
accelerating this hardware-dependent loading time, but on accelerating
connected components in the in-memory setting where the data is
already in the disk-cache.) We mark diameter values where we are
unable to calculate the exact diameter with * and report the effective
diameter observed during our experiments, which is a lower bound on
the actual diameter.  } \label{table:sizes}
\end{table}

\myparagraph{Experimental Setup} Our experiments are performed on a
72-core Dell PowerEdge R930 (with two-way hyper-threading) with
$4\times 2.4\mbox{GHz}$ Intel 18-core E7-8867 v4 Xeon processors (with
a 4800MHz bus and 45MB L3 cache) and 1\mbox{TB} of main memory. Our
programs use a work-stealing scheduler that we implemented. The
scheduler is implemented similarly to Cilk for
parallelism~\cite{blelloch2020parlay}. Our programs are compiled with
the \texttt{g++} compiler (version 7.3.0) with the \texttt{-O3} flag.
We use the command \texttt{numactl -i all} to balance the memory
allocations across the sockets. All of the numbers we report are based
on our parallel implementations on 72 cores with hyper-threading.

\myparagraph{Graph Data}
To show how \framework{} performs on graphs at different scales, we
selected a representative set of real-world graphs of varying sizes.
Most of our graphs are Web graphs and social networks---low-diameter
graphs that are frequently used in practice. To test our algorithms on
high-diameter graphs, we also ran our implementations on a road
network.

Table~\ref{table:sizes} lists the graphs we use. We used a collection
of graphs at different scales, including the largest publicly-available graphs.
\defn{road\_usa (RO)} is an undirected road network from the DIMACS
challenge~\cite{road-graph}.
\defn{LiveJournal (LJ)} is a directed graph of the LiveJournal social
network~\cite{boldi2004webgraph}. \defn{com-Orkut (CO)} is an undirected
graph of the Orkut social network. \defn{Twitter (TW)} is a directed graph
of the Twitter network~\cite{kwak2010twitter}. \defn{Friendster (FR)} is an
undirected graph describing friendships from a gaming network.
\defn{ClueWeb (CW)} is a directed Web graph from the Lemur project at
CMU~\cite{boldi2004webgraph}.  \defn{Hyperlink2012 (HL12)} and
\defn{Hyperlink2014 (HL14)} are directed hyperlink graphs obtained from the
WebDataCommons dataset where vertices represent Web
pages~\cite{meusel15hyperlink}.
We note that Hyperlink2012 is the
\emph{largest publicly-available real-world graph}.

Some of the inputs (such as the Hyperlink graphs) are originally
directed. Like previous work on connectivity for these
graphs~\cite{Stergiou2018}, we symmetrize them before applying our
algorithms and find that this results in a single massive
component for all graphs that we consider (the size of the largest
component is shown in Table~\ref{table:sizes}).

\newcommand{\STAB}[1]{\begin{tabular}{@{}c@{}}#1\end{tabular}}
\begin{table}[!t]
\footnotesize
\centering
\renewcommand{\tabcolsep}{0.5 mm}
\begin{tabular}[t]{@{}c |  l | c | c | c | c | c | c | c | c}
  \toprule
  \multicolumn{1}{c|}{\bf{Grp.}} &
  \multicolumn{1}{c|}{\bf Algorithm} &  \multicolumn{1}{c|}{\bf RO} &
  \multicolumn{1}{c|}{\bf LJ} & \multicolumn{1}{c|}{\bf CO} &
  \multicolumn{1}{c|}{\bf TW}  & \multicolumn{1}{c|}{\bf FR}  &
  \multicolumn{1}{c|}{\bf CW} & \multicolumn{1}{c|}{\bf HL14} &
  \multicolumn{1}{c}{\bf HL12}  \\
  \midrule
  \multirow{9}{*}{\STAB{\rotatebox[origin=c]{90}{No Sampling}}}
  & \unionearly{}         &3.61e-2                       & 3.48e-2                        & 8.63e-2                          &2.52                      & 1.50                       &59.8        &17.0       & 32.9 \\ 
  & \unionhook{}          &3.37e-2                       & 1.75e-2                        & 2.69e-2                          &0.390                     & 1.17                       &6.05        &9.37       & 20.0 \\ 
  & \unionasync{}         &4.02e-2                       & 2.03e-2                        & 3.12e-2                          &0.426                     & 1.21                       &7.92        &12.2       & 25.5 \\ 
  & \unionremcas{}        &{\color{mygreen} \bf 2.80e-2} & {\color{mygreen} 1.27e-2}      & {\color{mygreen} 1.91e-2}        &{\color{mygreen}0.316}    & {\color{mygreen}0.902}     &{\color{mygreen}4.04}  &{\color{mygreen}6.64} & {\color{mygreen}13.9} \\ 
  & \unionremlock{}       &5.07e-2                       & 1.95e-2                        & 2.84e-2                          &0.437                     & 1.23                       &5.64        &9.20       & 19.3 \\ 
  & \jayanti{}            &6.90e-2                       & 4.49e-2                        & 8.48e-2                          &0.965                     & 2.76                       &22.5        &36.4       & 72.1 \\ 
  & \liutarjan{}          &7.40e-2                       & 5.18e-2                        & 6.46e-2                          &2.78                      & 6.60                       &30.1        &67.1       & 142 \\ 
  & \shiloachvishkin{}    &0.138                         & 4.34e-2                        & 5.70e-2                          &1.65                      & 5.38                       &21.2        &38.5       & 106 \\ 
  & \labelpropagation{}   &13.4                          & 4.66e-2                        & 6.37e-2                          &1.24                      & 4.37                       &13.4        &20.7       & 46.5 \\ 

  \midrule
  \multirow{9}{*}{\STAB{\rotatebox[origin=c]{90}{\koutsample{}}}}
  &\unionearly{}         &{\color{mygreen} 3.25e-2}     & 9.00e-3                       & 8.61e-3                        &{\color{mygreen} 0.117}      &{\color{mygreen}0.227}      &2.28        &4.77       &8.94 \\ 
  &\unionhook{}          &3.62e-2                       & 9.18e-3                       & 9.16e-3                        &0.121                        &0.230                       &2.22        &3.63       &8.51 \\ 
  &\unionasync{}         &3.33e-2                       & 8.97e-3                       & {\color{mygreen} \bf 8.56e-3}  &{\color{mygreen} 0.117}      &0.228                       &2.21        &3.60       &8.49 \\ 
  &\unionremcas{}        &3.43e-2                       & {\color{mygreen} \bf 8.96e-3} & 8.62e-3                        &{\color{mygreen} 0.117}      &{\color{mygreen}0.227}      &{\color{mygreen}2.15}  &{\color{mygreen}3.51} & {\color{mygreen}\bf 8.20} \\ 
  &\unionremlock{}       &4.45e-2                       & 1.13e-2                       & 1.01e-2                        &0.138                        &0.344                       &2.63        &4.33       &9.91 \\ 
  &\jayanti{}            &3.89e-2                       & 9.77e-3                       & 8.80e-3                        &0.125                        &0.237                       &2.43        &4.05       &9.58 \\ 
  &\liutarjan{}          &6.34e-2                       & 9.90e-3                       & 9.18e-3                        &0.129                        &0.374                       &2.61        &6.74       &11.5 \\ 
  &\shiloachvishkin{}    &5.72e-2                       & 9.72e-3                       & 8.78e-2                        &0.124                        &0.237                       &2.70        &5.03       &12.5 \\ 
  &\labelpropagation{}   &12.6                          & 1.02e-2                       & 9.63e-3                        &0.121                        &0.375                       &2.44        &4.75      &9.68 \\ 

  \midrule
  \multirow{9}{*}{\STAB{\rotatebox[origin=c]{90}{\bfssample{}}}}
  &\unionearly{}         &2.69                         &1.07e-2                         &9.26e-3                         &9.42e-2                       &0.186                      &2.27        &4.02       &9.33 \\ 
  &\unionhook{}          &2.65                         &1.09e-2                         &9.71e-3                         &9.53e-2                       &0.186                      &2.29        &2.94       &9.40 \\ 
  &\unionasync{}         &2.69                         &1.08e-2                         &{\color{mygreen}9.12e-3}        &9.31e-2                       &0.189                      &2.23        &2.87       &9.23 \\ 
  &\unionremcas{}        &2.66                         &{\color{mygreen}1.06e-2}        &9.19e-3                         &{\color{mygreen} \bf 9.24e-2} &{\color{mygreen} \bf 0.183}&{\color{mygreen} 2.21}  &{\color{mygreen}\bf 2.83} & {\color{mygreen}9.11} \\ 
  &\unionremlock{}       &2.67                         &1.13e-2                         &1.07e-2                         &0.113                         &0.219                      &2.69        &3.68       &10.8 \\ 
  &\jayanti{}            &2.75                         &1.14e-2                         &9.52e-3                         &9.80e-2                       &0.195                      &2.38        &3.22       &9.88 \\ 
  &\liutarjan{}          &2.68                         &1.17e-2                         &9.80e-3                         &9.61e-2                       &0.383                      &2.85        &7.61       &13.4 \\ 
  &\shiloachvishkin{}    &{\color{mygreen}2.54}        &1.12e-2                         &9.72e-3                         &9.87e-2                       &0.196                      &2.59        &4.13       &12.2 \\ 
  &\labelpropagation{}   &2.58                         &1.19e-2                         &1.03e-2                         &9.47e-2                       &0.446                      &2.31        &3.21       &9.91 \\ 

  \midrule
  \multirow{9}{*}{\STAB{\rotatebox[origin=c]{90}{\lddsample{}}}}
  &\unionearly{}         &0.117                        &1.32e-2                        &8.63e-3                         &0.124                        &{\color{mygreen}0.193}       &1.74     &4.63    &8.52    \\ 
  &\unionhook{}          &0.112                        &1.33e-2                        &8.81e-3                         &0.127                        &0.197                        &1.75     &3.58    &8.46    \\ 
  &\unionasync{}         &0.103                        &1.32e-2                        &8.49e-3                         &0.123                        &{\color{mygreen}0.193}       &1.71     &3.48    &8.31    \\ 
  &\unionremcas{}        &{\color{mygreen}9.86e-2}     &{\color{mygreen} 1.29e-2}      &{\color{mygreen}8.48e-3}        &{\color{mygreen}0.122}       &{\color{mygreen}0.193}       &{\color{mygreen}\bf 1.69}     &{\color{mygreen}3.46}    &{\color{mygreen}8.28}    \\ 
  &\unionremlock{}       &0.126                        &1.54e-2                        &1.03e-2                         &0.144                        &0.226                        &2.16     &4.31    &9.97    \\ 
  &\jayanti{}            &0.148                        &1.35e-2                        &8.98e-3                         &0.131                        &0.202                        &1.85     &3.84    &9.13    \\ 
  &\liutarjan{}          &0.178                        &1.45e-2                        &8.73e-3                         &0.130                        &1.24                         &2.32     &8.33    &12.5    \\ 
  &\shiloachvishkin{}    &0.250                        &1.36e-2                        &8.81e-3                         &0.131                        &0.197                        &2.07     &4.70    &11.2    \\ 
  &\labelpropagation{}   &14.3                         &1.41e-2                        &8.99e-3                         &0.127                        &2.03                         &1.76     &3.79    &9.06    \\ 

  \midrule
  \multirow{7}{*}{\STAB{\rotatebox[origin=c]{90}{Other Systems}}}

  &\bfscc{}~\cite{ShunB14}                     &2.60             &1.94e-2          &1.05e-2        &0.169          &1.34 & 5.56        &61.6       & 62.5\\
  &\gbbscc{}~\cite{SDB14}                      &0.41             &0.247            &2.78e-2        &0.587          &2.18 & 5.97        &11.4       & 25.8\\
  &MultiStep~\cite{Slota14}                    &29.6             &0.272            &0.138          &--- &1.76  & --- & --- & --- \\
  &Galois~\cite{Nguyen2014}                    &6.10e-2          &2.55e-2          &3.40e-2        &1.167          &1.77  & --- & --- & --- \\
  &PatwaryRM~\cite{PatwaryRM12}                &6.81e-2          &3.65e-2          &3.93e-2        &0.428          &1.15  & --- & --- & --- \\
  &GAP-SV~\cite{BeamerAP15}  &0.103            &0.134            &0.150          &5.669          &7.01  & --- & --- & --- \\
  &GAP-AF~\cite{Sutton2018}          &4.29e-2          &5.30e-2          &7.32e-2        &0.172          &0.306 & --- & --- & --- \\

\end{tabular}
\caption{\small Running times (seconds) of \framework{} algorithms and
  state-of-the-art static connectivity algorithms on a
  72-core machine (with 2-way hyper-threading enabled). We report
  running times for the \emph{No Sampling} setting, as well as the
  \emph{\koutsample{}}, \emph{\bfssample{}}, and \emph{\lddsample{}}
  schemes considered in this paper. Within each group, we display the
  running time of the fastest \framework{} variant in green.
  Additionally, for each graph we display the fastest running time
achieved by any algorithm combination in bold. We mark entries that
did not successfully solve the problem with ---. }
  \label{table:static_cpu_all}
\end{table}

\subsection{Performance Without Sampling}\label{subsec:static_cpu_nosample}
We start by studying the performance of different \framework{} finish
methods without sampling, since trends observed in this setting hold
when sampling is applied.  Table~\ref{table:static_cpu_all} shows the
results of the fastest \framework{} implementations across our inputs.
For algorithms that have many options, such as union-find, we report
the fastest time out of all combinations of options.

Among all of our algorithms, we observe that \unionremcas{} is
consistently the fastest finish algorithm. We find that the fastest
variant of \unionremcas{} across all graphs uses the \findnaive{}
option for find (i.e., it does not perform any extra compression
after the union operation), and most frequently uses the \splitone{}
option to perform path compression during the union operation\fullver{(see
Algorithm~\ref{alg:union_rem_cas}).}{.} We observe that the running
times of the other splice options, \splice{} and \halveone{}, are
almost identical to \splitone{}.

The \unionhook{} algorithm also achieves high performance, and is
between 1.20--1.49x slower than \unionremcas{} across all graphs
(1.35x slower on average).  The \unionasync{} algorithm also achieves
consistent high performance, and is between 1.34--1.96x slower than
\unionremcas{} across all graphs (1.62x slower on average).
\jayanti{} is consistently slower than \unionremcas{}, being between
2.46--5.56x slower across all graphs, and 4.09x slower on average. The
fastest find option for \jayanti{} in this setting was always the
two-try splitting option, with the exception of the road\_usa graph
where \findnaive{} was slightly faster. We discuss the source of
performance differences between these algorithms in
Section~\ref{subsec:union_find_eval}.

Compared to the \unionfind{} algorithms, the \liutarjan{} algorithms are
much slower on our inputs. \fullver{The fastest variants in this
  setting were one of $\{\mathsf{EF}, \mathsf{PRF}, \mathsf{PR},
  \mathsf{CRFA}\}$ (we define these \liutarjan{} variants in
  Appendix~\ref{apx:pseudocode}).}{} The fastest \liutarjan{} variant
is still 2.64--10.2x slower than \unionremcas{} (6.74x slower on
average). Stergiou's algorithm was always slower than the fastest
variant from the \liutarjan{} framework.  Finally, our implementation
of \shiloachvishkin{} is between 2.98--7.62x slower than
\unionremcas{} (5.14x slower on average). The performance of
\labelpropagation{} is between 3.11--4.92x slower on all graphs except
for road\_usa. On road\_usa, its performance is 478x worse than that
of \unionremcas{} because it requires a large number of rounds where
most vertices are active due to the high diameter of the graph.

\emph{Takeaways.} Without sampling, the \unionremcas{} algorithm using
\splitone{}, with no additional path-compression option is a robust
algorithm choice, and consistently performs the best across all graphs
compared with other \framework{} implementations.

\begin{figure}[!t]
  \centering
  \hspace{-1.5em}
    \includegraphics[trim=0em 0em 0em 6em, width=0.435\textwidth]{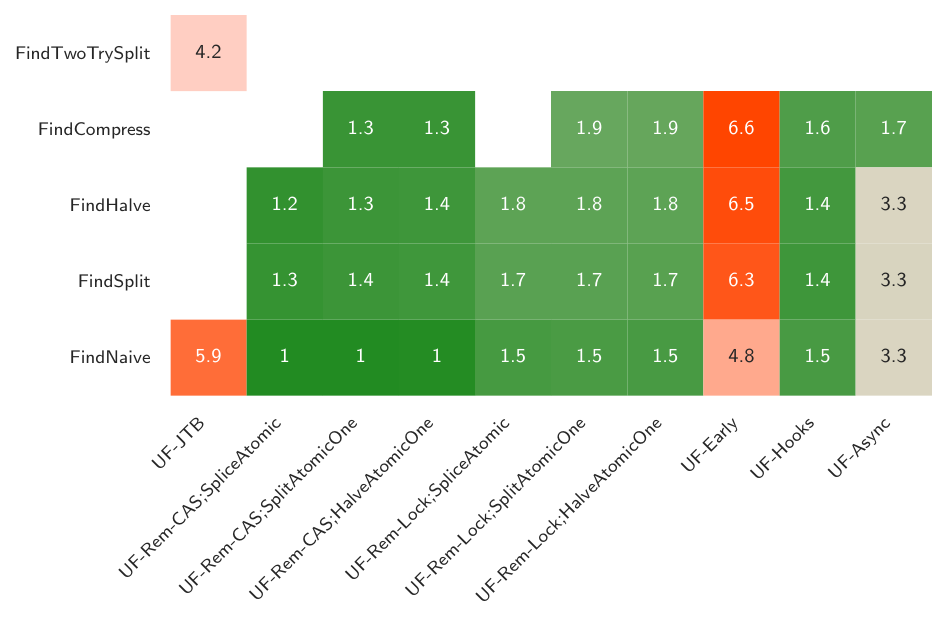}\\ 
    \caption{\small Relative performance of different \unionfind{}
    implementations on graphs used in our evaluation of
    \framework{} in the \emph{No Sampling} setting. The numbers are
    slowdowns relative to the fastest implementation.
  }\label{fig:heatmap_uf_nosample}
  \vspace{-0.5em}
\end{figure}

\subsubsection{Union-Find Evaluation}\label{subsec:union_find_eval}
Figure~\ref{fig:heatmap_uf_nosample} shows the relative performance of
different \unionfind{} variants from \framework{} in the \emph{No
Sampling} setting, averaged across all graphs.  \unionremcas{}
implementations achieve consistently high performance in this setting
using either \splitone{}, \halveone{}, or \splice{} to perform
compression.  The \unionremlock{} implementation is 1.5--1.9x
slower than \unionremcas{} for all 6 variants of this algorithm.
\unionearly{}, \unionhook{}, and \unionasync{} are all slower on
average: between 4.8--6.6x, 1.4--1.6x, and 1.7--3.3x on average,
respectively.  Finally, \jayanti{} is much slower, although the
\findtwotry{} option only incurs a more modest slowdown of
4.2x on average.

\myparagraph{Performance Analysis}
We annotated our \unionfind{} algorithms to measure the \defn{Max Path
Length (MPL)}, or the longest path length experienced by the
\unionfind{} algorithm during the execution of any \textsc{Union}
operation, and the \defn{Total Path Length (TPL)}, which is the
sum of all path lengths observed during all
\textsc{Union} executions. We also measured the number of LLC misses,
and the total amount of bytes transferred to the memory controller
(both reads and writes). \fullver{We note that adding this instrumentation
affects the overall running times between 10--20\%.}{} We report the
detailed results of the analysis in
\fullver{Appendix~\ref{apx:cpu_no_sample}.}{the full
paper~\cite{connectit}.}  We find that the TPL is relevant for
predicting the running time---the TPL has a Pearson correlation
coefficient of 0.738 with running time (the MPL has a weaker
coefficient of 0.344).

\emph{Takeaways.} Based on this discussion, we conclude that
minimizing both the total amount of data written to and read from the
memory controller, and improving the locality of these accesses,
thereby reducing the number of LLC misses, is critical for high
performance. One way of achieving these objectives is by minimizing
the TPL, although optimizing to minimize the
TPL does not by itself guarantee the fastest
performance.

\subsection{Performance With Sampling}\label{subsec:static_cpu_sample}
We defer a detailed analysis of our performance under different
sampling schemes to \fullver{Appendix~\ref{apx:cpu_sample}}{the full paper~\cite{connectit}}, and list our main
findings here. Our results for combining
\framework{} algorithms with different sampling schemes suggest the
following takeaways when selecting algorithms and sampling schemes for
different graphs:
\begin{itemize}[topsep=0pt,itemsep=0pt,parsep=0pt,leftmargin=8pt]
  \item For low-diameter graphs, such as social networks and
    Web graphs, the fastest finish algorithms combined with all three
    sampling schemes explored in this paper provide significant
    speedups over the fastest \framework{} algorithms that do not use
    sampling, ranging from 1.4--3.9x speedup using \koutsample{} (2.2x
    on average), 1.1--4.9x speedup using \bfssample{} (2.4x on
    average), and 0.98--4.6x speedup for \lddsample{} (2.3x on
    average).

  \item If the graph has high diameter, then using \koutsample{} with any
    union-find method seems to be the best fit, unless using
    \labelprop{}, in which case \bfssample{} is the best choice.

  \item For large Web graphs, like the ClueWeb and the two Hyperlink
    graphs studied in this paper, all three sampling schemes result in
    very high speedups over the fastest unsampled algorithms in
    \framework{}, ranging from 1.7--1.9x for \koutsample{},
    1.5--2.3x for \bfssample{}, and 1.7--2.4x for \lddsample{}.
    Furthermore, the fastest \framework{} algorithm for these graphs
    is obtained by combining one of these sampling schemes with the
    \unionremcas{} algorithm (each scheme is fastest on one graph).
\end{itemize}

\subsection{Comparison with State-Of-The-Art}\label{subsec:static_cpu_comparison}
Table~\ref{table:static_cpu_all} reports the performance of these
systems on the inputs used in our evaluation.  Aside from \bfscc{} and
\gbbscc{} (defined below), which we implemented as part of our system, we were unable
to run the other systems on the large inputs due to the fact that
these systems do not support compression, which is required for
systems to compactly store and process our largest graph inputs on our
machine without using an exorbitant amount of memory.
We list the main takeaways of the experimental results here, and
present detailed results of the comparison in
\fullver{Appendix~\ref{apx:static_cpu_comparison}}{the full paper~\cite{connectit}}.

\begin{itemize}[topsep=0pt,itemsep=0pt,parsep=0pt,leftmargin=8pt]
  \item {\bf \bfscc{}}: The BFS-based connectivity implementation available from
    Ligra~\cite{ShunB14}. This algorithm computes each connected
    component by running a parallel BFS from the vertex with the lowest ID
    within it. We find that using sampling, \framework{} is
    between 1.22--80.0x faster than \bfscc{}, and 15.6x faster on
    average.

  \item {\bf \gbbscc{}}: The work-efficient connectivity
    implementation by Shun et al.~\cite{SDB14}, which recursively
    computes LDD (publicly available as part of
    GBBS~\cite{DhBlSh18}). Our implementations without sampling are
    between 1.45--19.4x faster (5.6x faster on average). Our
    implementations with sampling are between 3.1--27.5x faster (9.03x
    faster on average). This algorithm previously held the record
    time for connectivity on the Hyperlink2012 graph on any
    system, running in 25.8 seconds on a 72-core machine. The fastest
    \framework{} algorithm, \unionremcas{} with \splitone{}, is 3.14x
    faster on the same machine, and thus breaks this record.

  \item {\bf Multistep}: The hybrid BFS/\labelprop{} method by
    Slota et al.~\cite{Slota14}. Our codes without sampling are
    between 1.95--21.4x faster, and 10.1x faster on average. Using
    sampling, our codes are between 9.6--30.3x faster, and 18.6x
    faster on average. Our codes are much faster on large-diameter
    networks (over 1,000x faster).

  \item {\bf Galois}: Galois is a state-of-the-art shared-memory
    parallel programming library~\cite{Nguyen2013}.
    We found that their label propagation algorithm is consistently
    the fastest implementation in their code base, and report running
    times for this implementation for all but one graph (road\_usa,
    discussed below).

    Our codes without sampling are between 1.78--3.69x faster than
    theirs (2.32x faster on average), and our codes using sampling are
    between 2.17--12.3x faster than theirs (6.21x faster on average).
    Our codes are 2.43x faster than their default (EdgetiledAsync)
    algorithm on the road\_usa graph, which was their fastest
    algorithm for this graph.

    \item {\bf PatwaryRM}: We compare with the multicore Rem's algorithm by
    Patwary et al.~\cite{PatwaryRM12}. Our fastest
    implementations without sampling achieve between 1.27--2.87x
    speedup over their implementation (1.99x speedup on average), and
    our fastest implementations with sampling achieve between
    2.09--6.28x speedup over their implementation (4.31x speedup on
    average).

  \item {\bf GAPBS}: We compared with the GAP Benchmark
    Suite, a state-of-the-art shared-memory graph processing
    benchmark~\cite{BeamerAP15}, which implements the Afforest
    algorithm~\cite{Sutton2018}. Our fastest implementations without
    sampling are between 0.33--4.17x faster than their
    Afforest~\cite{Sutton2018} implementation (2.08x faster on
    average). Our fastest implementations with sampling are between
    1.32--8.55x faster than their Afforest implementation (3.9x faster
    on average).

\end{itemize}

\begin{table}[!t]
\footnotesize
\renewcommand{\tabcolsep}{0.4 mm}
\scalebox{0.9}{
\begin{tabular}[t]{  l | c | c | c | c | c | c | c | c | c | c}
  \toprule
  {\bf Algorithm} &  \multicolumn{1}{c|}{\bf RO} & \multicolumn{1}{c|}{\bf LJ} & \multicolumn{1}{c|}{\bf CO} & \multicolumn{1}{c|}{\bf TW}  & \multicolumn{1}{c|}{\bf FR} & \multicolumn{1}{c|}{\bf RM} & \multicolumn{1}{c|}{\bf BA} & \multicolumn{1}{c|}{\bf CW} & \multicolumn{1}{c|}{\bf HL14} & \multicolumn{1}{c}{\bf HL12} \\
  \midrule
  \unionearly{}         &1.48e9  &9.23e8  & 1.38e9 & 4.31e8 & 1.05e9 &3.49e8 &5.16e8   &4.00e8   &3.15e9   &2.80e9  \\
  \unionhook{}          &3.12e9  &4.21e9  & 5.94e9 & 2.79e9 & 1.49e9 &7.27e8 &1.18e9   &4.69e9   &5.17e9   &4.48e9  \\
  \unionasync{}         &3.49e9  &3.36e9  & 5.29e9 & 2.73e9 & 1.41e9 &8.05e8 &1.13e9   &4.86e9   &5.92e9   &4.69e9  \\
  \unionremcas{}        &{\bf \color{mygreen}3.98e9}  &{\bf \color{mygreen}5.28e9}  &{\bf \color{mygreen}7.16e9} &{\bf \color{mygreen} 3.85e9} & {\bf \color{mygreen}2.01e9} &{\bf \color{mygreen}8.78e8} &{\bf \color{mygreen}1.46e9}&{\bf \color{mygreen}5.73e9}   &{\bf \color{mygreen}6.64e9}   &{\bf \color{mygreen}5.64e9}  \\
  \unionremlock{}       &1.56e9  &3.68e9  & 5.95e9 & 3.36e9 & 1.74e9 &7.67e8 &1.42e9  &3.56e9   &2.99e9   &3.21e9 \\
  \jayanti{}            &1.15e9  &1.06e9  & 2.68e9 & 1.42e9 & 7.33e8 &2.88e8 &5.27e8  &2.15e9   &2.26e9   &1.79e9 \\
  \liutarjan{}          &2.87e8  &4.31e8  & 5.98e8 & 3.77e8 & 1.84e8 &1.11e8 &1.98e8  &3.02e8   &2.80e8   &2.62e8 \\
  \shiloachvishkin{}    &1.79e8  &4.56e8  & 1.13e9 & 2.89e8 & 1.76e8 &1.06e8, &2.43e8 &3.34e8   &2.65e8   &2.24e8  \\
\end{tabular}
}
\caption{\small Maximum parallel streaming throughput (directed edge insertions
per second) achieved by \framework{} streaming algorithms on our graph
inputs, and two synthetic graph inputs from the RMAT (RM) and
Barabasi-Albert (BA) families.  Due to memory constraints, we were
unable to materialize a COO representation of our three largest
graphs, and instead, sample 10\% of the edges to use in the batch.
Otherwise, the entire graph is used as part of a single batch, which
is not permuted. Note that there are no queries in the batch. For each
graph we display the algorithm with the highest throughput in bold.
}
\label{table:streaming_cpu_all}
\end{table}

\subsection{Streaming Parallel Graph Connectivity}\label{sec:streaming_cpu}

\myparagraph{Experiment Design}
We run two types of streaming experiments.  The first type generates a
stream of edge insertions by sampling a fraction of the edges
($f_{u}$) from a static input graph to use as insertions. Unless
otherwise mentioned, we use all edges as insertions ($f_{u} = 1$). The
second type uses synthetic graph generators to sample edge insertions.
We consider the RMAT and Barabasi-Albert (BA) graph
generators~\cite{bader2006gtgraph,barabasi1999emergence} in these
experiments. We generate RMAT graphs using the parameters $(a,b,c) =
(0.5, 0.1, 0.1)$.  In both the RMAT and BA graphs, the
number of vertices  is $n=2^{30}$ and the number of
edges is $10n$. For both graphs,
the batches used by our streaming algorithms are represented in the
COO format. For the ClueWeb, Hyperlink2014, and Hyperlink2012 graphs,
we are unable to represent the entire graph in COO, and sample 10\% of
the edges to use as insertions.

\subsubsection{Streaming Throughput}
\ We first consider the throughput achieved by each algorithm family in
the setting where only insertions are applied.
Table~\ref{table:streaming_cpu_all} reports the streaming throughput
achieved by the fastest variant of each algorithm on each graph. Note
that no sampling is applied in these streaming experiments. The
insertions are streamed as part of a single, large batch, which is
applied by the algorithm in parallel for Type~\ref{lab:waitfreealgs}
and \ref{lab:syncupdates}, and separately applied for insertions
first, then queries second, for Type~\ref{lab:phaseconcurrentalgs} algorithms
(see the discussion on these types in Section~\ref{subsec:streaming}).

The performance of most of our union-find algorithms is consistently
high on these graphs, in particular the performance of \unionhook{},
\unionasync{}, \unionremcas{}, and \unionremlock{}.  As in the static
setting, the \unionremcas{} algorithm consistently performs the best
across all input graphs, achieving a maximum throughput of over 7 billion
edge insertions per second on the com-Orkut graph. The other two union-find
algorithms, \unionearly{} and \jayanti{}, achieve somewhat
inconsistent performance, which is consistent with our findings from
Section~\ref{subsec:static_cpu_nosample}.
\fullver{
For \unionearly{}, the version of the algorithm with no extra
compression performed best compared with the versions that perform
additional path compression algorithms, as in the static case in
Figure~\ref{fig:heatmap_uf_nosample} (comparing \findnaive{}
performance to that of non-trivial path compression options).
The performance of \shiloachvishkin{} seems to depend mostly on the
number of rounds required (more rounds on larger diameter networks
like road\_usa). We note that the algorithm requires just two rounds
on com-Orkut. The performance of the fastest \liutarjan{} algorithm is
similar to that of the \shiloachvishkin{} algorithm in the streaming
setting across all graphs. Both of these methods are significantly
slower than the fastest union-find methods. We note that the fastest
\liutarjan{} algorithm in this setting is consistently the version
using \connect{}, \rootupdate{}, \fullshortcut{}, and \alter{} (the
$\mathsf{CRFA}$ algorithm in Appendix~\ref{apx:pseudocode}).}{
We discuss the performance of the remaining algorithms, which are even
slower than  \unionearly{} and \jayanti{}, in the full version of our
paper~\cite{connectit}.
}

\begin{table}[!t]
\footnotesize
\centering
\tabcolsep=0.12cm
\hspace*{-0.2cm}
\scalebox{0.9}{
\begin{tabular}[t]{c | c|c|c|c}
  \toprule
  {\bf Batch Size} & \textbf{\stinger{}} & \textbf{Updates/sec} &
  \textbf{\framework{}} & \textbf{Updates/sec} \\
  \midrule
  10             & 6.07e-2 & 164     & 2.14e-6  & 4.67M \\ 
  $10^{2}$       & 9.87e-2 & 1013     & 1.19e-5  & 8.40M \\
  $10^{3}$       & 0.171   & 5,847   & 2.19e-5  & 45.6M \\
  $10^{4}$       & 0.137   & 72,992  & 5.19e-5  & 192M  \\
  $10^{5}$       & 0.503   & 198,807  & 3.25e-4  & 307M  \\
  $10^{6}$       & 3.99    & 250,626 & 2.73e-3  & 366M  \\
  $2\cdot10^{6}$ & 6.52    & 306,748 & 4.313e-3 & 463M  \\
  \bottomrule
\end{tabular}
}
\caption{\small Running times (seconds) and edge insertion rates (directed edges/second) for
  \stinger{}'s dynamic connected components algorithm and
  \framework{}'s \unionremcas{} algorithm (with \splitone{}) when
  performing batch edge insertions on an empty graph with varying
  batch sizes. Inserted edges are sampled from the RMAT graph
  generator.}
\label{table:stinger_vs_connectit}
\end{table}

\subsubsection{Streaming Comparison with \stinger{}}
\stinger{} is based on a dynamic data structure for representing
dynamic graphs in the streaming setting~\cite{Ediger12}.
\stinger{} supports a dynamic connected components algorithm
by McColl et al.~\cite{McColl13}.  Their algorithm also supports edge
deletions, which makes it more general than our algorithms.
\fullver{
Since our codes do not update the graph structure, for the sake of
comparing \stinger{} with our codes, we only report the time required
for \stinger{} to update its connectivity labeling, which is strictly
less than the overall update time of \stinger{}.}{}
As far as we know, the only existing parallel algorithm designed for
incremental connectivity is by Simsiri et al.~\cite{Simsiri2017}, but
unfortunately we were unable to obtain the code from the authors.

\myparagraph{Comparison}
The \stinger{} code takes a parameter which trades off space usage
with the amount of re-computation that has to be done upon an update
(edge insertion or deletion). We set it to the lowest possible value,
which gave the best performance.
\fullver{
We note that the \stinger{} dynamic connectivity algorithm has an
unusually long initialization period, which depends on the number of
vertices.}{}
We were unable to initialize the \stinger{} dynamic connectivity
algorithm within several hours for graphs with more than 1 million
vertices, and were only able to evaluate batches of size up to 2
million due to limitations in their system.  Based on this, we opted
to generate batches from an RMAT graph generator using $2^{20}$
vertices. Our experiment inserts batches of varying sizes and measures
the time required by the dynamic connectivity algorithm to process the
batch. The times we report ignore the time taken by \stinger{} to
update adjacency information.

Table~\ref{table:stinger_vs_connectit} reports the results of the
experiments for \stinger{} and for \framework{}'s
\unionremcas{} algorithm with the \splitone{} option. Both implementations are
run on the same machine using all cores with hyper-threading.
\framework{} significantly outperforms \stinger{} in this setting,
achieving between \emph{1,461--28,364x} speedup across different batch
sizes. Surprisingly, \framework{} applied with a batch size of 10
achieves significantly higher throughput than \stinger{} for a batch
size of 2M. We realize that our comparison is somewhat unfair toward
\stinger{}, since although this implementation is one of the fastest
batch-incremental connectivity implementations currently publicly
available, it is designed for both edge insertions and deletions, and must perform extra
work in anticipation of edge deletions, which our algorithms do not
handle.

\revised{
\subsection{Algorithm Selection in \framework{}}\label{sec:discussion}

Given a graph, which combination of \framework{}'s sample and
finish methods should one apply to obtain the best performance?

\begin{figure}[!t]
\centering
\vspace{-16pt}
\begin{subfigure}[b]{0.49\columnwidth}
\includegraphics[width=\textwidth]{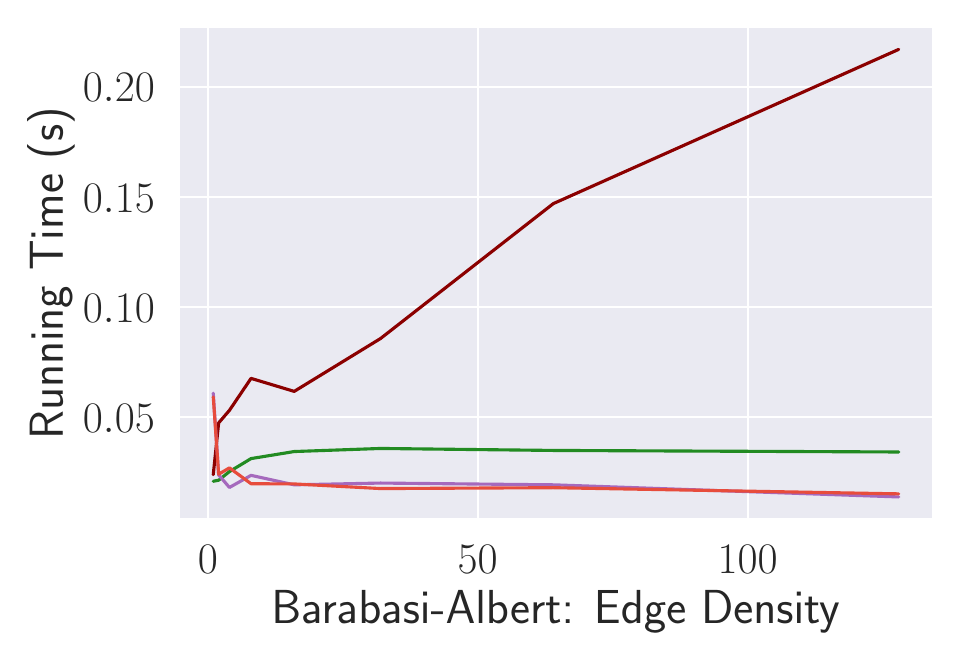}
\caption{
\label{fig:ba_sampling}
}
\end{subfigure}
\hfill
\begin{subfigure}[b]{0.49\columnwidth}
\includegraphics[width=\textwidth]{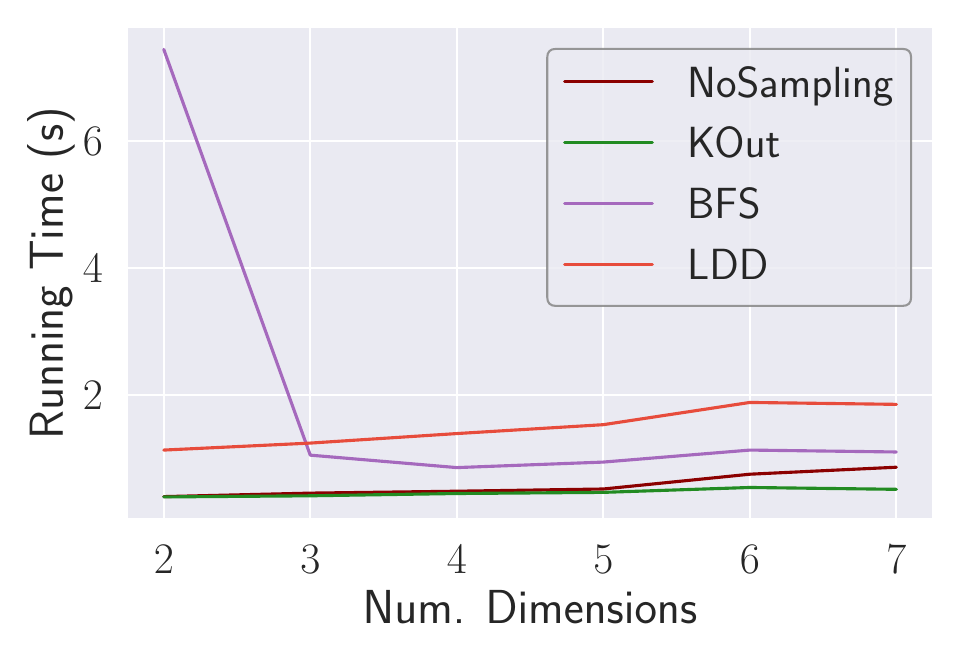}
\caption{
\label{fig:kdgrid_sampling}
}
\end{subfigure}\\
\caption{\label{fig:discussion_sample}
Running times (seconds) of the \unionremcas{} algorithm with
\splitone{} and \findnaive{} using different sampling methods on two
synthetic graph families.
Subfigure~\ref{fig:ba_sampling} displays results for graphs drawn from
the Barabasi-Albert generator where the $x$-axis specifies the edge
density, or the number of edges drawn for each newly added vertex.
Subfigure~\ref{fig:kdgrid_sampling} shows results for $d$-dimensional
torii where the $x$-axis shows the dimension.
}
\end{figure}

\myparagraph{Evaluation on Synthetic Networks}
To provide guidance, we evaluated \unionremcas{} with \splitone{} and
\findnaive{}, which are consistently  \framework{}'s fastest
connectivity methods, on two synthetic graph families with
significantly different properties. Figure~\ref{fig:ba_sampling}
displays results for graphs drawn from the Barabasi-Albert (BA)
generator on $10^{7}$ vertices, varying the edge density of the
graphs in powers of two from $1$ to $128$.
Figure~\ref{fig:kdgrid_sampling} displays results for $d$-dimensional
torii on $10^{9}$ vertices where we vary the dimension of the grid
(each vertex is connected to its $2d$ adjacent neighbors).

The diameter of the BA graphs decreases with increasing density, and
as a result both BFS and LDD Sampling  achieve good performance as the density
increases. Both schemes use the direction-optimization technique used when traversing
frontiers~\cite{Beamer12, ShunB2013}, which is beneficial in high-density graphs.
For the sparse (high-diameter) case where $n =
m$,
we found that \koutsample{} and not using sampling were the
fastest. 

For $d$-dimensional torii, \koutsample{} consistently performs the
best across the range of $d$ that we evaluate. For small $d$, there is
little difference between \koutsample{} and no sampling, but for
larger $d$, and thus higher average degree, \koutsample{} begins to
achieve significant speedups over no sampling. The performance of BFS
Sampling is poor on this graph family since the diameter of the
$d$-dimensional torus is $O(n^{1/d})$, causing BFS to perform many
rounds. The performance of LDD Sampling is also poor, since the
induced clustering consists of many small clusters and thus most of
the vertices must be processed in the finish phase.

\myparagraph{Picking an Algorithm}
Based on both our evaluation on the synthetic graph families above, and
the results for a broad collection of large real-world graphs shown in
Table~\ref{table:static_cpu_all}, we devised a decision tree that can
help users select an appropriate algorithm, which we show in
Figure~\ref{fig:decision_tree}. We recommend using the \unionremcas{}
algorithm, with any of the three possible splice strategies in
combination with the \findnaive{} strategy. Based on our evaluation of
union-find algorithms in Section~\ref{subsec:union_find_eval}, this
family of algorithms is consistently the fastest. Regarding sampling, for extremely sparse networks with low average degree such as road\_usa $(m/n < 3)$, not
using sampling can be beneficial. The reason is that the cost of
simply going over this small set of edges twice during two-phase
execution outweighs the additional benefit provided by sampling, since
we almost finish computing connectivity after applying \koutsample{}.
On the other hand, if the graph is reasonably dense, and also has low
diameter, either BFS or LDD Sampling can obtain the best results.
Finally, if the graph diameter is unknown, or known to be high, then
we recommend using \koutsample{}.

\begin{figure}[!t]
\begin{center}
\vspace{-20pt}
\includegraphics[scale=0.60]{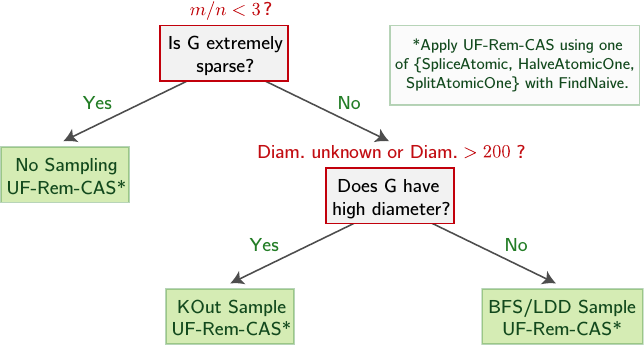}
\vspace{1em}
\caption{\label{fig:decision_tree}
\revised{
Decision tree for selecting a \framework{} algorithm based on the
input graph properties. The text shown in red are suggested
(heuristic) decision rules derived from our evaluation.
}
}
\end{center}
\end{figure}
}

\section{Applying \framework{}}\label{sec:applications}
In this section, we examine how \framework{} can be used to accelerate
graph processing in two important graph applications.

\subsection{Approximate Minimum Spanning Forest}
Computing a minimum spanning forest (MSF) of a weighted undirected
graph is an important problem, and some popular graph clustering
algorithms including single linkage, and affinity clustering can be
viewed as post-processing of a graph's minimum spanning
forest~\cite{irbook, AffinityCluster}. We consider the closely related
problem of computing an \emph{approximate} MSF, and show how a
folklore approach can be accelerated using \framework{}.

\myparagraph{Definition}
Consider a weighted graph $G(V,E,w)$. Let $\mathcal{F}_{\mathsf{OPT}}$
represent any spanning forest of $G$ of minimum weight. The
approximate minimum spanning forest (AMSF) problem is to compute a
spanning forest $\mathcal{F}_{\mathsf{APX}}$ of $G$, where
$W(\mathcal{F}_{\mathsf{OPT}}) \leq W(\mathcal{F}_{\mathsf{APX}}) \leq
(1+\epsilon)W(\mathcal{F}_{\mathsf{OPT}}).$
We make the standard assumption that the weights are
polynomially-bounded, i.e., $\exists c$ s.t. $\forall e \in E, w(e) =
O(n^{c})$ for some constant $c$.

\myparagraph{Algorithm}
The following is a folklore algorithm for the AMSF problem. Let
$W_{\min} = \min_{e \in E} w(e)$. Bucket the edges of $G$, such that
the $i$'th bucket contains edges with weight in the range
$[W_{\min}(1+\epsilon)^{i}, W_{\min}(1+\epsilon)^{i+1})$.  As the
weights in $G$ are polynomially bounded, there are
$O(\log_{1+\epsilon} n)$ buckets. The algorithm maintains a
connectivity labeling $\mathcal{\labelvar{}}$, which is updated as it
processes the buckets from smallest to largest weight. For the $i$'th
bucket, it first removes all edges that are self-loops. It then
computes a spanning forest $\mathcal{F}_{i}$ on the remaining edges in
the bucket, and updates the connectivity labeling so that
$\mathcal{\labelvar{}}$ represents the connected components of
$\cup_{j=0}^{i} \mathcal{F}_j$. The final minimum spanning forest is
simply the union of all of the computed spanning forests.

\newcommand{\amsfea}[0]{\ensuremath{\mathsf{AMSF}\mhyphen\mathsf{COO}}}
\newcommand{\amsffilter}[0]{\ensuremath{\mathsf{AMSF}\mhyphen\mathsf{F}}}
\newcommand{\amsfnofilter}[0]{\ensuremath{\mathsf{AMSF}\mhyphen\mathsf{NF}}}
\newcommand{\amsfnofiltersample}[0]{\ensuremath{\mathsf{AMSF}\mhyphen\mathsf{NF}\mhyphen\mathsf{S}}}
\newcommand{\gbbsmsf}[0]{\ensuremath{\mathsf{GBBS}\mhyphen\mathsf{MSF}}}

\myparagraph{Variants}
We consider several variants of the algorithm above:

\begin{enumerate}[label=(\textbf{\arabic*}),topsep=1pt,itemsep=0pt,parsep=0pt,leftmargin=15pt]
\item \amsfea{} is a direct implementation of the algorithm described
above, which works by writing the edges in the graph into the COO
format, which is then sorted by weight. The buckets are represented
using $O(\log n)$ pointers into this array.

\item \amsffilter{} avoids extracting \emph{all} of the edges into
COO, and instead extracts the \emph{buckets} in COO format one at a
time from input graph (in CSR representation).  The extracted edges
are removed (filtered out) from the CSR representation, and at the end
of the algorithm, the CSR graph has no remaining edges.

\item \amsfnofilter{} works similarly to \amsffilter{}, with the
exception that the graph is not mutated when extracting edges (and
thus all edges in the graph are inspected in every round).
\end{enumerate}

For the \amsfnofilter{} variant, we consider applying a
\emph{sampling} optimization similar to the one used in \framework{},
which we refer to as \amsfnofiltersample{}.  The idea is to compute,
in each round, the largest connected component in the connectivity
labeling, $L_{\max}$, and to skip processing vertices in the
$L_{\max}$ component and their incident edges when searching for the
spanning forest over edges in the current round.  This optimization is
correct, since any edge emanating out of the $L_{\max}$ component that
is skipped will be considered from the other endpoint, which will not
be skipped.

\begin{figure}[!t]
\begin{center}
\vspace{-14pt}
\includegraphics[scale=0.36]{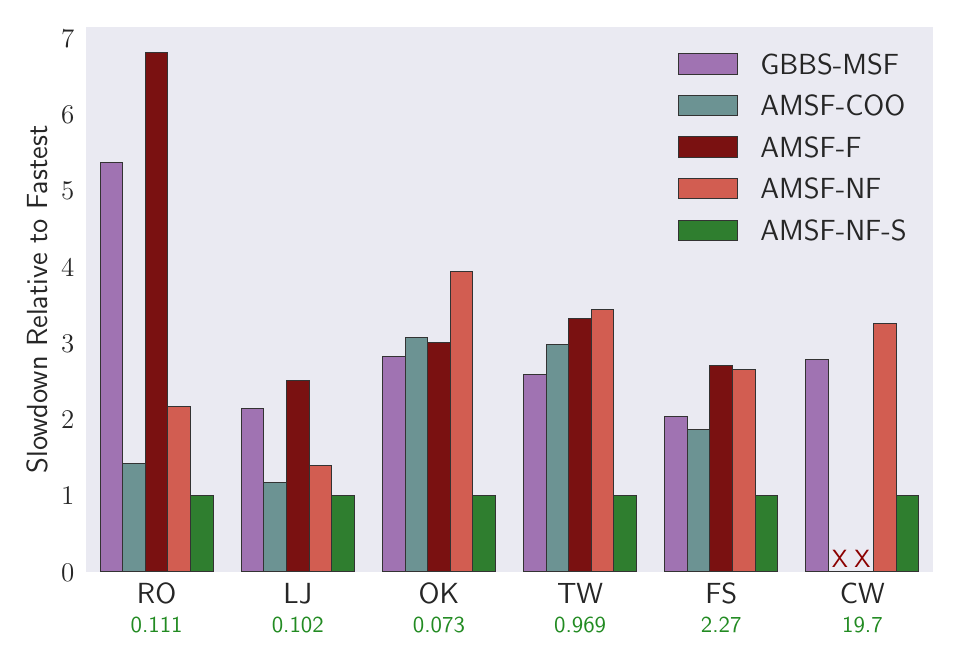}
\caption{\label{fig:amsf_results}
\revised{
Relative performance of the AMSF algorithms with $\epsilon = 0.25$.
The values on the $y$-axis are normalized to the running time of the
fastest algorithm; this time in seconds is shown below the name of
each graph in green. Methods resulting in failure are marked with a
red $\mathsf{x}$.
}
}
\end{center}
\end{figure}

\myparagraph{Experimental Results}
We evaluated all of the AMSF variants above over weighted versions of
our unweighted graph inputs by adding random weights drawn from an
exponential distribution with a constant mean. We set  $\epsilon=0.25$.
We did not evaluate weighted versions of the Hyperlink
graphs due to storage constraints in our experimental environment.
All AMSF variants use the \unionremcas{} method with \splitone{} and
\findnaive{} to concurrently update the connectivity labeling when
processing edges within a bucket. We compare these variants to
\gbbsmsf{}, an \emph{exact} minimum spanning forest (MSF) algorithm from
GBBS~\cite{DhBlSh18}, which is a CSR-based implementation of
\boruvka{}'s algorithm. To the best of our knowledge, \gbbsmsf{} is
the fastest existing multicore MSF algorithm.

Figure~\ref{fig:amsf_results} shows our results.
Other than the two smallest graphs, no AMSF variant without sampling
outperforms the \gbbsmsf{} algorithm by a significant margin. For \amsfea{},
the primary reason for its poor performance on large graphs is the
large overhead of sorting and processing edges stored in COO.
\amsffilter{} suffers for similar reasons, since a large fraction of
the edges fall into the first few buckets, which are explicitly stored
in COO. Both \amsfea{} and \amsffilter{} fail due to memory allocation
errors on the ClueWeb graph for this reason. \amsfnofiltersample{}
consistently attains the best performance across all graphs, obtaining
between 2.03--5.36x speedup over the exact algorithm (2.95x on
average), since it quickly finds a partial spanning forest spanning
most of the largest connected component after processing the first few
buckets, and can skip processing vertices in this component in
subsequent rounds.

\subsection{Index-Based SCAN}
The Structural Clustering Algorithm for Networks (SCAN) clustering
algorithm~\cite{xu2007scan} clusters graphs using the idea that
`similar' vertices have similar neighbor sets, e.g., using a
similarity measure such as cosine similarity.  SCAN is defined using
two parameters: a similarity threshold $\epsilon \in [0, 1]$ and $\mu
\ge 2$. Two vertices are said to be \emph{$\epsilon$-similar} if their
similarity is at least $\epsilon$. A vertex is a \emph{core} vertex if
it has at least $\mu$ neighbors that it is $\epsilon$-similar to. The
objective is to find a maximal clustering where all vertices within a
cluster are connected over a path of $\epsilon$-similar edges.

\myparagraph{GS*-Index and GS*-Query~\cite{wen2017efficient}}
Motivated by the fact that often one is interested in finding multiple
clusterings with varying $\epsilon$ and $\mu$, the
GS*-Index~\cite{wen2017efficient} algorithm builds an index structure
so that cluster retrieval can be quickly performed using the index.
Once the index has been computed, for a given $\epsilon$ and $\mu$ the
query algorithm, GS*-Query (Algorithm 4 in~\cite{wen2017efficient}),
performs a sequential search from the core vertices, considering only
$\epsilon$-similar edges.

\begin{figure}[!t]
\begin{center}
\vspace{-18pt}
\includegraphics[scale=0.48]{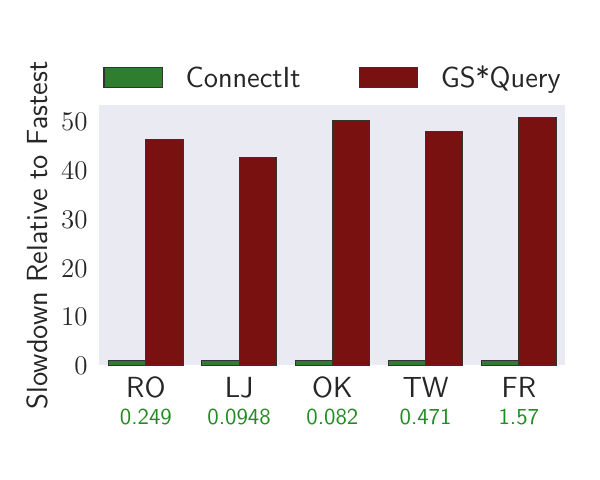}
\caption{\label{fig:scan_query}
\revised{
Relative performance of our parallel implementation of GS*-Query using
\framework{}, and GS*-Query when searching for a clustering with $\epsilon = 0.1$ and $\mu
= 3$. The fastest running time in seconds for each graph is shown in green.
}
}
\end{center}
\end{figure}

\myparagraph{Experimental Results}
We parallelized the GS*-Query algorithm with \framework{} using the
\unionremcas{} algorithm with \findnaive{} and \splitone{}, and
compared the parallel query algorithm to the sequential GS*-Query
algorithm.  Figure~\ref{fig:scan_query} shows the
results of our evaluation for $\epsilon=0.1$ and $\mu=3$ on a subset
of our graph inputs. The index requires $O(m)$ space (with a
significant constant), and so we were unable to evaluate it on our
larger graph inputs. Our results show that by using \framework{}, we
can obtain between 42.5--50.5x speedup (47.4x on average) over the
original sequential GS*-Query, potentially enabling an order of
magnitude more clusterings to be evaluated by users.
We have recently also parallelized the index construction algorithm
(GS*-Index)~\cite{tseng2021scan}.

\section{Conclusion}\label{sec:conclusion}
We have introduced the \framework framework, which provides orders of
magnitude more parallel static and incremental connectivity and
spanning forest implementations than what currently exist today.  We
have found that the fastest multicore implementations in \framework
significantly outperform state-of-the-art parallel solutions.  We
believe that this paper is one of the most comprehensive evaluation of
multicore connectivity implementations to date.


\section*{Acknowledgement}
Thanks to Edward Fan for initial implementations of Rem's algorithm.
Thanks to Guy Blelloch, Siddhartha Jayanti, Jakub Lacki, and Yuanhao
Wei for for helpful discussions and suggestions. The name of our framework is inspired by Saman Amarasinghe and the Commit group.
This research was supported by DOE Early Career Award \#DE-SC0018947, NSF
CAREER Award \#CCF-1845763, Google Faculty Research Award, DARPA SDH Award \#HR0011-18-3-0007, and
Applications Driving Architectures (ADA) Research Center, a JUMP
Center co-sponsored by SRC and DARPA.

\bibliographystyle{ACM-Reference-Format}
\bibliography{references}

\clearpage

\appendix

\section{Appendix Overview}
The appendix is structured as follows:

\begin{enumerate}[topsep=0pt,itemsep=0pt,parsep=0pt]
  \item Section~\ref{apx:framework} contains deferred content and
  proofs about the \framework{} framework.

  \begin{enumerate}[label=(\roman*),topsep=0pt,itemsep=0pt,parsep=0pt,leftmargin=20pt]
  \item Section~\ref{apx:samplepf} contains correctness proofs for our
  sampling algorithms.

  \item Section~\ref{apx:subsec:finish} contains correctness proofs
  about our finish algorithms, as well as more information about our
  Shiloach-Vishkin, Stergiou, and Label Propagation algorithms.

  \item Section~\ref{apx:spanning_forest} describes the \framework{}
  framework for spanning forest.

  \item Lastly, Section~\ref{apx:streaming} describes the \framework{}
  framework for streaming in the batch-incremental and wait-free
  asynchronous settings.

  \end{enumerate}

\item Section~\ref{apx:experiments} contains deferred experimental
  results including:
  \begin{enumerate}[label=(\roman*),topsep=0pt,itemsep=0pt,parsep=0pt,leftmargin=20pt]
  \item Results about \framework{} algorithms in the setting
    \emph{without sampling} (Section~\ref{apx:cpu_no_sample}). This
    includes a performance characterization of our union-find
    and Liu-Tarjan algorithms based on performance counters and other
    metrics.

  \item Results about \framework{} algorithms in the setting
    \emph{with sampling} (Section~\ref{apx:cpu_no_sample}).

  \item Additional details about the experimental comparisons to
    state-of-the-art algorithms done in this paper
    (Section~\ref{apx:static_cpu_comparison}).

  \item An experimental evaluation of the sampling schemes considered
  in this paper, analyzing how adjusting their parameters affects
  their running time, and the sampling quality
  (Section~\ref{sec:sampling_eval}).

  \item A comparison between \framework{} algorithms, and basic graph
    processing primitives that map, and perform indirect reads across
    all edges (Section~\ref{sec:edge_gather}).

  \item A comparison between \framework{} algorithms on very large
    graphs, and existing state-of-the-art results for such large
    graphs (Section~\ref{sec:large_graphs}).

  \end{enumerate}

\item For completeness, we provide pseudocode for our
  algorithms (Section~\ref{apx:pseudocode}).
\end{enumerate}

\myparagraph{Additional Primitives}
We define an additional primitive used in the
following sections:
a \defn{writeMin} takes two arguments, a memory location \emph{x} and
a value \emph{val}. If \emph{val} is less than the value stored at
\emph{x}, then \defn{writeMin} atomically updates the value at \emph{x} to
\emph{val} and returns \emph{true}; otherwise it returns \emph{false}.
A \defn{writeMin} can be implemented with a loop that attempts a CAS as long
as \emph{val} is less than the value at
\emph{x}~\cite{ShunBFG2013}.

\section{Framework}\label{apx:framework}
In this section, we provide deferred correctness proofs, as well as
descriptions of algorithms deferred from the main body of the paper.

\subsection{Correctness for Sampling Algorithms}\label{apx:samplepf}

\begin{theorem}\label{thm:apx:samplingcorrect}
\koutsample{}, \bfssample{}, and \lddsample{} all produce connectivity
labelings satisfying Definition~\ref{def:sampling_produces_trees}.
\end{theorem}
\begin{proof}

Recall that Definition~\ref{def:sampling_produces_trees} requires that
the labeling $\labelvar{}$ emitted by the sampling method is a valid
partial connectivity labeling (i.e., if two vertices $u,v \in V$ are
in the same tree in $C$, then they are connected in $G$). \bfssample{}
satisfies this requirement since the labeling represents a single
connected component which is clearly a valid partial labeling.
Similarly, \lddsample{} satisfies this definition since the labeling
is precisely the contraction induced by the low-diameter decomposition
clustering, and it is easy to check that this is a valid partial
labeling.  Finally, since we use a linearizable union-find algorithm
as our sub-routine in \koutsample{}, a proof almost identical to that
of Theorem~\ref{thm:connectit_connectivity_correct_monotone} shows
that the labeling induced by \koutsample{} is a correct partial
labeling of $G$.
\end{proof}

\subsection{Finish Algorithms}\label{apx:subsec:finish}
We now finish describing our \framework{} finish methods and providing
deferred correctness proofs for our finish methods. We start with the
correctness proofs for linearizably monotonic methods.

\subsubsection{Correctness for Linearizably Monotonic
Methods}\label{apx:linmoncorrect}

\begin{theorem}\label{thm:connectit_connectivity_correct_monotone}
The connectivity algorithm in \framework{}
(Algorithm~\ref{alg:framework_connectivity}) applied with a sampling
method $\mathcal{S}$ satisfying
Definition~\ref{def:sampling_produces_trees} and a linearizably
monotone finish method $\mathcal{F}$ is correct.
\end{theorem}
\begin{proof}
Let $\labelvar{}$ be the partial labeling produced by the sampling
method $\mathcal{S}$.
Now, consider $\mathcal{F}$, which is a linearizably monotone finish
method operating on the partial connectivity labeling $\labelvar{}$. Since we
are handling an abstract connectivity algorithm $\mathcal{F}$, it is
helpful to think about operations of two types: (i) those that try to
``apply" an edge $(u,v)$ (i.e., it checks whether $u$ and $v$ are in
different trees, and if so changes the labeling so that $u$ and $v$
are in the same tree in the new labeling), and (ii) lower-level
operations that compress trees.
Furthermore, note that $\mathcal{F}$ will avoid applying (directed)
edges coming out of $\largestcomp{}$, which could be problematic
since there may be many edges of the form $(u,v)$ where $u \in
\largestcomp{}$ and $v \notin \largestcomp{}$ that must be
examined to compute a correct labeling. However, since the graph is
symmetric, $\mathcal{F}$ does apply all such inter-cluster
edges, since such edges will be applied while processing vertex $v
\notin \largestcomp{}$, which is contained in some initial component
$c \neq \largestcomp{}$.  The algorithm must process this component,
since it only skips applying the directed edges incident to vertices
in component $\largestcomp{}$.

Consider the linearization order of the finish algorithm
$\mathcal{F}$. Before running $\mathcal{F}$, the labeling $\labelvar{}$
represents some set of intermediate clusters that is a partial
connectivity labeling of $G$ (i.e., if there is a path between two
vertices $u$ and $v$ in the trees represented by the labeling
$\labelvar{}$, then this implies that there is a path between $u$ and
$v$ in $G$).  We will show by induction on the linearization order
that each operation preserves the fact that $\labelvar{}$ is a partial
connectivity labeling of $G$. Consider the two types of operations:

\begin{itemize}
  \item Operations of type (ii) described above are trivial to handle
    since the operations only change the structure of the labeling,
    and not the connectivity information represented by it. Thus, the
    labeling is still a valid partial connectivity labeling of $G$.

  \item Next, consider an operation of type (i), which applies an edge
    $(u,v)$.  If $u$ and $v$ are in the same tree before the
    operation, the operation does not modify the structure of the
    trees and thus the new connectivity labeling (if any labels
    changed) is equivalent to the previous labeling (and is thus
    valid).  Otherwise, if $u$ and $v$ are in different trees, the
    operation will produce a new labeling where $u$ and $v$ are in the
    same tree, and the labeling for any vertices $z$ not in the same
    tree as either $u$ or $v$ is unmodified. Thus, the new labeling is
    a partial connectivity labeling of $G$, since the algorithm only
    changes labels of vertices in the trees containing $u$ and $v$,
    and by induction all vertices in each of $u$'s and $v$'s trees
    were connected by paths in $G$, and due to the $(u,v)$ edge, all
    vertices in the union of these trees have paths between each other
    in $G$.
\end{itemize}

Therefore, at the end of the execution, the algorithm has computed a
partial connectivity labeling, and has linearized the apply operations
of all $(u,v)$ edges, where $u$ and $v$ are not both in
$\largestcomp{}$. Note that all edges $(u,v)$, where both $u$ and
$v$ are initially in $\largestcomp{}$ are initially in the same
tree. Therefore, the algorithm has computed a connectivity labeling,
having applied all edges $(u,v) \in E$, and so for any $u,v \in V$,
$u$ and $v$ are in the same tree in $\labelvar{}$ if and only if there
is a path between $u$ and $v$ in $G$, completing the proof.
\end{proof}

\subsubsection{Correctness Properties of Rem's Algorithm}\label{apx:rem_correct}

\begin{restatable}{theorem}{remcorrect}\label{thm:rem_correct}
  \unionremcas{} and \unionremlock{} are correct connectivity
  algorithms in the phase-concurrent setting, and result in
  linearizable results for a set of phase-concurrent union and find
  operations.
\end{restatable}
\begin{proof}
  We provide a high-level proof sketch.  The proof is based on
  building a graph of completed and currently active union operations.
  Let this graph be $G_{C}$. The graph (roughly) corresponds to the
  trees concurrently maintained by the algorithm, and includes extra
  edges corresponding to the currently executing \textsc{Union}
  operations.
  \begin{itemize}
    \item When a $\textsc{Union}(u,v)$ operation starts (is invoked),
      we place a \emph{shadow edge} between the root of the tree
      currently containing $u$ and the tree currently containing $v$.

    \item When a $\textsc{Union}(u,v)$ operation completes (we receive
      its response), we first check if it returns $\mathsf{True}$
      due to atomically hooking a tree root $r_u$ to a vertex $h_v$.
      If so, we link the tree containing $r_u$ to $h_v$ in the graph.
      Then, we remove the shadow edge corresponding to this Union
      operation from the graph.
  \end{itemize}
  The proof is by induction on the completion times of the operations
  that (i) once a $\textsc{Union}(u,v)$ call completes, $u$ and $v$
  always remain connected in $G_{C}$ and (ii) any vertices $u$ and $v$
  in different components in $G$ are never connected in $G_{C}$. The
  proof is similar to that of
  Theorem~\ref{thm:connectit_connectivity_correct_monotone}.
\end{proof}

\subsubsection{Counter-example for Rem's Algorithm with SpliceAtomic and FindCompress}

\begin{figure}[!t]
  \centering
  \hspace{-2em}
    \includegraphics[width=0.8\columnwidth]{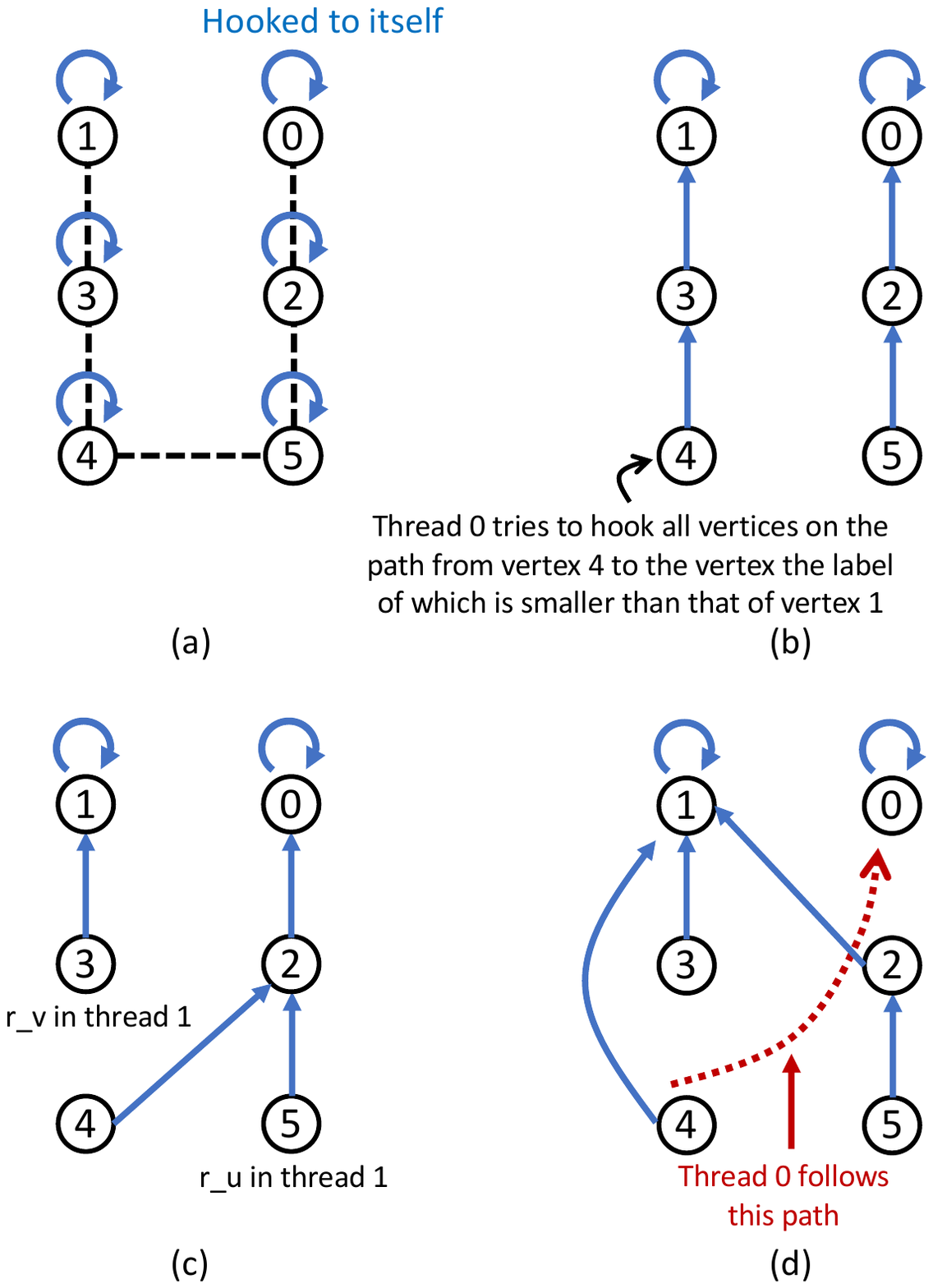}\\
        \caption{\small A counter example of Rem's algorithm with \splice{} and \findcompress{}.
  }\label{fig:rem_countex}
\end{figure}

As mentioned in Section~\ref{subsec:finish}, combining the
\ensuremath{\mathsf{Find-}} \ensuremath{\mathsf{Compress}} option with the \splice{} rule in the \unionremlock{} and
\unionremcas{} algorithms results in incorrect connectivity
algorithms. Figure~\ref{fig:rem_countex} illustrates a counter-example
for these algorithms. Without loss of generality, we will explain why
\unionremlock{} can lead to an incorrect result with \findcompress{}
and \splice{} (the counter-example is identical for \unionremcas{}).

Consider the following input graph where $V=\{0,1,2,3,4,5\}$ and
$E=\{(0,2),(1,3),(2,5),(3,4),(4,5)\}$ (shown in Figure~\ref{fig:rem_countex} (a)).
Assume the \unionremlock{} algorithm is applied to the edges $(2,5)$ and $(0,2)$ sequentially.
After that, thread 0 applies \unionremlock{} to the edge $(3,4)$, and
is about to execute Line~\ref{countex:compress} in
Algorithm~\ref{alg:find_options}
(i.e., thread 0 is in the process of executing \findcompress{}
(Line~\ref{countex:remfind} in Algorithm~\ref{alg:union_rem_lock})).

In the meantime, suppose another thread finishes applying
\ensuremath{\mathsf{UF\mhyphen}} \ensuremath{\mathsf{Rem\mhyphen Lock}} to the edge $(1,3)$. The state of the trees at this
point in the execution is shown in Figure~\ref{fig:rem_countex} (b).
The only outstanding operation is in thread 0, which will still
execute \findcompress{} from vertex $4$.

Now, thread 1 applies \unionremlock{} to the edge $(4,5)$: after one iteration of while statement (Line~\ref{countex:while} in Algorithm~\ref{alg:union_rem_lock}),
$P[4]$ is changed from 3 to 2, and $r_u=5$, $r_v=3$ (Figure~\ref{fig:rem_countex} (c)).

After that, thread 0 executes the while loop
(Line~\ref{countex:compress} in Algorithm~\ref{alg:find_options}),
which hooks $P[4]$ and $P[2]$ to 1 (Figure~\ref{fig:rem_countex} (d));
$P[0]$ is not hooked to vertex $1$ as $P[0] = 0 < 1$.

Therefore vertex 0 is now isolated, and so thread 1 cannot link vertex
1 to any other vertices during \unionremcas{}. Thus, the number of
connected components after applying \unionremlock{} is 2, which is
incorrect.

\subsection*{Additional Finish Methods}

\subsubsection{Shiloach-Vishkin}\label{sec:shiloachvishkin}

We include Shiloach-Vishkin's algorithm~\cite{ShiloachV82} in
\framework{}, which is a classic parallel connectivity algorithm. The
algorithm works by combining vertices into trees over a series of
synchronous rounds using linking rules.  Since only roots of trees can
be linked (from larger root to smaller), the algorithm is naturally monotonic.

We provide pseudocode for our implementation of Shiloach-Vishkin in
Algorithm~\ref{alg:shiloach_vishkin} in Appendix~\ref{apx:pseudocode}.
The algorithm iteratively hooks roots of trees onto each other and
compresses these trees by linking all vertices to the root of the
tree.  Each \emph{round} of the algorithm maps over all edges. If an
edge goes between two tree roots, the algorithm tries to \emph{hook}
the larger root to the smaller one.  This operation can either be done
using a plain write, which is what existing implementations of
Shiloach-Vishkin use, or using the atomic \emph{writeMin} operation to
hook to the lowest incident root, which is what our algorithm does. A
round ends by making the parent pointers of all vertices point to the
root of their tree using pointer jumping.
By considering the hook operations the algorithm tries to apply to
each edge as a set of concurrent union operations, it is easy to show
that this algorithm is linearizable, and thus linearizably monotonic,
since it only links tree roots. Overall, our implementation of
Shiloach-Vishkin requires $O(\log n)$ rounds, and each round can be
implemented in $O(n + m)$ work and $O(\log n)$ depth. The overall work
and depth of the algorithm are $O((n+m)\log n)$ and $O(\log^2 n)$ respectively.

\subsubsection{Stergiou}\label{sec:stergiou}
Stergiou et al.~\cite{Stergiou2018} presented a simple connectivity
algorithm for the bulk synchronous parallel (BSP) model. Stergiou's
algorithm is nearly expressible in the Liu-Tarjan framework, with the
exception that the algorithm uses two parents arrays,
\emph{prevParent} and \emph{curParent}, instead of just one array as
in the Liu-Tarjan algorithms.
\footnote{Through personal communication with the authors of~\cite{Stergiou2018}, we have learned that they use a single array in their actual implementation.}
In each round, the algorithm first updates \emph{prevParent} to be
equal to \emph{curParent}. It then performs a \parentconnect{},
updating \emph{curParent} based on the parent information in
\emph{prevParent}. Finally, it applies \shortcut{} on \emph{curParent}
and repeats until the parent array no longer changes. This algorithm
is similar to the $\mathsf{PUS}$ algorithm of Liu-Tarjan, except that
it uses two parent arrays instead of a single array. Stergiou's
algorithm is not monotonic since the algorithm can relabel non-root
vertices in the parents array. The correctness proof for this
algorithm is similar to the one for the Liu-Tarjan framework. We
compose Stergiou's algorithm with sampling methods in the same way as
with the Liu-Tarjan framework.

\subsubsection{Label Propagation}\label{sec:labelprop}
Label propagation is arguably the most frequently implemented parallel
graph connectivity today, and has been implemented by many existing
popular graph processing systems, including Pregel, Giraph, and other
frameworks~\cite{Pregel, Giraph, Nguyen2014, ShunB2013,
da2015flashgraph}. The \labelprop{} algorithm can be viewed as
iterative sparse-matrix vector multiplication (SpMV), where the
product is done over the $(\min,\min)$ semiring.

The implementation of this algorithm works as follows.  Initially,
each vertex is its own parent (we interpret the parents in the
algorithm as the labels).  In each round, the algorithm maintains a
subset of vertices called a \defn{frontier} that had their parent
change in the previous round (at the start of the algorithm, this set
contains all vertices).  Then, in each round, the algorithm processes
all edges incident to vertices in the current frontier, writing the
minimum parent ID to the neighbors using a \textsc{writeMin}.  The
algorithm terminates once the array of parents no longer changes,
which will occur within $D$ rounds where $D$ is the diameter of the
graph. In the worst case, the algorithm will process all edges in each
round giving overall work $O(mD)$ and depth $(D \log n)$. Note that
the \labelprop{} algorithm is \emph{not monotonic}, since the
label updates only affect the 1-hop neighborhood of a vertex spreading
its label, and not entire trees.

The correctness of this algorithm is folklore, but we provide a
generic proof of correctness for it, and \otherminbased{} algorithms
like Liu-Tarjan and Stergiou's algorithm in
Theorem~\ref{thm:min_based_correctness} below.
We compose \labelprop{} with sampling methods by using the same
idea as with our sampled Liu-Tarjan framework, since relabeling the
largest component to have the smallest ID ensures that vertices in
this component never change their ID in subsequent rounds.

\subsection*{Correctness for Other Min-Based
Methods}\label{apx:min_based_correctness}
We start by abstractly defining the \emph{\otherminbased{} finish
methods}, which captures Liu-Tarjan's algorithms, \labelprop{},
and Stergiou's algorithm.
\begin{definition}
An \defn{\otherminbased{} finish method} is a round-based algorithm
that repeatedly performs an ``application" step on all edges. An edge
application takes an edge $(u,v)$ and may update $P[u]$ (resp. $P[v]$)
based on $P[v]$ (resp.  $P[u]$), where a parent value is updated if
and only if the new value is strictly smaller than the previous value.
The algorithm terminates once the parent values stop changing.
\end{definition}
Label propagation, Liu-Tarjan's algorithms, and Stergiou's algorithms
all satisfy this definition. The following theorem states that
composing \otherminbased{} algorithms with sampling algorithms
produces correct results in \framework.

\begin{theorem}\label{thm:min_based_correctness}
  The connectivity meta-algorithm in \framework{}
  (Algorithm~\ref{alg:framework_connectivity}) applied with a sampling
  method $\mathcal{S}$ satisfying
  Definition~\ref{def:sampling_produces_trees} and an \otherminbased{}
  finish method $\mathcal{F}$ produces a correct connectivity
  labeling.
\end{theorem}
\begin{proof}
The proof is by contradiction. Let $\labelvar{}$ be the labels produced by
$\mathcal{S}$ satisfying Definition~\ref{def:sampling_produces_trees},
that is they are consistent with some contraction of $G$ such that
composing the contraction with the connectivity induced by
inter-component edges results in a correct connectivity labeling of
$G$. The proof only relies on the fact that the sampling method only
maps vertices to the same label in $\labelvar{}$ if they are in the same
component. Let $\labelvar{}'$ be the final connectivity labeling after running
$\mathcal{F}$.

Suppose that two vertices $u$ and $v$ in the same component in $G$ are
assigned different labels in $\labelvar{}'$. Since $u$ and $v$ are connected,
there is a path between them in $G$, and on this path, there must be an
edge $(x,y)$ where $\labelvar{}'[x] \neq \labelvar{}'[y]$. However, since the
\otherminbased{} method $\mathcal{F}$ applies all edges in each round,
on the last round it must have applied the $(x,y)$ edge, and since
$\labelvar{}'[x] \neq \labelvar{}'[y]$, one of these values must have
changed in this round, contradicting the assumption that this was the
last round of the algorithm.

Furthermore, the \otherminbased{} algorithms composed with sampling
methods can skip the edges incident to the largest sampled component,
since the largest sampled component is relabeled to have the smallest
ID, and thus by the preceding argument, all vertices that are
connected to the largest sampled component in $G$ will eventually
acquire its ID.  \end{proof}

\subsection{\framework{}: Spanning Forest}\label{apx:spanning_forest}
In this section, we show how to extend \framework{} to solve the
problem of generating a spanning forest of the graph.
Algorithm~\ref{alg:framework_spanning_forest} shows the generic
\framework{} algorithm for this problem. Like our connectivity
algorithm, the algorithm can be supplied with different combinations
of sample and finish methods and generate correct spanning forest
algorithms. The main idea in the algorithm is to map each edge in the
spanning forest to a unique vertex that is one of the edge's
endpoints.

\myparagraph{Algorithm Description} The main changes between the
\oursystem{} connectivity algorithm and spanning forest algorithm are
as follows. The spanning forest algorithm requires (i) supplying an
$\codevar{edges}$ array to the algorithm, which is an array of size
$|V|$ where each entry is initially an edge pair of a suitable null
value $(\infty, \infty)$ (Line~\ref{sf:initialize_edges}); (ii) the
sampling method generates a subset of the spanning forest edges in
addition to computing partial connectivity information
(Line~\ref{sf:sampling}); and (iii) the $\codevar{finish}$ procedure
now takes both the partially computed components and the partially
computed forest from the sampling step, and finishes computing the
spanning forest of the graph (Line~\ref{sf:finish}).  Finally, since
some of the vertices may not have a corresponding spanning forest
edge, we filter the \textsc{edges} array to remove any pairs which are
not edges (Line~\ref{line:filteredges}).
Our goal when building this interface was to allow as many finish
methods as possible to work with it, and so we did not consider
optimizations, such as replacing the edges array, which is an array of
edge pairs, by a parents array which would only store a single
endpoint per vertex (such a modification could be made to work for
a small subset of the finish methods, but would exclude many other
finish methods from working in the interface).

We show in the following subsections that any correct sampling method
can be combined with any \emph{root-based finish method} (defined in
Section~\ref{sec:spanningforestfinish}) to generate correct spanning
forest algorithms (Theorem~\ref{thm:spanning_forest_correct}). Our
framework currently only works for root-based algorithms, and thus
excludes algorithms such as the class of non-root-based algorithms in
Liu-Tarjan's framework.

\begin{algorithm}\caption{\framework{} Framework: Spanning Forest} \label{alg:framework_spanning_forest}
\small
\begin{algorithmic}[1]
\Procedure{SpanningForest}{$G(V, E), \codevar{sampling}, \codevar{finish}$}
  \State $\codevar{labels} \gets \{i \rightarrow i\ |\ i \in [|V|]\}$\label{sf:initialize_comps}
  \State $\codevar{edges} \gets \{i \rightarrow (\infty, \infty) \ |\ i \in [|V|]\}$\label{sf:initialize_edges}
  \State $(\codevar{labels},\codevar{edges}) \gets \codevar{sampling}.\textsc{SampleForest}(G, \codevar{labels}, \codevar{edges})$\label{sf:sampling}
  \State $\largestcomp{} \gets \textsc{IdentifyFrequentComp}(\codevar{labels})$
  \State $(\codevar{labels},\codevar{edges}) \gets \codevar{finish}.\textsc{FinishForest}(G, \codevar{labels},\codevar{edges}, \largestcomp{})$\label{sf:finish}
  \State \algorithmicreturn{} $\textsc{Filter}(\codevar{edges}, \textsc{fn} (u, v) \rightarrow ((u,v) \neq (\infty, \infty)))$\label{line:filteredges}
\EndProcedure
\end{algorithmic}
\end{algorithm}

\subsubsection{Sampling Methods for Spanning Forest}
We first introduce the following correctness definition for spanning
forest sampling methods, which ensures that the sample methods compose
correctly with all finish methods that are subsequently applied.

\begin{definition}[Sampling for Spanning Forest]\label{def:samplingvalidity}
  Given an undirected graph $G(V, E)$ let $\mathcal{E}_{S}$ denote
  the set of forest edges generated by a sampling method
  $\mathcal{S}$, and $C$ denote the connectivity labeling after
  applying $S$. The sampling method $\mathcal{S}$ is \defn{correct} if
  \begin{enumerate}[label=(\arabic*)]
    \item The sampling method satisfies
    Definition~\ref{def:sampling_produces_trees}.\label{def:sv_orig}

  \item $\textsc{Connectivity}(V, \mathcal{E}_{S}) = C$, i.e., finding
  the connected components on the spanning forest edges induces the
  connectivity labeling $C$.\label{def:sv_one}

  \item Each edge $e=(u,v) \in \mathcal{E}_{S}$ is assigned to exactly
  one non-root vertex, and no vertex is assigned more than one edge.\label{def:sv_two}
\end{enumerate}

\end{definition}
Requirement~\ref{def:sv_orig} is simply the original notion of
correctness for sampling algorithms in the context of connected
components.
Requirement~\ref{def:sv_one} corresponds to the intuitive
notion that the connectivity labeling generated by the sampling method
corresponds is equivalent to contracting the subset of spanning forest
edges returned by the method.
Requirement~\ref{def:sv_two} is required to correctly compose the
output of our sampling methods with our finish methods. We now discuss
how to adapt each of the sampling methods presented in
Section~\ref{subsec:sampling} to satisfy the requirements of
Definition~\ref{def:samplingvalidity}.

\myparagraph{$k$-Out Sampling} For \koutsample{}, we require the
following modifications. First, we include an edge $e=(u,v)$ as a
forest edge if it was successfully used to hook a root to a different
vertex in the underlying union-find algorithm used by the $k$-Out
sampling method. Let $r(u)$ and $r(v)$ be the roots of $u$ and $v$,
respectively, at the time this edge is hooked. Without loss of
generality, suppose $r(u)$ is the root we hook ($r(u)$ is no longer a
tree root after hooking). We store this edge in the \textsc{edges}
array by setting $\textsc{edges}[r(u)] = e$.

\myparagraph{BFS Sampling and LDD Sampling} We modify \bfssample{} to
take all edges in the BFS tree as spanning forest edges. Consider a
tree edge $e=(u,v)$ in the BFS tree where $u$ is the parent of $v$ in
the tree. Then, we assign $e$ to $v$ as part of the sampling
method. Finally, for low-diameter decomposition sampling, we take the
edges traversed by the searches in the LDD as spanning forest
edges.  As in BFS sampling, we assign each tree edge $e=(u,v)$ to the
\emph{child} in the BFS tree rooted at the cluster containing $u$ and
$v$.

\begin{theorem}
\koutsample{}, \bfssample{}, and \lddsample{} all satisfy Definition~\ref{def:samplingvalidity}.
\end{theorem}
\begin{proof}
Theorem~\ref{thm:apx:samplingcorrect} shows that all three sampling
methods satisfy requirement~\ref{def:sv_orig}
Definition~\ref{def:samplingvalidity}.

For \koutsample{}, since the union-find algorithm is monotonic, it is
easy to verify that the set of forest edges we include from the
sampling produces exactly the trees represented by the parent array
after \koutsample{}, and thus this method satisfies
requirement~\ref{def:sv_one} of
Definition~\ref{def:samplingvalidity}. Requirement~\ref{def:sv_two}
can also be satisfied since \koutsample{} algorithm only hooks
roots. Recall that when the algorithm applies an edge $(u,v)$ that
hooks a root $r(u)$ to $r(v)$, it sets $\textsc{edges}[r(u)] = e$.
This assignment satisfies the first part of \ref{def:sv_two} since $e$
is included in the forest, and is only assigned to exactly one
non-root vertex, namely $r(u)$.  It also satisfies the second part of
\ref{def:sv_two} since only vertices that are roots at some point in
the algorithm have edges assigned to them, and every root is hooked at
most once.  Therefore, a vertex has at most one edge assigned to it,
as required.

Let $T$ be the BFS tree of the sampled component. Since a BFS tree is
clearly connected, the sampling method clearly satisfies
requirement~\ref{def:sv_one} of Definition~\ref{def:samplingvalidity}.
Now, consider a tree edge $e=(u,v)$ in the BFS tree where $u$ is the
parent of $v$ in the tree, and recall that we assign $e$ to $v$, which
is not a root vertex. Since each vertex in the tree has exactly one
parent, this assignment satisfies requirement~\ref{def:sv_two} of
Definition~\ref{def:samplingvalidity}.

The proof for \lddsample{} is identical to that of \bfssample{} since
each cluster in the LDD is explored using a BFS rooted at the center.
We described the algorithm used to compute LDD earlier in
Section~\ref{subsec:sampling}, and recall that it works by sampling a
start time from an exponential distribution for each vertex, and
performing a simultaneous breadth-first search.
Vertices join the search at their start time if they are not already
covered by a cluster. We refer to Miller et al.~\cite{MillerPX2013}
and Shun et al.~\cite{SDB14} for more details on LDD.
\end{proof}

\subsubsection{Finish Methods for Spanning
Forest}\label{sec:spanningforestfinish}

Next, we apply our framework to handle a class of \defn{root-based
spanning-forest algorithms}. As the name suggests, a root-based
spanning-forest algorithm is one that adds an edge to a spanning
forest if the edge successfully hooks the root of the tree containing
one of its endpoints to a smaller vertex in the other endpoint's tree.
It is easy to see that the entire class of union-find algorithms
considered in this paper satisfy this definition, since each of the
union steps in these algorithms only operates on roots, and atomically
hooks a root to a smaller vertex in the other tree.
The main correctness property of our framework for spanning forest is
summarized by the following theorem:

\begin{theorem}\label{thm:spanning_forest_correct}
Given an undirected graph $G(V, E)$, composing a correct sampling
method $\mathcal{S}$
(as per Definition~\ref{def:samplingvalidity}) with a root-based
spanning forest method $\mathcal{F}$ using
Algorithm~\ref{alg:framework_spanning_forest} outputs a correct
spanning forest.
\end{theorem}
\begin{proof}
We provide a high-level proof sketch. The parents array, $P$, emitted
by the sampling method represents a set of rooted trees, where the
forest edges discovered by the sampling method, $E_S$, induce the
connectivity information stored in the trees. Applying a root-based
spanning forest method to this configuration (the parents array, $P$,
and the edge array, $E_S$) can be viewed as applying a root-based
spanning forest algorithm to $G[P]$, the graph induced by the
connectivity mapping generated by the sampling method.

The finish method by definition will hook a number of roots during its
execution. Since each root can be hooked at most once, and does not
have an edge written into its index in $E_S$ yet (by
requirement~\ref{def:sv_two} of
Definition~\ref{def:samplingvalidity}), the finish method can
successfully store this edge in $E_S$. Since the finish method is
correct, the algorithm outputs a correct spanning forest for $G[P]$.
Applying the fact that the sampling method also satisfies
Definition~\ref{def:sampling_produces_trees} shows that the union of
the forest generated by the sampling method and the forest generated on
$G[P]$ is a spanning forest for $G$.
\end{proof}

Theorem~\ref{thm:spanning_forest_correct} provides correctness
guarantees for the following algorithms in our framework, which can be
composed with all three sampling methods:
\begin{itemize}[itemsep=0pt]
  \item All union-find variants, including variants of Rem's algorithm, and all
    internal parameter combinations for these algorithms
  \item Shiloach-Vishkin's algorithm
  \item The root-based algorithms in the Liu-Tarjan framework
\end{itemize}

\subsection{\framework{}: Streaming}\label{apx:streaming}

\myparagraph{Algorithm Description} We first describe the generic
\framework{} algorithm for computing streaming connectivity. The
pseudocode for our algorithm is described in
Algorithm~\ref{alg:framework_connectivity_streaming}.

Algorithm~\ref{alg:framework_connectivity_streaming} is initialized
similarly to the \framework{} connectivity algorithm
(Algorithm~\ref{alg:framework_connectivity}) on an initial graph
(which is possibly empty). Before batch-processing starts, the
framework calls the \textsc{Initialize} method on the initial graph
(Line~\ref{line:streaming_initialize}).  The \textsc{ProcessBatch}
function provided by the framework can then be called by clients,
where each \textsc{ProcessBatch} call is executed internally in
parallel by calling the underlying finish algorithm's batch-processing
mechanism (Line~\ref{line:streaming_finish}). The algorithm then
returns the results of the queries (Line~\ref{line:returnqueries}).

\begin{algorithm}\caption{\framework{} Framework: Streaming} \label{alg:framework_connectivity_streaming}
\small
\begin{algorithmic}[1]
\Procedure{Initialize}{$G(V, E), \codevar{sampling\_opt}, \codevar{finish\_opt}$}
  \State $\codevar{sampling} \gets \textsc{GetSamplingAlgorithm}(\codevar{sample\_opt})$
  \State $\codevar{finish} \gets \textsc{GetFinishAlgorithm}(\codevar{finish\_opt})$
  \State $\codevar{components} \gets \{i \rightarrow i\ |\ i \in [|V|]\}$\label{line:initialize_init}
  \State $\codevar{sampling}.\textsc{SampleComponents}(G, \codevar{components})$\label{line:initialize_sample}
  \State $\largestcomp{} \gets \textsc{IdentifyFrequent}(\codevar{components})$\label{line:initialize_identifyfreq}
  \State $\codevar{finish}.\textsc{FinishComponents}(G, \codevar{components},\largestcomp{})$\label{line:initialize_finish}
  \State \algorithmicreturn{}$\ \codevar{finish}$
\EndProcedure

\State $\mathsf{let}\ \mathcal{BC} \gets \textsc{Initialize}(G,
\codevar{sampling\_opt},
\codevar{finish\_opt})$\label{line:streaming_initialize}

\Procedure{ProcessBatch}{$B = \{\emph{updates}, \emph{queries}\}$}
\State $\codevar{query\_results} \gets \mathcal{BC}.\textsc{ProcessBatch}(B)$\label{line:streaming_finish}
  \State \algorithmicreturn{}$\
  \codevar{query\_results}$\label{line:returnqueries}
\EndProcedure
\end{algorithmic}
\end{algorithm}

\myparagraph{Correctness in the Parallel Batch-Incremental Setting}
What do we mean by an algorithm being correct in the parallel
batch-incremental setting? A natural definition is to enforce that the
algorithm is linearizable for a set of union and find operations, and
that the linearization points of operations that occur in previous
batches are fixed, i.e., operations in future batches cannot affect
the linearization points of operations in previous batches that have
already happened. If union and find operations within a batch are run
concurrently, we require these operations to be linearizable
with respect to the connectivity information at the start of the
batch. We believe this correctness definition, which we use in this
paper, is intuitive, and captures the desired properties of parallel
batch-incremental algorithms.

Recall the different types of streaming algorithms introduced in
Section~\ref{subsec:streaming}.  The reason we categorize algorithms
as Type~\ref{lab:phaseconcurrentalgs}, and apply the updates and finds
phase concurrently is because these algorithms are not linearizable
when applying finds concurrently with updates.

The proof that the algorithms of Type~\ref{lab:waitfreealgs} and
Type~\ref{lab:syncupdates} are correct according to this definition
follows due to the fact that these algorithms are linearizable for a
collection of concurrent unions and finds.
The proof that algorithms of Type~\ref{lab:phaseconcurrentalgs} are
correct in the parallel batch-incremental setting when run
phase-concurrently follows from their correctness for a collection of
edge updates in the static setting. After applying all updates, the
parents array stores the correct connectivity information such that two
vertices are in the same component if and only if they have the same
root in the parents array. Therefore, we have the following theorem:
\begin{theorem}
  Type~\ref{lab:waitfreealgs}, Type~\ref{lab:syncupdates}, and
  Type~\ref{lab:phaseconcurrentalgs} streaming algorithms implemented
  in \framework{} are correct when run in the parallel
  batch-incremental setting.
\end{theorem}

\myparagraph{Correctness in the Wait-Free Asynchronous Setting}
We use the standard definition of linearizability in the wait-free
asynchronous setting (see Section~\ref{sec:prelims}). Recall from
Section~\ref{subsec:streaming} that we only consider
Type~\ref{lab:waitfreealgs} and Type~\ref{lab:syncupdates}
algorithms. The correctness for these two types follows from the fact
that both types of algorithms are linearizable for a collection of
concurrent unions and finds. Note that for Type~\ref{lab:syncupdates},
only the \textsc{IsConnected} operations may run concurrently, since
both the Shiloach-Vishkin and Liu-Tarjan algorithms process updates
synchronously, over a number of rounds. Our point for
Type~\ref{lab:syncupdates} algorithms is that the finds can run
concurrently while the connectivity information is being updated by
the batch-incremental update algorithm.

\begin{theorem}
  Type~\ref{lab:waitfreealgs} and Type~\ref{lab:syncupdates}
  algorithms implemented in \framework{} are correct when run in the
  wait-free asynchronous setting.
\end{theorem}

\section{\framework{} Evaluation}\label{apx:experiments}
In this section we provide additional experimental results that are
deferred from the main body of the paper.

\subsection{Parallel Connectivity without Sampling}\label{apx:cpu_no_sample}
This section reports additional results and performance analyses
about \framework{} in the setting where no sampling is performed.

\subsubsection*{Union-Find Algorithms}
We start by analyzing how quantities such as the \maxpathlen{} and
\totalpathlen{} for union-find algorithms affect the algorithm
performance, as well as how hardware counters can explain performance
differences between union-find algorithms.

\begin{figure*}[!t]
  \centering
  \hspace{-2em}
    \includegraphics[width=0.8\textwidth]{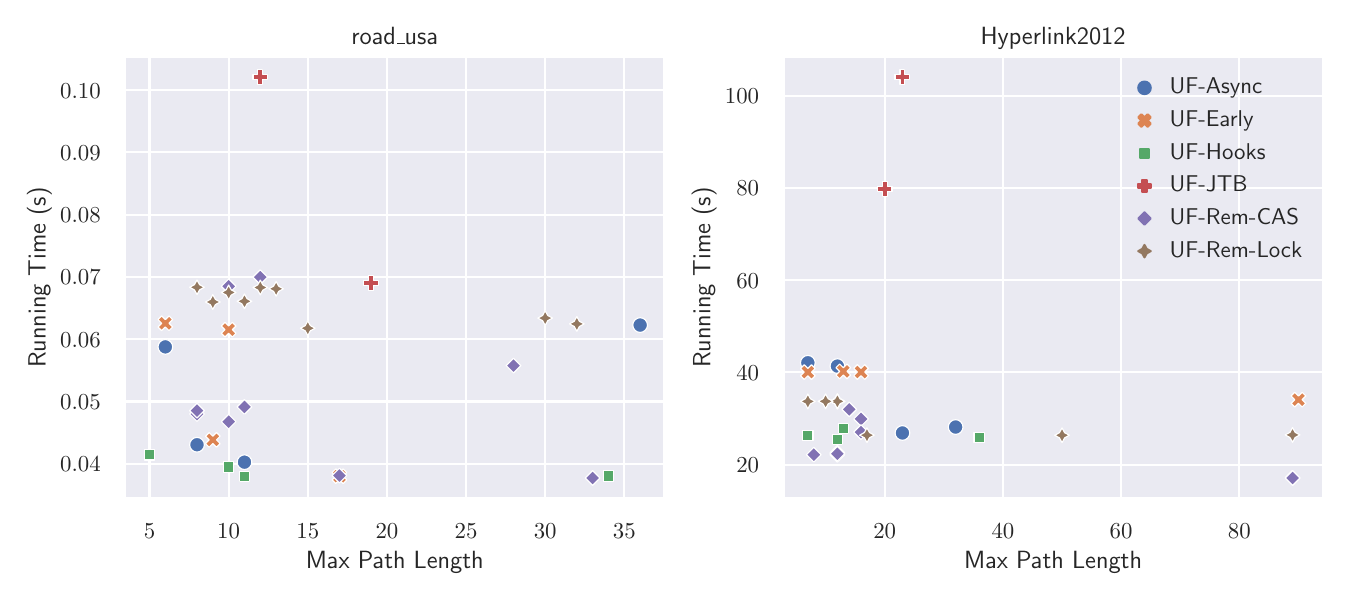}\\ 
	\caption{\small Plot of the \maxpathlen{} encountered during
      algorithm execution vs. running time in seconds for the
      road\_usa and Hyperlink2012 graphs.
  }\label{fig:max_path_vs_running_time}
\end{figure*}

\myparagraph{Max Path Length (MPL)}
Figure~\ref{fig:max_path_vs_running_time} plots the \defn{\maxpathlen{}}
statistic vs. the total running time for different \unionfind{} variants
on the road\_usa and Hyperlink2012 graphs, which are the smallest and
largest graphs used in our evaluation (the plots for other graphs
are elided in the interest of space; we note that they are very
similar). We observed that this value is consistently small across all
executions for both the smallest and largest graphs we test on,
reaching a maximum of 36 on the road\_usa graph, and a maximum of less
than 100 on the Hyperlink2012 graph. Perhaps surprisingly, this
quantity does not seem to be correlated with the running time---on the
Hyperlink2012 graph the \unionremcas{} implementation has one of the
largest \maxpathlen{}s, but still achieves the lowest running
time in this setting.

\begin{figure}[!t]
  \centering
  \hspace{-2em}
    \includegraphics[width=\columnwidth]{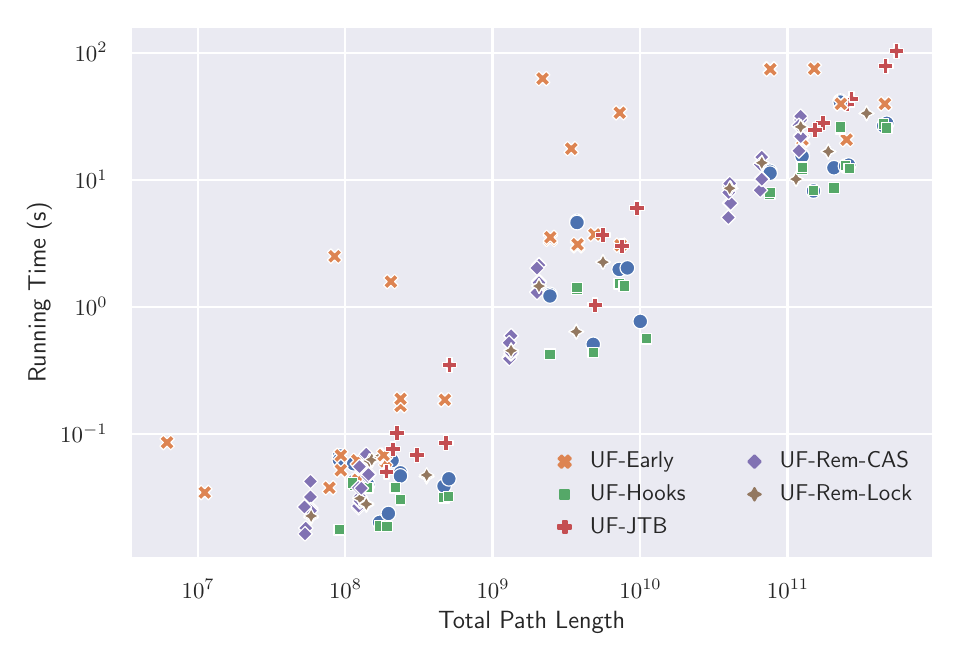}\\ 
	\caption{\small Plot of the \totalpathlen{} traversed
      during algorithm execution (log-scale) vs. running time in
      seconds (log-scale) for all \unionfind{} experiments on all of
      our inputs.
  }\label{fig:total_path_vs_running_time}
\end{figure}

\begin{figure*}[!t]
  \vspace{-1em}
  \centering
  \begin{subfigure}{0.75\textwidth}
  \hspace*{1cm}
  \includegraphics[width=\textwidth]{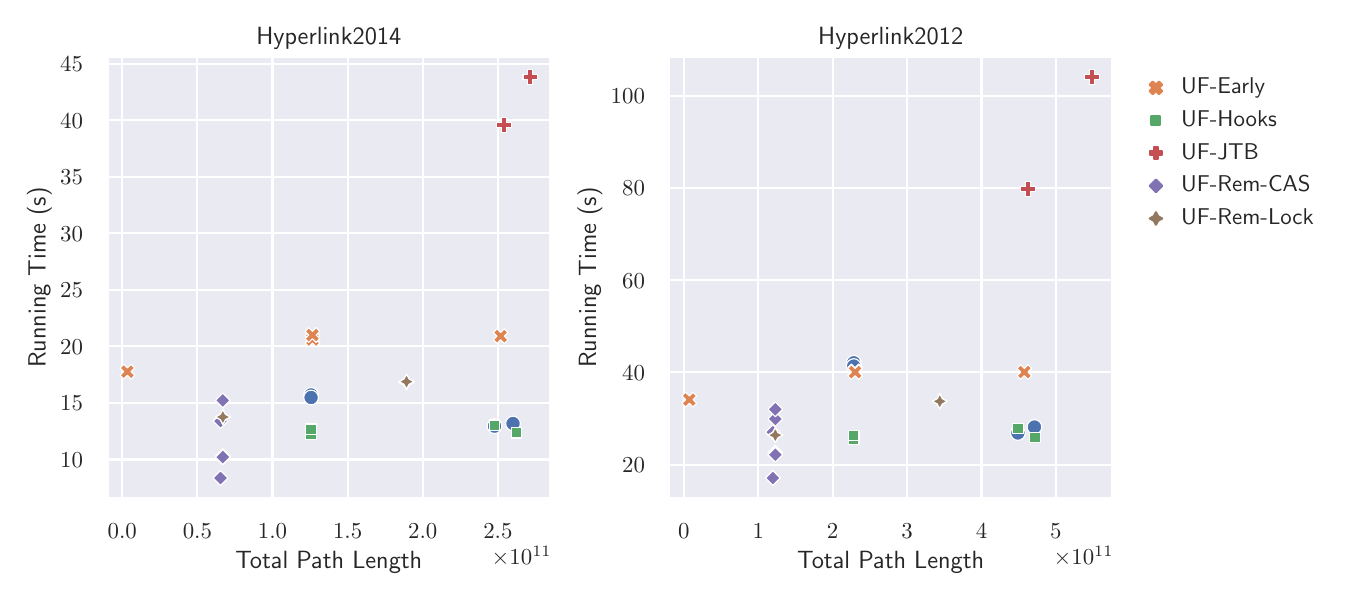}\\ 
	  \caption{\small Plot of the \totalpathlen{} traversed
      during algorithm execution (log-scale) vs. running time in
      seconds (log-scale) for all \unionfind{} experiments on the
      Hyperlink2014 and Hyperlink2012 graphs.
  }\label{fig:hyperlink_total_path}
  \end{subfigure}\\
  \begin{subfigure}{0.75\textwidth}
  \hspace*{1cm}
    \includegraphics[width=\textwidth]{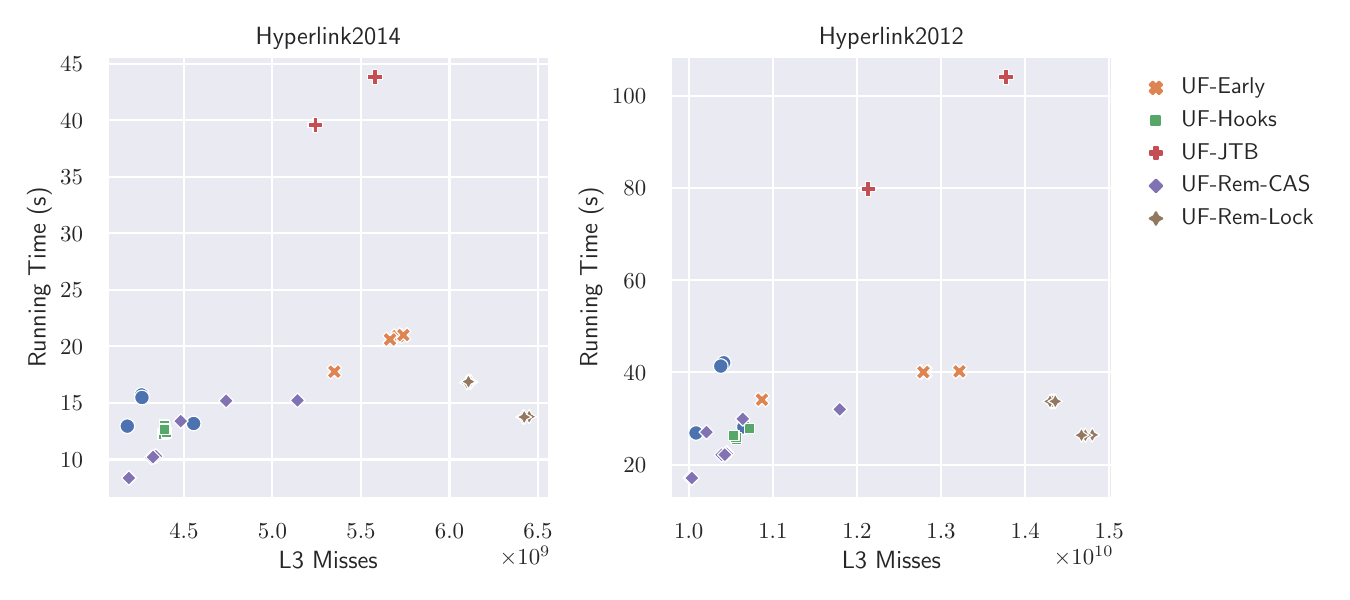}\\ 
    \caption{\small Plot of the number of LLC misses during algorithm
      execution vs. running time in seconds for all \unionfind{}
      experiments on the Hyperlink2014 and Hyperlink2012 graphs.
    }\label{fig:hyperlink_l3_misses}
  \end{subfigure}\\
  \begin{subfigure}{0.75\textwidth}
  \hspace*{1cm}
    \includegraphics[width=\textwidth]{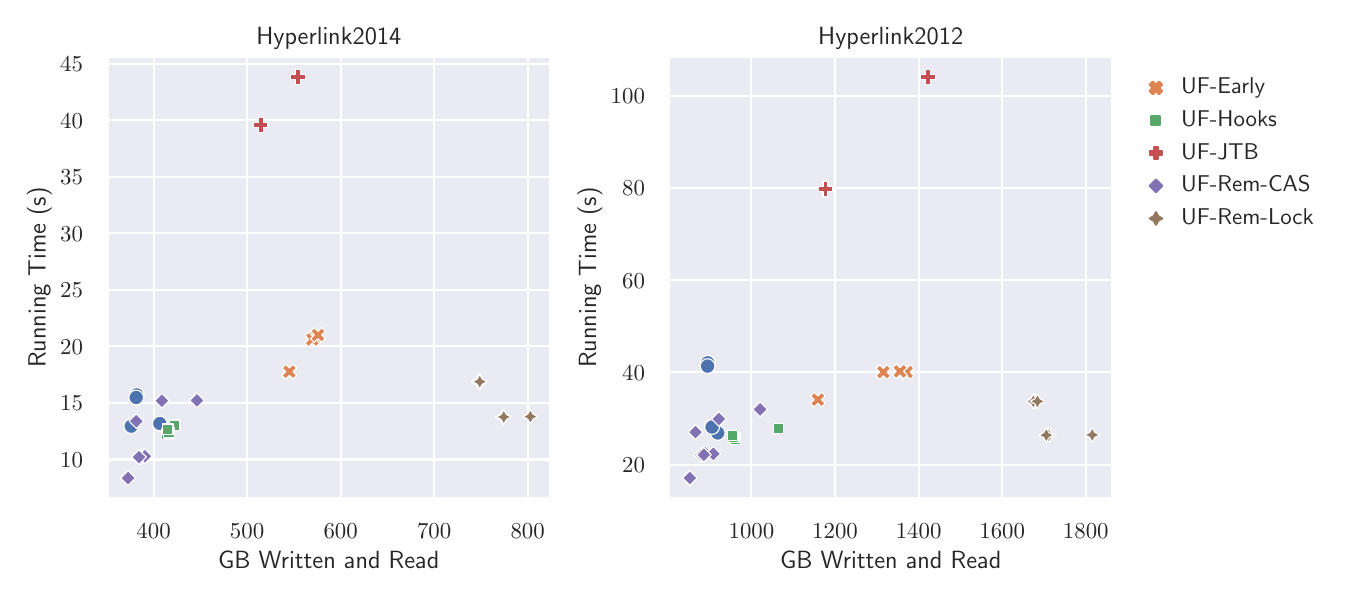}\\ 
    \caption{\small Plot of the total GB written and read during
      execution vs. running time in seconds for all \unionfind{}
      experiments on the Hyperlink2014 and Hyperlink2012 graphs.
    }\label{fig:hyperlink_bytes}
    \vspace{0.5em}
  \end{subfigure}
  \caption{Performance statistics for the Hyperlink2014 and
  Hyperlink2012 graphs.}
\end{figure*}

\myparagraph{\totalpathlen{} (TPL)}
Figure~\ref{fig:total_path_vs_running_time} plots the
\defn{\totalpathlen{} (TPL)} statistic vs. the total running time
across all graphs for different \unionfind{} variants (both axes are
shown in log-scale).  The plot shows that the TPL is highly relevant
for predicting the running time---the TPL has a Pearson
correlation coefficient of 0.738 with the running time (the
MPL has a much weaker coefficient of 0.344).
Figure~\ref{fig:hyperlink_total_path} plots the same statistic for the
Hyperlink2014 and Hyperlink2012 graphs, where we see a similar trend.
The poor performance of the \jayanti{} algorithms under both find
configurations is explained largely by the large total
path lengths traversed by these algorithms (the algorithm with larger
TPL and running time was always for the \findnaive{} implementation).
We note that the \unionearly{} algorithm exhibits a single data point
which experiences significantly lower TPL than any of the other
algorithms. This data point is for the version of the \unionearly{}
algorithm combined with the \findnaive{} strategy. Since the algorithm
does not perform any finds, the only contribution to the path length
is from the path length traversed by the $\mathsf{Union}$ method.
Despite this fact, it is interesting that the \unionearly{} code with
the \findnaive{} option incurs such a low TPL.  We note that
\findnaive{} is always the fastest option for this algorithm.

\begin{figure}[!t]
  \centering
  \hspace{-2em}
    \includegraphics[width=\columnwidth]{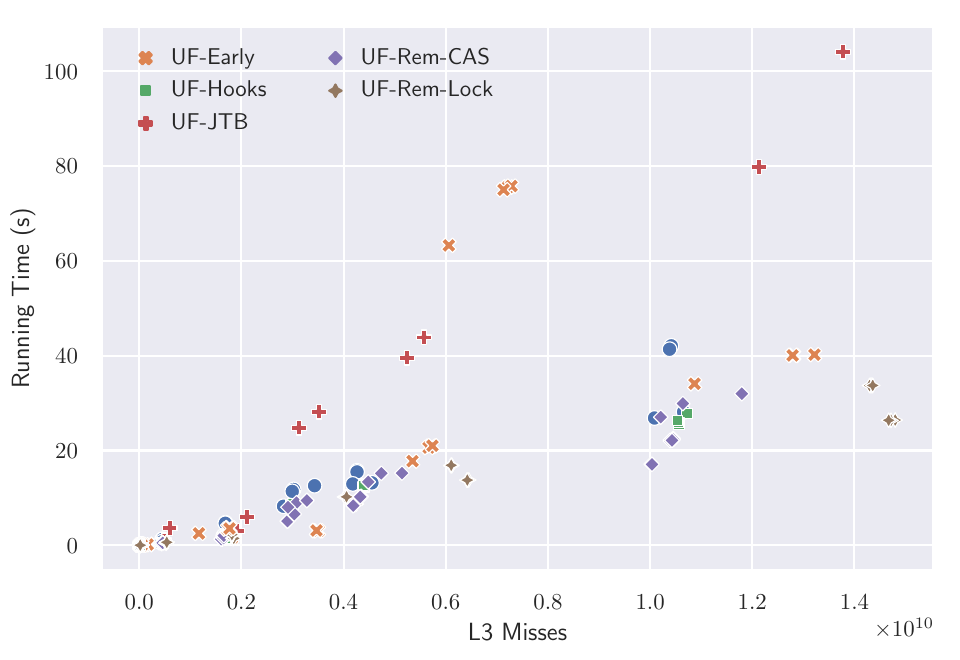}\\ 
    \caption{\small Plot of the number of LLC misses incurred
      during algorithm execution vs. running time in seconds for all
      \unionfind{} experiments on all of our inputs.
  }\label{fig:l3_vs_running_time}
\end{figure}

\myparagraph{LLC Misses}
In Figure~\ref{fig:l3_vs_running_time}, we plot the number of
LLC misses vs. running time across all graphs.
Figure~\ref{fig:hyperlink_l3_misses} plots the same quantities for
the Hyperlink2014 and Hyperlink2012 graphs. We see that the
\unionearly{} implementations, which perform poorly on the
Hyperlink2014 graph incur a large number of LLC misses. The Pearson
correlation coefficient between the number of LLC misses and the
overall running time is 0.797 across all graphs, which is comparable
with the correlation between the TPL and the running
time. Lastly, we observe that the fastest algorithms incur the fewest
number of LLC misses, with \unionremcas{} with the \findsplit{} option
achieving the point in the bottom left corner of the plot for the
Hyperlink graphs.

\begin{figure}[!t]
  \centering
  \hspace{-2em}
    \includegraphics[width=\columnwidth]{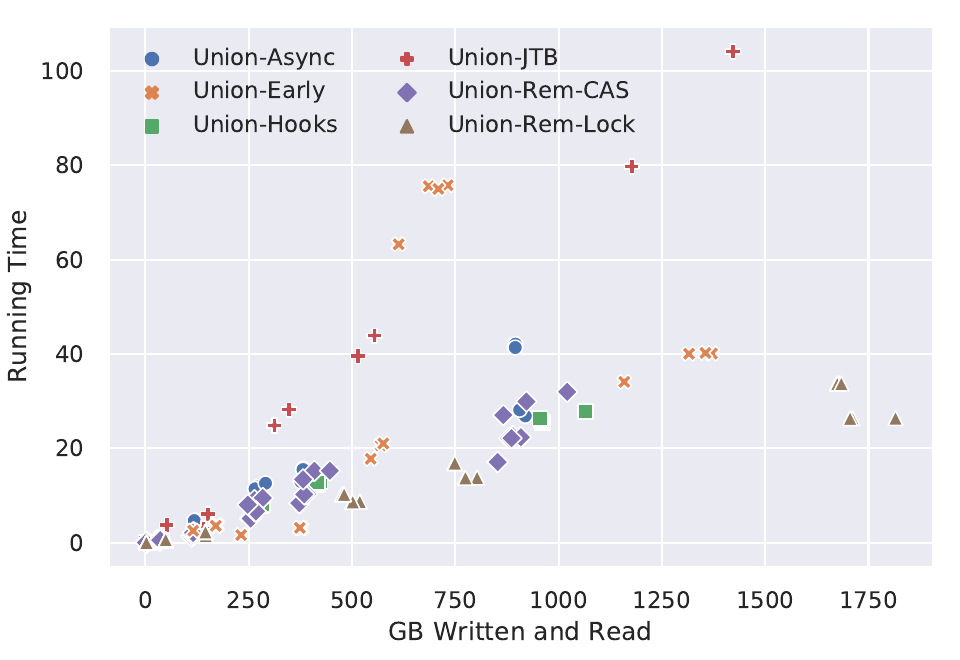}\\ 
    \caption{\small Plot of the total number of bytes
      written and read (in GB) during algorithm execution vs. running
      time for all \unionfind{} experiments on all of our
      inputs.
  }\label{fig:bytes_vs_running_time}
\end{figure}

\myparagraph{Total Bytes Written/Read}
In Figure~\ref{fig:bytes_vs_running_time}, we plot the total gigabytes
written to and read from the memory controller vs. running time
across all graphs. Figure~\ref{fig:hyperlink_bytes} plots the same
quantities for the Hyperlink2014 and Hyperlink2012 graphs.  As one
might expect, the trends appear to be similar to the number of LLC
misses. Perhaps surprisingly, although the \unionremlock{} requires the
largest amount of memory traffic, these algorithms are much faster
than algorithms that read and write fewer bytes across the memory
controllers (\unionearly{} and \jayanti{}). This quantity has a
Pearson correlation coefficient with the running time of 0.774, which
is similar to that of the LLC misses statistic. Finally, we observe
that the \unionremcas{} and \unionasync{} algorithms require
reading and writing very little data.

\subsubsection{Liu-Tarjan Algorithms}

\begin{figure}[!t]
  \centering
  \hspace{-4em}
    \includegraphics[width=\columnwidth]{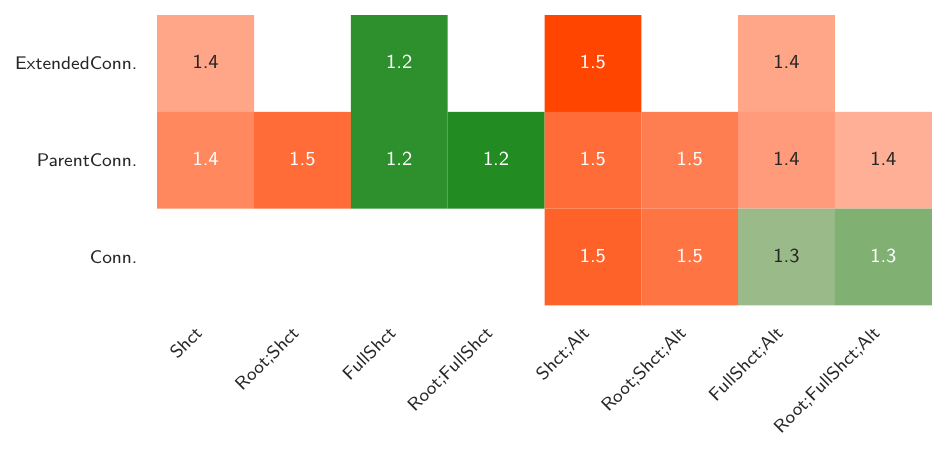}\\ 
  \caption{\small Relative performance of different Liu-Tarjan
    implementations on graphs used in our evaluation of
    \framework{} in the \emph{No Sampling} setting.
  }\label{fig:heatmap_lt_nosample}
\end{figure}

\begin{figure*}[!t]
  \centering
  \hspace{-1em}
    \includegraphics[width=0.75\textwidth]{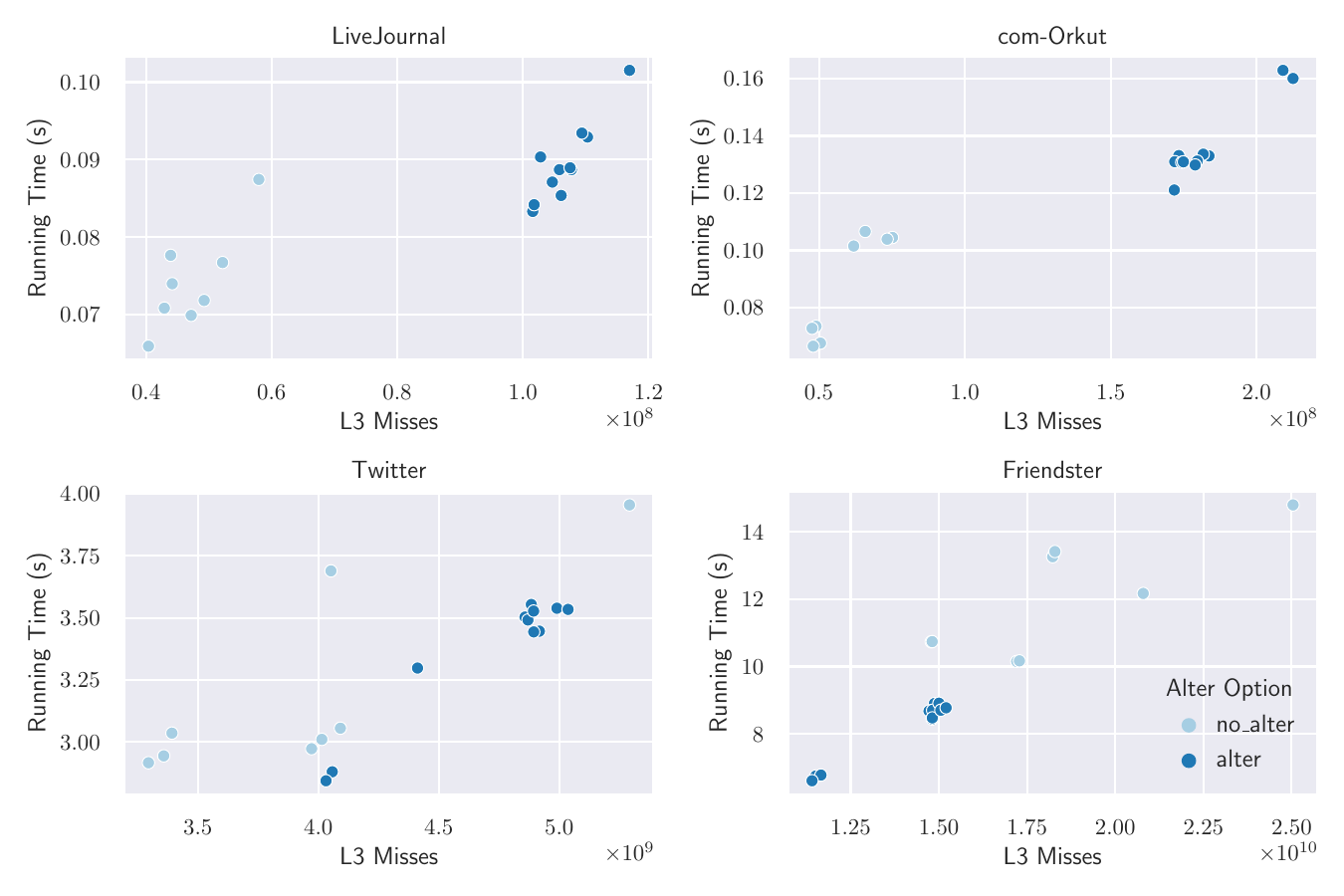}\\ 
    \caption{\small Plot of the number of LLC misses for each
      algorithm execution vs. running time in seconds for the LiveJournal,
      com-Orkut, Twitter, and Friendster graphs.
  }\label{fig:lt_llc_misses}
\end{figure*}

Figure~\ref{fig:heatmap_lt_nosample} shows the relative performance of
different Liu-Tarjan implementations from \framework{} in the \emph{No
Sampling} setting. These results are averaged across all graphs used
in our evaluation, and are computed the same way as in
our \unionfind{} evaluation previously.

We observed that the four fastest Liu-Tarjan variants in this setting
are the new variants we implemented that perform a
$\mathsf{FullShortcut}$, namely $\mathsf{PUF}$, $\mathsf{EUF}$,
$\mathsf{PRF}$, and $\mathsf{CRFA}$.  The remaining methods achieve
between 1.3x--1.5x slowdown on average across all graphs, with the
slowest algorithms on average being those using the $\mathsf{Alter}$
method.

Figure~\ref{fig:lt_llc_misses} analyzes how the number of
LLC misses affects the running time. We see that the algorithm
performance of the Liu-Tarjan variants is significantly determined by
the number of LLC misses (the Pearson correlation coefficient is
0.982).  We note that the algorithms using the $\mathsf{Alter}$ method
are sometimes the fastest Liu-Tarjan variants, which can be seen in
Figure~\ref{fig:lt_llc_misses}.

\subsection{Parallel Connectivity with Sampling}\label{apx:cpu_sample}

This section reports additional results, and performance analyses
about \framework{} in the setting where sampling is performed.
In particular, we evaluate how \framework{} algorithms perform when
combined with the \koutsample{}, \bfssample{}, and \lddsample{}
strategies, and discuss the speedups of the sampling-based algorithms
compared to the performance numbers without sampling presented in
Section~\ref{subsec:static_cpu_nosample}.

\myparagraph{\framework{} Using \koutsample{}}
The second group of rows in Table~\ref{table:static_cpu_all} shows the performance of \framework{}
algorithms using the \koutsample{} scheme on our graph inputs.
We observe that the performance of \framework{} algorithms combined
with \koutsample{} is consistently fast, and produces the fastest
overall algorithm on almost half of our graph inputs (across both
small and large graphs). Like in the no sampling setting, the
\unionfind{} algorithms achieve very high performance, with
\unionasync{} and \unionremcas{} achieving similar performance (their
running times are within a few percent of each other on all of the
graphs in this setting). Unlike in the no sampling setting, the
performance of other \unionfind{} variants is very close to that of
the fastest \unionfind{} variants, with the slowest overall variant on
average being \unionremlock{}. On average, the union-find
implementations achieve a 2.11x speedup over the fastest connectivity
implementations without sampling, with the highest average speedup
obtained by \unionremcas{} (2.25x faster than the fastest algorithms
without sampling, on average).
Our Liu-Tarjan, Shiloach-Vishkin, and \labelprop{} algorithms
also achieve high performance when compared with the fastest
algorithms for each graph. Liu-Tarjan achieves 1.71x average speedup
over the fastest algorithm in the no sampling setting,
Shiloach-Vishkin achieves 1.53x average speedup, and \labelprop{}
achieves 1.68x average speedup. These results indicate that these
algorithms compose efficiently with the \koutsample{} scheme.

Across all of our graphs, the one exception where \koutsample{} does
not provide significant speedups is for the road\_usa graph. The
main reason is that the road\_usa graph has very low average degree,
with most vertices having between 2--5 edges.  Therefore, the cost of
simply going over this small set of edges twice during two-phase
execution outweighs the additional benefit provided by sampling,
since we almost finish computing connectivity after applying
\koutsample{} (see Section~\ref{sec:sampling_eval} for more
experimental analysis of \koutsample{} on this graph).

Next, we compare the effectiveness of \koutsample{} compared with the
fastest results obtained by any combination of sampling and finish
methods in the \framework{} framework. We can see from
Table~\ref{table:static_cpu_all}
that the \unionremcas{} algorithm achieves the most consistent
performance relative to the fastest algorithms, being at most 27\%
slower than the fastest implementation for all graphs, and within 14\%
on average. \unionasync{} and \unionhook{} are similarly competitive,
being at most 30\% and 31\% slower, and within 19\% and 18\% of the
fastest on average, respectively. We conclude that applying one of
these three union-find algorithms in addition to \koutsample{} results
in consistent high performance.

To summarize, we find that \koutsample{} is an effective and robust
sampling scheme that consistently works across a variety of different
graph types, and algorithm types (\unionfind{} variants, Liu-Tarjan
variants, and Shiloach-Vishkin) to produce algorithms that are within
1.12x on average of the fastest running times obtained by any
algorithm in our framework. Finally, for the Hyperlink2012 graph, the
largest publicly-available real-world graph, the \unionremcas{}
algorithm combined with \koutsample{} achieves the fastest time out of
all algorithm combinations in our framework, running in \emph{8.2
seconds} in total. This is a \emph{3.14x improvement} over the fastest
known result for this graph, previously achieved by the Graph Based
Benchmark Suite~\cite{DhBlSh18} using the work-efficient connectivity
algorithm by Shun et al.~\cite{SDB14}, and is competitive with two
empirical lower-bounds on the performance of any connectivity
algorithm (analyzed in Section~\ref{sec:edge_gather}).

\myparagraph{Performance Using \bfssample{}}
The third group of rows in Table~\ref{table:static_cpu_all} show the performance of \framework{}
algorithms using the \bfssample{} scheme on our graph inputs.
We can immediately see that the performance of this method depends
significantly on the diameter of the input graph (listed in
Table~\ref{table:sizes}) as the number of rounds in parallel BFS is
proportional to the diameter---for example, the performance of these
methods on the road\_usa graph is significantly worse than the
unsampled running times (roughly 60--70x slower). One interesting
exception is for \labelprop{}, where the algorithm is 5.1x faster
than the unsampled version, since the \labelprop{} algorithm can
avoid processing the high-diameter component identified by the BFS.

The performance on the remaining graphs is highly competitive: the
fastest running time using \bfssample{} achieves between 1.2--4.9x
speedup over the fastest running time without sampling, and 2.5x
speedup on average. Like with \koutsample{}, the fastest running times
using \bfssample{} are obtained using \unionfind{} variants, with
\unionremcas{} running within 3\% of the fastest variant using
\bfssample{}. \unionasync{} and \unionhook{} are similarly
competitive, running between 1.1--4.7x faster and 1.1--4.8x faster
than the fastest algorithm without sampling (both achieve 2.4x speedup
on average). Finally, we observe that for the three largest graphs,
\unionremcas{} consistently achieves the fastest performance out of
all variants using \bfssample{}, with \bfssample{} achieving the
fastest overall running time for \framework{} on the Hyperlink2014
graph.

Our Liu-Tarjan, Shiloach-Vishkin, and \labelprop{} algorithms
compose efficiently with \bfssample{}, achieving significant speedups
over the versions of these algorithms that do not use sampling. On
graphs excluding road\_usa (where as discussed, \bfssample{} performs
poorly), the Liu-Tarjan algorithms achieve between 4.4--28x speedup
over the unsampled versions, achieving 12.3x speedup over the
unsampled versions on average. Our Shiloach-Vishkin algorithm also
achieves strong speedups over the unsampled versions, achieving
between 3.87--12.4x speedup, and 7.58x speedup on average. Our
\labelprop{} algorithm achieves between 3.41--13.1x speedup over
the unsampled \labelprop{} algorithm, and achieves 7x speedup on
average.

To summarize, we find that \bfssample{} is a robust sampling strategy
that achieves very high performance when combined with our
\framework{} algorithms on \emph{low-diameter graphs exhibiting a
massive connected component}, such as social networks and Web graphs.
On high-diameter graphs such as road networks, the performance results
suggest that a different sampling scheme, such as \koutsample{} would
be a better choice, unless the \labelprop{} algorithm is used in which
case \bfssample{} is a good choice (the main source of
performance improvement for \labelprop{} is having the BFS fully
cover the largest component, since the BFS can be accelerated using
direction-optimization~\cite{Beamer12}).

\myparagraph{Performance Using \lddsample{}}
Finally, we consider the performance of \framework{} algorithms using
the \lddsample{} scheme on our graph inputs (the fourth group of rows in
Table~\ref{table:static_cpu_all}).
First, we consider the high-diameter road\_usa graph. Compared to
\bfssample{}, the performance numbers using \lddsample{} are
significantly better (between 22--25x), except for \labelprop{}
where the time using \lddsample{} is 5.5x slower since the
massive high-diameter component is still slowly traversed by a large
number of costly iterations, even after applying \lddsample{}. The
reason for the improvement on other algorithms is since \lddsample{}
only requires a few rounds to compute an LDD, whereas a BFS on this
network requires thousands of rounds to traverse the high-diameter
component in this graph.

Focusing now on low-diameter graphs, the performance of \lddsample{}
achieves significant speedups over the \framework{} algorithms in the
no sampling setting. Our \unionremcas{} implementation achieves
between 0.98--4.67x speedup over the fastest connectivity
implementation in the no sampling setting, and 2.3x speedup on
average. \unionasync{} and \unionhook{} also achieve high performance,
achieving between 0.96--4.6x and 0.95--4.5x speedup (2.34x and 2.28x
speedup on average), respectively. The remaining union-find
implementations also achieve high performance under this scheme.
Overall, we see that \unionremcas{} is a robust choice in this
setting, and is either the fastest implementation, or within 30\% of
the fastest across all low-diameter graphs. On the three largest
graphs, \unionremcas{} is either the fastest algorithm, or within 22\%
of the fastest running time, and for the largest graph, this algorithm
combined with \lddsample{} is within 1\% of the fastest running time.

Our Liu-Tarjan, Shiloach-Vishkin, and \labelprop{} algorithms
achieve good speedups over the unsampled versions on graphs other than
road\_usa. The Liu-Tarjan algorithms achieve between 3.5--21.4x
speedup over these algorithms in the unsampled versions, and 10x
speedup on average. The Shiloach-Vishkin implementation achieves
between 3.1--27.3x speedup over the unsampled versions, and 11.6x
speedup on average. Finally, the \labelprop{} implementation also
achieves good speedup, ranging from 2.1--9.7x speedup, and 5.7x
speedup on average.

To summarize, we find that \lddsample{} offers a more robust sampling
mechanism compared with \bfssample{} with performance that achieves
significant speedups over the \framework{} algorithms that do not use
sampling. Furthermore, its performance is not overly degraded on
high-diameter networks, and achieves excellent performance on
low-diameter social networks and Web graphs. We perform a thorough
analysis of the parameters to the \lddsample{} scheme in
Section~\ref{sec:sampling_eval}.

\begin{figure}[!t]
  \centering
  \hspace{-2.5em}
    \includegraphics[width=\columnwidth]{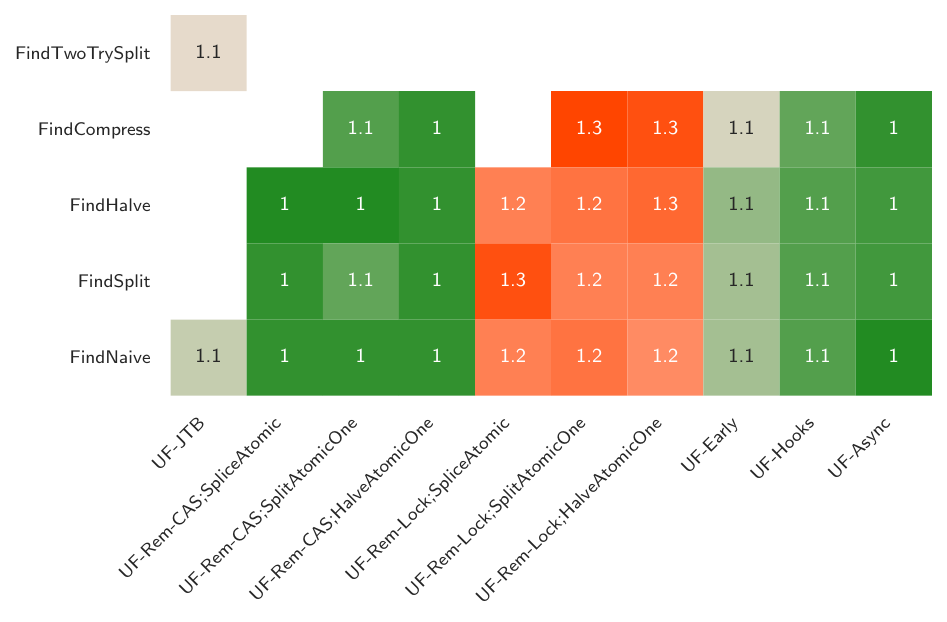}\\ 
    \caption{\small Relative performance of \unionfind{} implementations, averaged across all inputs with \koutsample{}.
  }\label{fig:heatmap_uf_kout}
\end{figure}
\begin{figure}[!t]
  \centering
  \hspace{-2.5em}
    \includegraphics[width=\columnwidth]{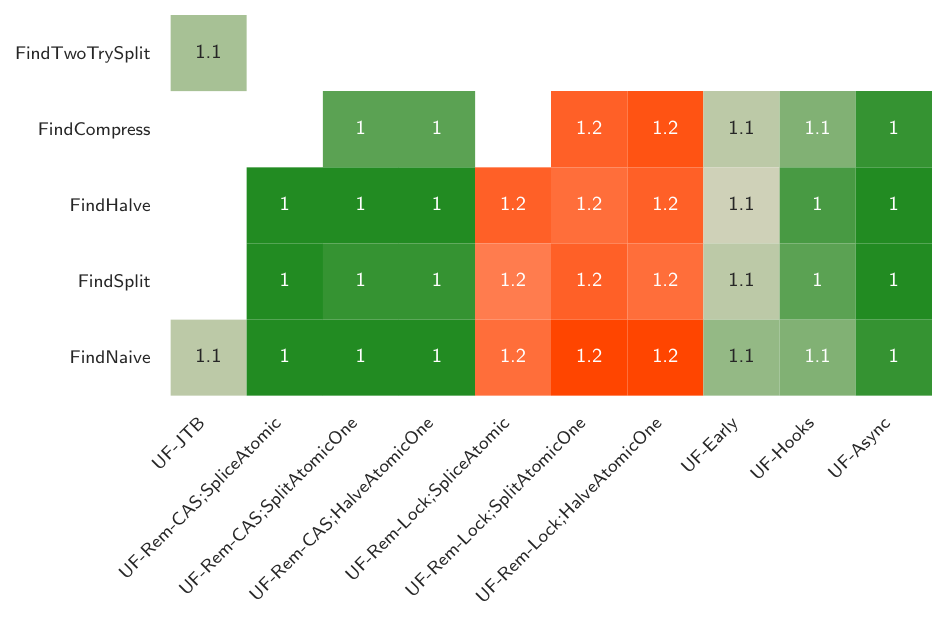}\\ 
    \caption{\small Relative performance of \unionfind{} implementations, averaged across all inputs with \bfssample{}.}\label{fig:heatmap_uf_bfs}
\end{figure}

\myparagraph{Union-Find Performance using Sampling}
Figures~\ref{fig:heatmap_uf_kout}, \ref{fig:heatmap_uf_bfs}, and
\ref{fig:heatmap_uf_ldd} show heatmaps of the relative performance of
different \unionfind{} variants when using \koutsample{}, \bfssample{},
and \lddsample{}, respectively.  Like
Figure~\ref{fig:heatmap_uf_nosample}, the running times are aggregated
across all of our graph inputs.
Unlike Figure~\ref{fig:heatmap_uf_nosample}, the relative running times of
each algorithm are significantly closer together for all three
sampling schemes, ranging from 1--1.3x slower than the fastest
\unionfind{} variant on average for \koutsample{}, 1--1.2x for
\bfssample{}, and 1--1.3x for \bfssample{}. As mentioned previously,
the fastest \unionfind{} variants are just 1.12x slower on average than
the fastest \framework{} algorithm across all sampling schemes.
Therefore, we conclude that applying any \unionfind{} variant in
conjunction with our sampling schemes results in high performance
across all of our graph inputs.

\begin{figure}[!t]
  \centering
  \hspace{-2.5em}
    \includegraphics[width=\columnwidth]{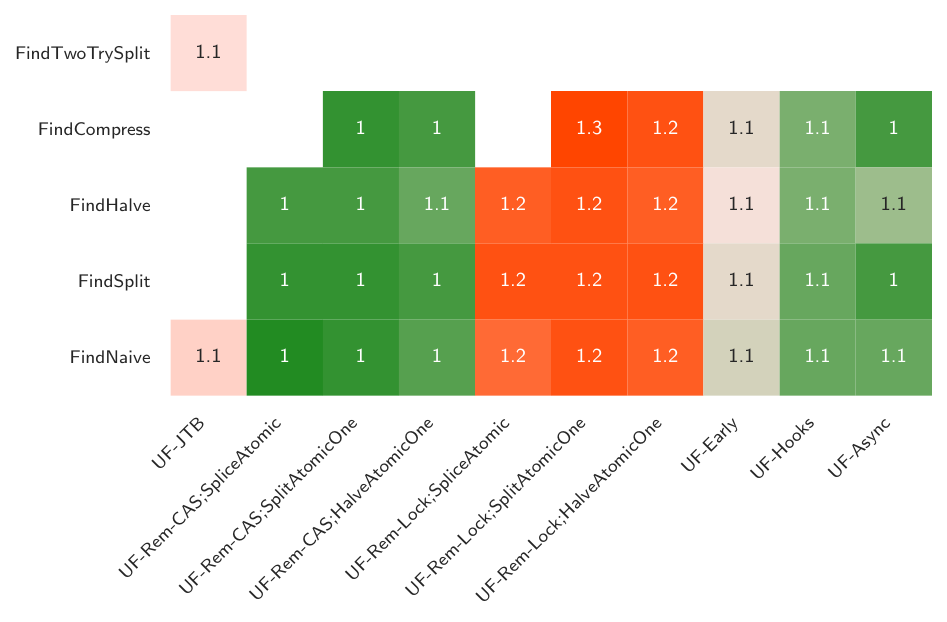}\\ 
    \caption{\small Relative performance of \unionfind{} implementations, averaged across all inputs with \lddsample{}.}\label{fig:heatmap_uf_ldd}
\end{figure}

In summary, all three schemes result in meaningful performance
improvements over the unsampled versions, and significantly outperform
existing state-of-the-art connectivity results using \framework{}
algorithms.
\subsection{\framework{} vs. State-of-the-Art}\label{apx:static_cpu_comparison}

In this section, we provide a comparison between
\framework{} and a set of state-of-the-art publicly-available
connectivity implementations. The results we analyze can be found in
Table~\ref{table:static_cpu_all}.

{\myparagraph{Comparison with \bfscc{}~\cite{ShunB14}} First, we
  observe that the \bfscc{} algorithm has highly variable performance
  that depends on the input graph diameter. In the no sampling setting
  for \framework{}, it ranges from being 92x slower than the
  fastest \framework{} algorithm for the road\_usa graph, to up to
  1.81x and 1.86x faster for the com-Orkut and Twitter graphs,
  respectively. Compared to \framework{} algorithms using sampling on
  instances other than the road\_usa graph, where it performs
  understandably poorly, our implementations are still always faster,
  ranging from between 1.22--21.7x faster, and 6.46x faster on
  average.

{\myparagraph{Comparison with Shun et al.~\cite{SDB14}} We observe that the
  provably work-efficient LDD-based algorithm from Shun et
  al.~\cite{SDB14} achieves more consistent performance across all
  graphs, including both low and high-diameter networks. Our
  implementations without sampling are between 1.45--19.4x faster,
  and 5.6x faster on average. Our implementations with sampling are
  between 3.1--27.5x faster, and 9.3x faster on average. We note that
  prior to this paper, this algorithm, implemented as part of the
  Graph Based Benchmark Suite (GBBS) was the fastest reported running time
  for the Hyperlink2012 graph, running in 25 seconds~\cite{DhBlSh18}.
  \emph{The fastest \framework{} algorithm, \unionremcas{} with
    \splitone{} improves on this result, even without using any
  sampling, running in just under 14 seconds}.

	{\myparagraph{Comparison with MultiStep~\cite{Slota14}} The
MultiStep method implements a hybrid scheme that performs
the \bfssample{} scheme combined with \labelprop{}. The MultiStep
algorithm performs $O(\mathsf{diam}(G)m + n)$ work in the worst case,
and requires $O(\mathsf{diam}(G)\log n)$ depth, where
$\mathsf{diam}(G)$ is the graph diameter.  The performance of the
algorithm is highly variable, especially on high-diameter graphs like
road\_usa, where it is 2.2x slower than our \labelprop{}
implementation. Compared with the fastest \framework{} implementation
without sampling, our implementations are significantly faster, even
without sampling. Our codes are between 1.95--21.4x faster across all
low-diameter graphs, and 10.1x faster on average. Compared with the
fastest \framework{} implementation with sampling, our codes are
between 9.6--30.3x faster across all low-diameter graphs, and 18.6x
faster on average. Their code performs significantly worse than ours
on road\_usa, running 1,057x slower than our fastest code. We note that
their algorithm was unable to successfully run on the Twitter graph.
Compared to the Slota et al.~implementation, our implementation of the
MultiStep approach (\bfssample{} combined with \labelprop{}) is
between 3.9--22x faster than their implementation, and 12.8x faster on
average. The MultiStep code implements the direction-optimization, but
we observed that their BFS performance is significantly slower than
BFS performance in \framework{}, which could be due to a
cache-optimized version of the edge traversal from the Graph Based
Benchmark Suite used in our implementation.

{\myparagraph{Comparison with Galois~\cite{Nguyen2013}} The connectivity
algorithm implemented in Galois uses \labelprop{} in conjunction with
a priority-based asynchronous scheduler~\cite{Nguyen2013}. The
algorithm is not work-efficient, and we are not aware of any
theoretical upper-bounds on its work when using Galois' priority-based
scheduler. Excluding the road\_usa graph where their code performs
very poorly using \labelprop{}, our codes without sampling are between
1.78--3.69x faster than theirs (2.36x faster on average), and our
codes using sampling are between 2.84--12.6x faster than theirs (7.27x
faster on average). Our codes are 2.17x faster than their \emph{Async}
algorithm on the road\_usa graph (we tried other implementations from
Galois and report times for the fastest one, since their \labelprop{}
implementation performed poorly on this graph).  Finally, their slow
performance on the Twitter graph is likely due to the fact that the
Galois implementation requires multiple rounds, where a large fraction
of the vertices are active when performing \labelprop{}.

{\myparagraph{Comparison with Patwary et al.~\cite{PatwaryRM12}} The
  connectivity algorithm by Patwary et al. is a
  concurrent version of Rem's algorithm. The algorithm is the
  same as the \unionremlock{} implementation in our framework.
  However, we find that our implementations are significantly faster
  than theirs, and that both of our implementations of Rem's
  algorithms are competitive with, and often outperform theirs. Our
  fastest implementations without sampling achieve between 1.27--2.87x
  speedup over their implementation (1.99x speedup on average), and
  our fastest implementations with sampling achieve between
  2.43--6.28x speedup over their implementation (4.4x speedup on
  average). We note that they also provide a lock-free version of
  Rem's algorithm, but we observed that this algorithm does not always
  produce correct answers, and thus do not report results for it.

{\myparagraph{Comparison with GAPBS~\cite{BeamerAP15}}}
Finally, we compare \framework{} with the Shiloach-Vishkin and
Afforest implementations from the GAP Benchmark Suite (GAPBS). The
Shiloach-Vishkin implementation from GAPBS suffers from an
implementation issue that can cause it to perform $O(mn)$ work
and $O(n)$ depth in the worst case (under an adversarial
scheduler), but we note that this behavior does not seem to arise in
practice. The Afforest implementation is implemented jointly with the
authors of the Afforest paper by Sutton et al.~\cite{Sutton2018}, and
was always faster than the original Afforest
implementation~\cite{Sutton2018}. Our fastest implementations
significantly outperform their Shiloach-Vishkin implementation. Our
fastest implementations without sampling are between 3.67--17.9x
faster (9.55x faster on average), and our fastest implementations with
sampling are between 3.67--61.2x faster (27.1x faster on average). We
observe that the Afforest algorithm is sometimes faster than our
algorithms without sampling (note that Afforest does perform
sampling). Our fastest implementations without sampling are between
0.33--4.17x faster than their Afforest implementation (2.08x faster on
average). Our fastest implementations with sampling are between
1.53--8.55x faster than their Afforest implementation (3.9x faster on
average). We perform an in-depth evaluation of the Afforest sampling
scheme in Section~\ref{sec:sampling_eval}.

\subsection{Streaming Parallel Graph Connectivity}\label{apx:streaming_cpu}
In this appendix, we provide deferred experimental results on our
\framework{} streaming algorithms. We provide more experimental
results showing tradeoffs between batch size and throughput, and
addition latency plots.

\myparagraph{Additional Observations about Throughput Experiments}
We observe that the throughput of all algorithms, even on the
smallest graph that we consider is at least hundreds of millions of updates
per second for all graphs. In general, the throughput appears to be
higher for larger graphs, although some graphs, such as com-Orkut have
very high throughput despite being quite small, which could be due to
cache effects, since the vertex data fits within LLC.
Another source of high performance is that the updates within a batch
are left unpermuted in this experiment, and are ordered
lexicographically. This ordering provides good data locality when
performing work, since the consecutive edges that are applied are
likely to contain vertices whose IDs are close together in the vertex
ordering.

We ran an experiment to measure the effect of the edge ordering on
throughput across different input graphs and found that the
performance difference between the two methods scales almost linearly
with the number of LLC cache misses incurred by the algorithm. For
example, for the Twitter graph, applying the edges of the graph in a
single batch for the \unionremcas{} algorithm with \splitone{} without
and with permuting the batch results in 3.85 billion and 2.23
billion edges per second, respectively, or a 1.72x higher throughput
in the unpermuted setting. The number of LLC misses measured in these
settings is 1.10 billion and 2.05 billion, respectively, which is 1.8x
fewer misses in the unpermuted setting. The trend is similar for other
graphs, with the performance gap explained by the larger number of
cache misses in the permuted setting. We feel that enforcing that
batches are permuted is somewhat unrealistic, since our expectation is
that real-world data streams possess a large degree of locality
(intuitively many of the same vertices are updated, and updates to the
same vertex are likely to occur in close succession).

\myparagraph{Throughput vs. Batch Size}
Figure~\ref{fig:batch_size_vs_throughput_extra} shows the throughput
obtained by our algorithms as a function of the batch size for a
subset of our input graphs (the trends for graphs not shown in the
figure are very similar). We observe that all \framework{} algorithms
achieve consistently high throughput on very small to very large
batch sizes. Our fastest union-find algorithms achieve a throughput of
over 10 million updates per second at a batch size of 100, and exceed
100 million updates per second for batches of size 1,000 on all
graphs.

\begin{figure*}[!t]
\begin{subfigure}{0.35\textwidth}
  \centering
  \includegraphics[width=\textwidth]{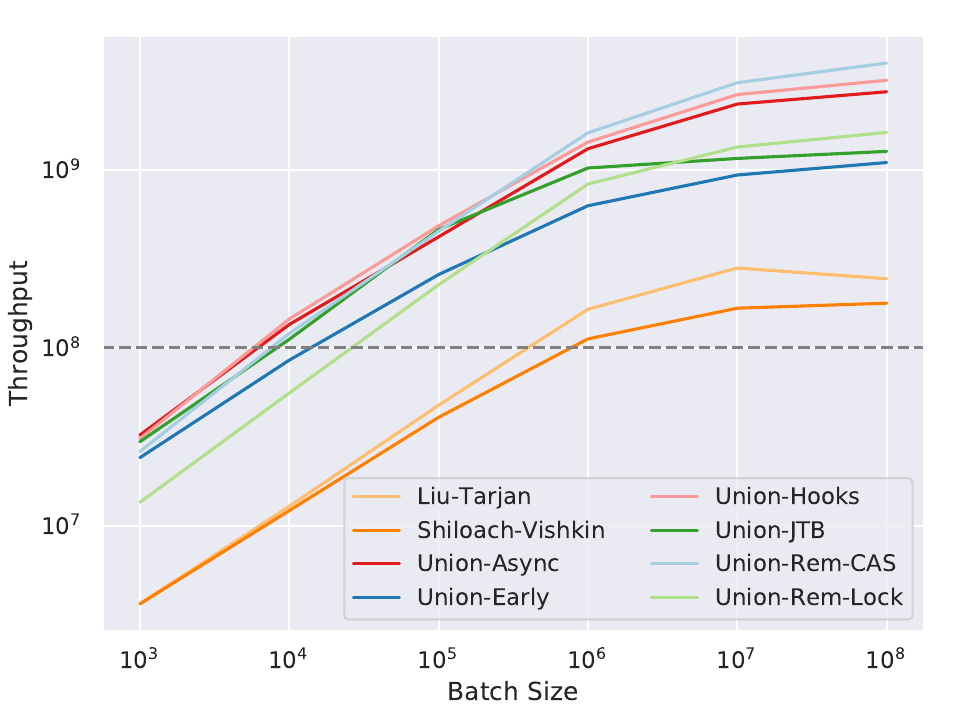}
  \caption{Throughput vs. batch-size for road\_usa} 
\end{subfigure}
\begin{subfigure}{0.35\textwidth}
  \centering
  \includegraphics[width=\textwidth]{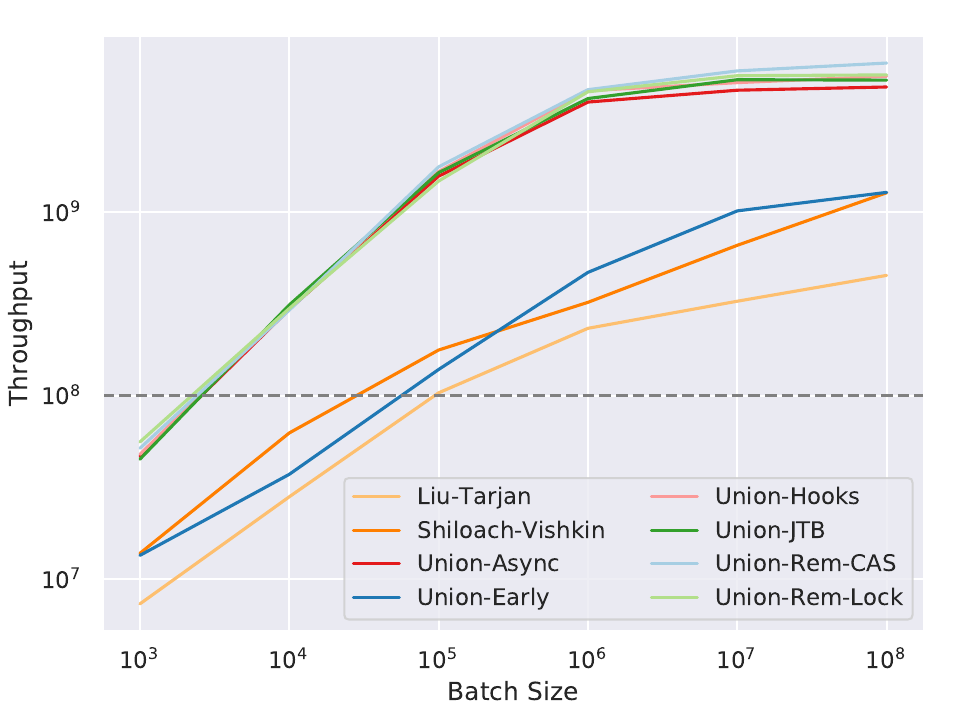}
  \caption{Throughput vs. batch-size for Orkut} 
\end{subfigure} \\
\begin{subfigure}{0.35\textwidth}
  \centering
  \includegraphics[width=\textwidth]{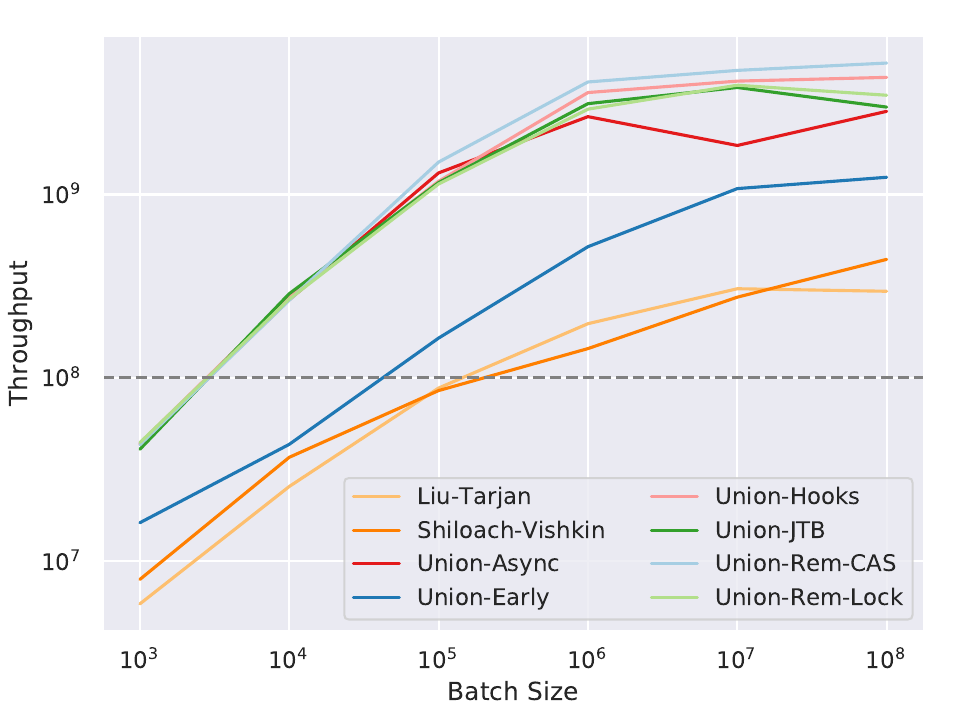}
  \caption{Throughput vs. batch-size for LiveJournal} 
\end{subfigure}
\begin{subfigure}{0.35\textwidth}
  \centering
  \includegraphics[width=\textwidth]{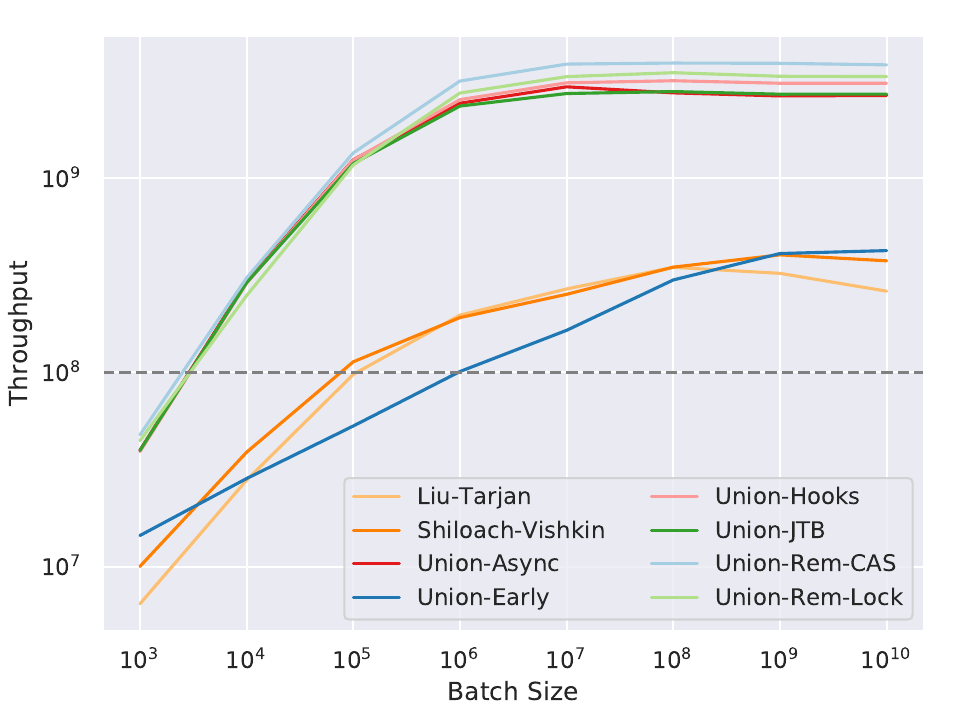}
  \caption{Throughput vs. batch-size for Twitter} 
\end{subfigure}
  \caption{\small Throughput vs. batch size of different union-find
  implementations on different graphs. The experiment is run on a
  72-core machine with 2-way hyper-threading enabled.
}\label{fig:batch_size_vs_throughput_extra}
\end{figure*}

\myparagraph{Mixed Inserts and Queries Throughput}
In these experiments we study how the ratio of insertions to queries
affects the overall throughput. We use the \unionremcas{} algorithm as
a case study for this experiment, and evaluate the performance of
different variants of this algorithm, varying the ratio of insertions
to updates.
To generate an update stream with an insert-to-query ratio of $x$,
we generate $1/x$ many queries on random pairs of vertices, and create
a batch permuting the original updates and queries together. Another
alternate experiment would be to fix a particular batch size, and then
vary the ratio of queries to updates. We tried this strategy but
observed that this has the undesirable effect of varying the number of
connected components, which affects both query and update performance
and thus makes it difficult to interpret the results. Instead, we
choose to fix the updates, and generate random pairs of vertices to
use as queries. Unlike the experiments in the previous section, we
randomly permute the batch so as to accurately reflect the cost of the
\unionremcas{} variants which must reorder a batch for correctness.

\begin{figure*}[!t]
\begin{subfigure}{0.4\textwidth}
  \centering
  \includegraphics[width=\textwidth]{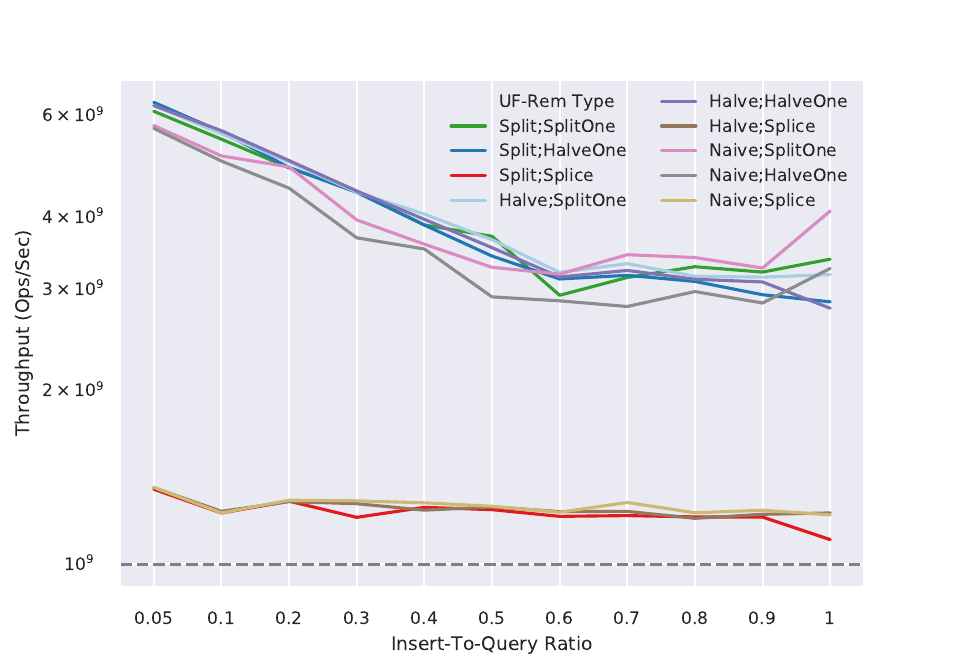}
\end{subfigure}
\begin{subfigure}{0.4\textwidth}
  \centering
  \includegraphics[width=\textwidth]{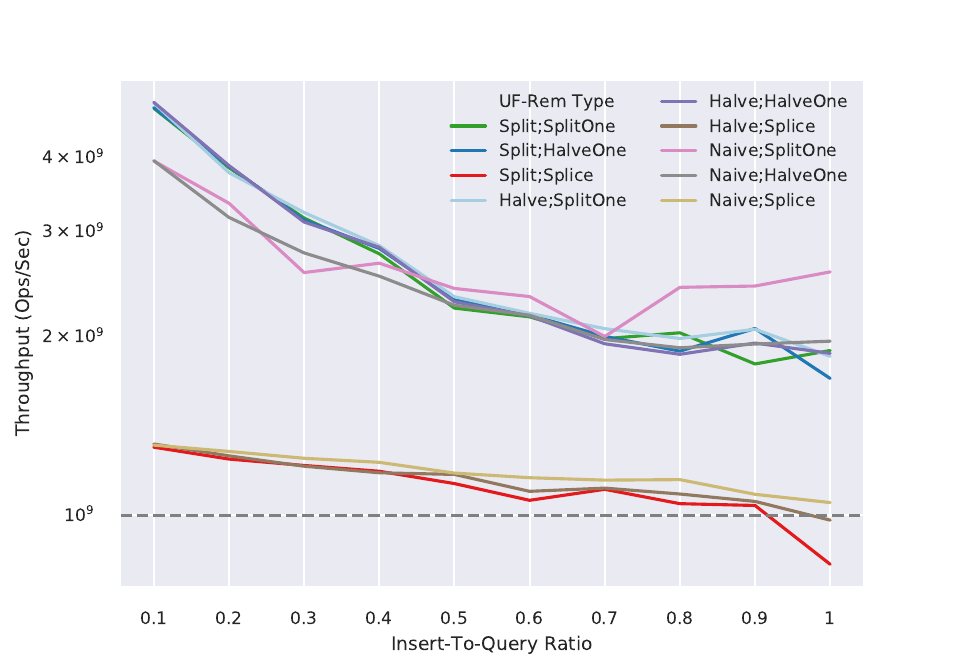}
\end{subfigure}
  \caption{\small Throughput of different \unionremcas{}
  implementations on the com-Orkut and LiveJournal graphs, plotted
  against the ratio of insertions to queries in a batch. The
  experiment is run on a 72-core machine with 2-way hyper-threading
  enabled.
  }\label{fig:streaming_union_rem_cas_orkut}
\end{figure*}

Figure~\ref{fig:streaming_union_rem_cas_orkut} shows the result of
this experiment on com-Orkut and LiveJournal graphs (the results for
other graphs are similar). We observe that the throughput of the
algorithms is higher when the ratio between the inserts-to-queries is
small, and gradually drops as the ratio approaches $1$. We also
observe that the ordering of the algorithms (the relative throughputs)
changes as the ratio of inserts-to-queries changes. When there are
very few inserts and mostly queries, the fastest algorithms are those
that perform additional compression within the query, since queries
can help subsequent queries save work by compressing the paths
traversed in the union-find trees. On the other hand, when the ratio of
inserts-to-queries is large, the algorithms that perform best are
those which perform no additional compression during the finds. On
both graphs, after a ratio of about 0.6--0.7, the fastest algorithm is
the variant which performs no additional compression on a find, and
uses the \splitone{} option (like in the static setting).

\begin{figure*}[!t]
  \hspace{1em}
  \includegraphics[width=0.8\textwidth]{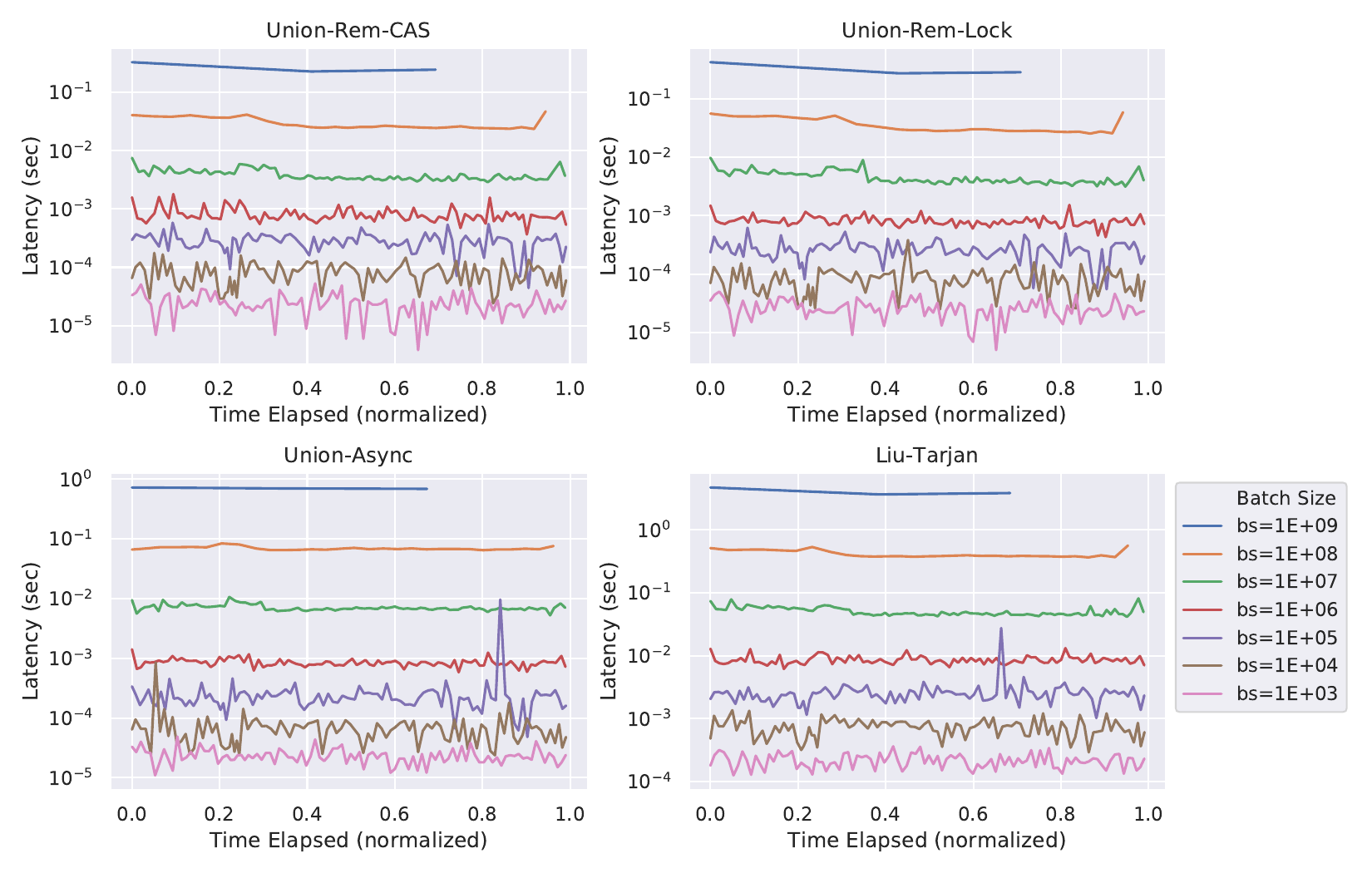}
  \caption{\small Latency of batch operations (seconds) plotted for
    varying batch sizes for the \unionremcas{}, \unionremlock{},
    \unionasync{}, and \liutarjan{} algorithms on the Hyperlink2012
    graph. The running time of each operation sequence is normalized
    by the duration of the experiment at that batch size. The $y$-axis is in log-scale. All
    experiments are run on a 72-core machine with 2-way
    hyper-threading enabled.
  }\label{fig:all_hyperlink_latency}
\end{figure*}

\myparagraph{Additional Latency Plots}
Figure~\ref{fig:all_hyperlink_latency} shows the latency of different
\framework{} implementations while processing an update-only sequence
containing a 10\% sample of the Hyperlink2012 graph. We observe that
the latency for all of our algorithms is highly consistent, with the
median time per batch for different batch sizes always being within
1--2\% of the average batch time across all algorithms and graphs. The
update time, which is plotted in log-scale, grows linearly with the
batch size which is also linearly increasing in log-scale.  Since the
total time taken to process the graph at different batch sizes is
different, we rescale the $x$-axis for each experiment to lie in the
range $[0, 1)$, and subsample the streams to plot a manageable number
of operations. We plot the operations as occurring at their start
time.

\subsection{Sampling Evaluation}\label{sec:sampling_eval}
We now discuss the empirical performance of the sampling methods
presented in this paper. We focus on analyzing (i) the running time of
each sampling method; (ii) the fraction of the vertices covered by the
largest component the sampling method identifies; and (iii) the number
of inter-component edges remaining after running the sampling method.

\begin{table}[!t]
\footnotesize
\centering
\begin{tabular}[t]{l | c|c|c|c|c|c}
  \toprule
  {\bf Graph} & $\mathsf{BFS}$ (s) & $\mathsf{BFS}$ Cov. & $\mathsf{BFS}$ IC. & $\mathsf{LDD}$ (s) & $\mathsf{LDD}$ Cov. & $\mathsf{LDD}$ IC. \\
  \midrule
  {RU}          &2.67     & 100\%   & 0\%          &0.0743     &100\%     &0\% \\ 
  {LJ}        &0.0109   & 99.9\%  & 0.0129\%   &0.0136     &99.9\%    &0.0129\%, \\ 
  {CO}          &0.00909  & 100\%   & 0\%          &0.00823    &100\%     &0\% \\
  {TW}            &0.0868   & 100\%   & 0\%          &0.117      &100\%     &0\% \\
  {FR}         &0.330    & 100\%   & 0\%          &0.3266     &29.0\%    &43.9\% \\
  {CW}            &2.04     & 97.6\%  & 0.161\%    &1.523      &97.6      &0.161\% \\
  {HL14}      &2.49     & 91.4\%  & 0.560\%    &3.101      &91.4\%,   &0.530\%, \\
  {HL12}      &8.42     & 93.9\%  & 0.538\%    &7.586      &93.9\%    &0.483\%, \\
  \bottomrule
\end{tabular}
\caption{Performance of the \bfssample{} and \lddsample{} schemes on
  graphs used in our experiments. The {\bf Cov.} columns report the
  percentage of vertices captured in the most frequent component captured
  by the sampling method. The {\bf IC.} columns report the fraction of
  inter-component edges remaining after applying the sampling method.
  All times reported are measured on a 72-core machine (with 2-way
  hyper-threading enabled).}
\label{table:bfs_ldd_sampling}
\end{table}

\myparagraph{\bfssample{} and \lddsample{} Performance}
We start with the closely related \bfssample{} and \lddsample{}
schemes. Table~\ref{table:bfs_ldd_sampling} reports the running times
and sampling quality of these two schemes. For \lddsample{}, we report
the results using the fixed parameters that are used in our
experiments ($\beta = 0.2$, and setting permuting to
$\mathsf{False}$).

Perhaps as expected, the running time of \bfssample{} is high on the
high-diameter road\_usa graph, although it completely covers this
graph achieving 100\% coverage, with no inter-component edges
remaining. For the remaining graphs, the running times are reasonable,
and a BFS from a random source covers a large portion of the graph,
leaving just a fraction of a percentage of inter-component edges on
four of our graphs.

For \lddsample{}, the performance is quite reasonable, running in close
to the running time of a BFS on each of these graphs, ranging from
35.9x faster on road\_usa, to at most 1.24x slower on Hyperlink2014.
\lddsample{} achieves an average speedup of 5.35x over \bfssample{},
running between 1.39x slower to 1.34x faster, and is within 1\% of the
running time of \bfssample{} on average on graphs excluding road\_usa.
Thus, for graphs with low to moderate diameter, the performance of the
two schemes is roughly equivalent. Compared to \bfssample{},
\lddsample{} achieves similar coverage of vertices in the most
frequent connected component, and also cuts a similar number of
inter-component edges, actually cutting 5--11\% fewer edges on both
the Hyperlink2014 and Hyperlink2012 graphs. However, both schemes
already cut fewer than 1\% of the edges, and so the difference in
quality is very modest.

\begin{figure}[!t]
  \centering
  \hspace{5em}
    \includegraphics[width=\columnwidth]{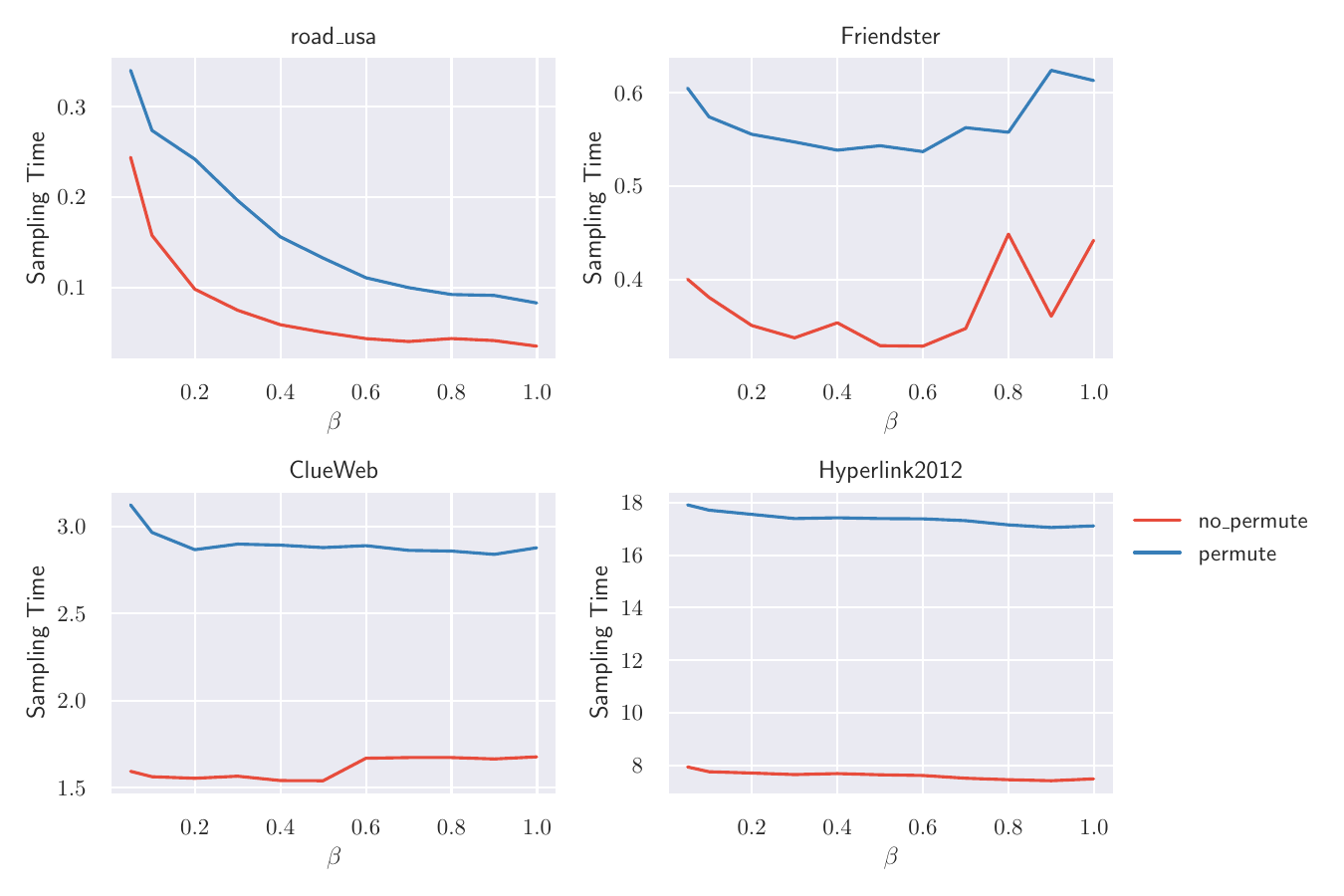}\\ 
    \caption{\small Plot of $\beta$ vs. running time (seconds) when
    permuting is both enabled and disabled in LDD sampling.
    Experiments are run on a 72-core machine with 2-way hyper-threading.}\label{fig:ldd_running_times}
\end{figure}

\begin{figure}[!t]
  \centering
  \hspace{5em}
    \includegraphics[width=\columnwidth]{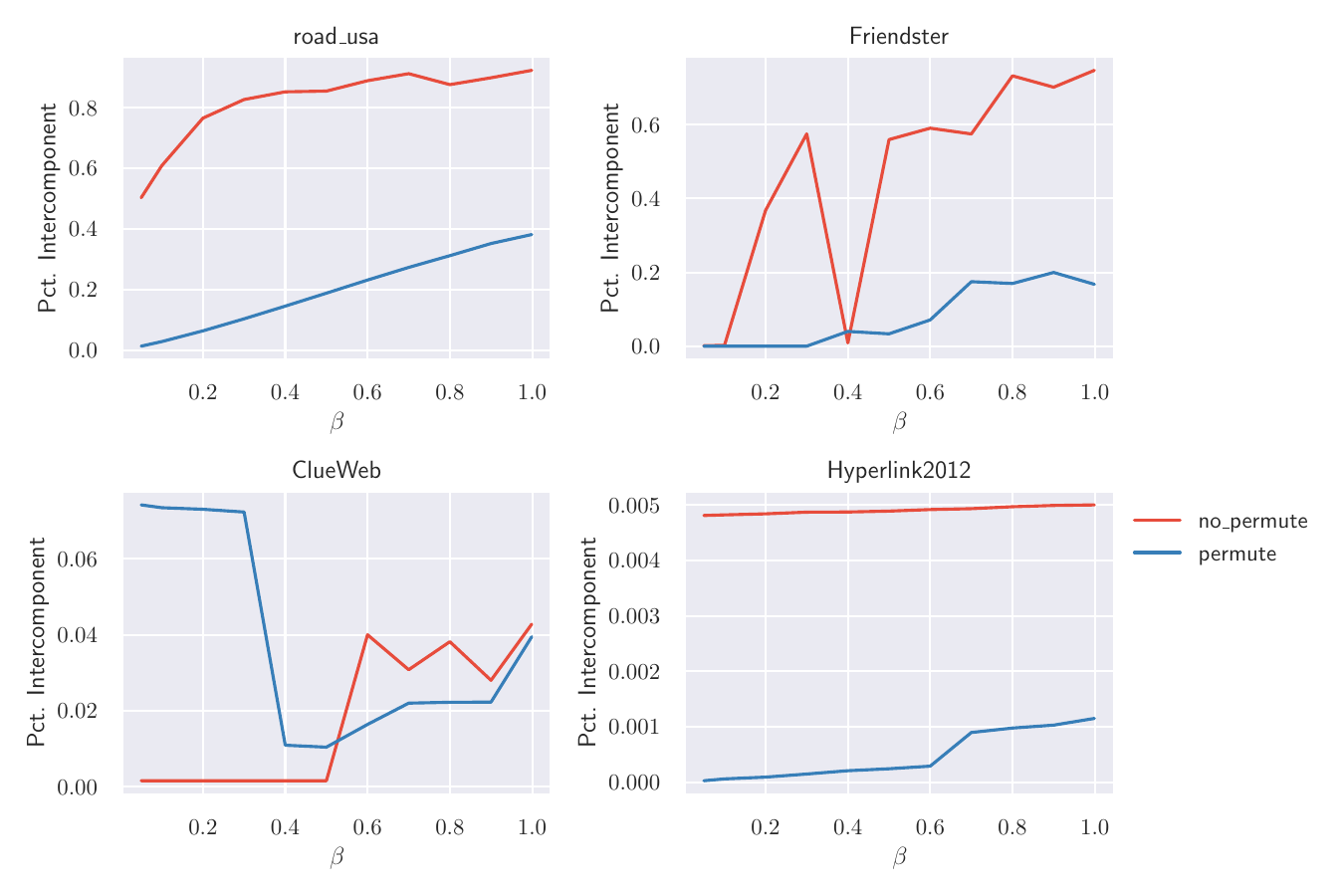}\\ 
    \caption{\small Plot of $\beta$ vs. fraction of inter-component
    edges remaining after sampling when permuting is both enabled and
    disabled in LDD sampling.}\label{fig:ldd_ic_edges}
\end{figure}

\begin{figure}[!t]
  \centering
  \hspace{5em}
    \includegraphics[width=\columnwidth]{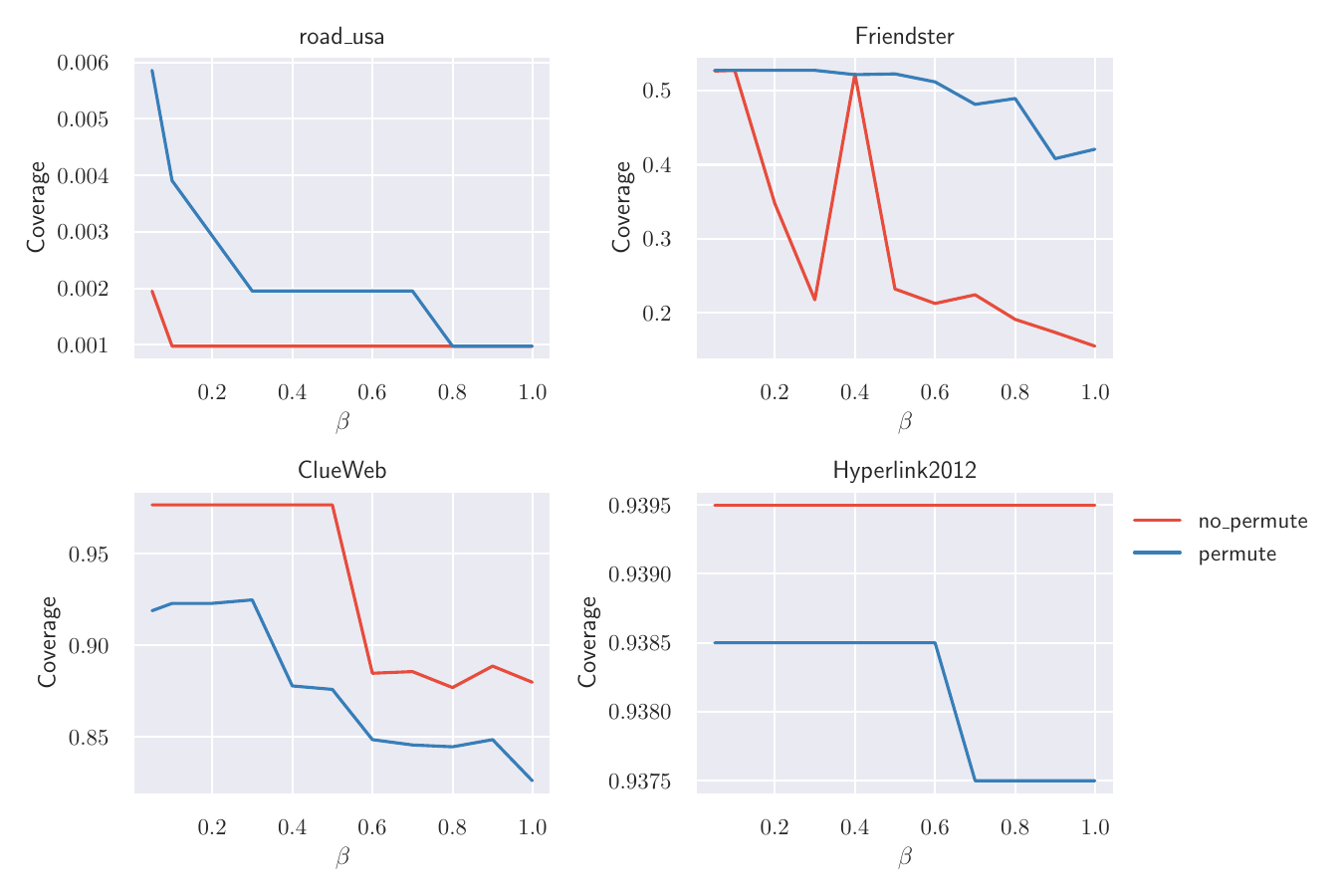}\\ 
    \caption{\small Plot of $\beta$ vs. percentage of vertices covered
    in the largest component when permuting is both enabled and
    disabled in LDD sampling. }\label{fig:ldd_coverage}
\end{figure}

\myparagraph{\lddsample{} Parameters}
Next, we study how adjusting the value of $\beta$ and applying a
random permutation on the vertex ordering in the LDD algorithm affects
both the coverage and number of inter-component edges. In theory, the
$\beta$ parameter in \lddsample{} controls both the strong diameter of
the components, as well as the number of edges cut in the graph in
expectation. In particular, for some $0 < \beta < 1$, \lddsample{}
produces clusters with diameter $O(\log n / \beta)$, and cuts $O(\beta
m)$ edges in expectation.

Figure~\ref{fig:ldd_running_times} plots the tradeoff between
different values of $\beta$ and the overall running time of the
algorithm on four of our input graphs, including the high-diameter
road\_usa graph, the Friendster social network, and the ClueWeb and
Hyperlink2012 Web graphs. We see that for road\_usa, the running time
decreases as a function of $\beta$, which is expected since larger
values of $\beta$ produce lower-diameter clusterings which in turn
require fewer synchronous rounds to compute. However, higher values of
$\beta$ also indicate a larger number of distinct components, since
more vertices wake up in early rounds and start their own clusters.
This explains why the running time actually increases for large values
of $\beta$ on the Friendster graph.

Figure~\ref{fig:ldd_ic_edges} plots the fraction of
\emph{inter-component edges}; letting $I$ be the fraction of
inter-component edges, the theory predicts that $I \approx \beta$.
This fraction is of interest, since $I$ is a lower-bound on the
number of edges that must be processed by our connectivity algorithms,
since all edges coming out of the most frequently identified component
are skipped. We can see that theory is very useful in predicting the
number of inter-component edges---for the road\_usa graph, the number of
inter-component edges grows linearly with $\beta$ (almost exactly
following the line $.4 \beta$). We also observe that the number of cut
edges is not always monotonically increasing with $\beta$, which can
be observed on the ClueWeb graph, although the fraction of cut edges
does seem to grow as $\beta$ is increased as a rule of thumb.

Lastly, Figure~\ref{fig:ldd_coverage} shows the fraction of vertices
covered in the most frequent component identified by \lddsample{}, for
different values of $\beta$. We see that the size of this component is
extremely small for the road\_usa graph (less than 1\%), which is
consistent with what theory would predict for a high-diameter graph.
For lower-diameter graphs, the fraction covered in the most frequent
component is much higher, ranging from 80\%--98\% for the ClueWeb
graph and 93\%-94\% for the Hyperlink2012 graph. Both using a random
permutation and using the original vertex ordering perform well,
although on the large real-world Web graphs, the original ordering
consistently produces a massive component that is 1--5\% larger than
the one produced by a random permutation of vertices.

\begin{figure}[!t]
  \centering
  \hspace{5em}
    \includegraphics[width=\columnwidth]{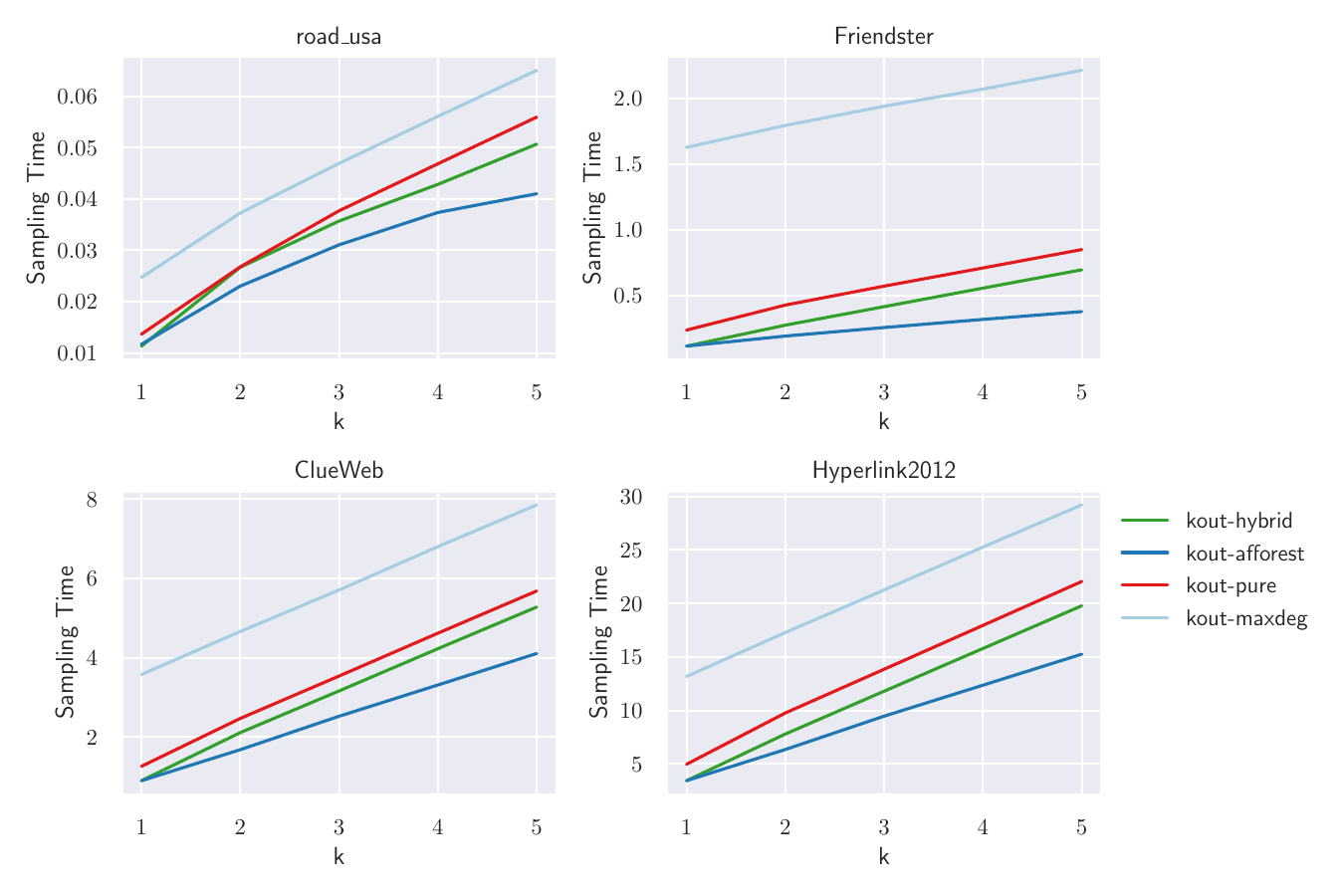}\\ 
    \caption{\small Plot of $k$ vs. running time (seconds) for
      different variants of \koutsample{}.
    Experiments are run on a 72-core machine with 2-way hyper-threading.}\label{fig:kout_running_times}
\end{figure}

\begin{figure}[!t]
  \centering
  \hspace{5em}
    \includegraphics[width=\columnwidth]{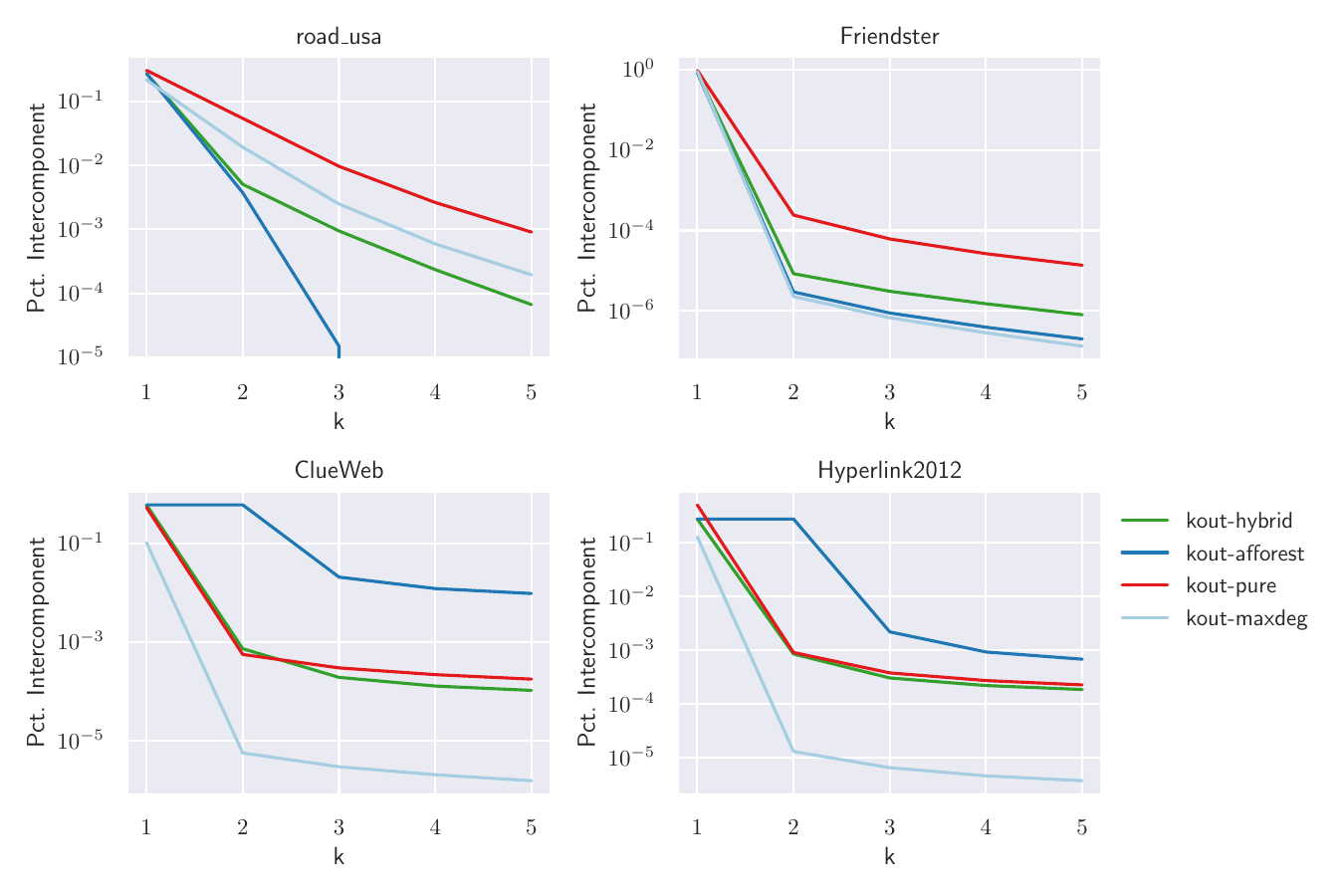}\\ 
    \caption{\small Plot of $k$ vs. fraction of inter-component
    edges remaining after sampling for different variants of
  \koutsample{}.}\label{fig:kout_ic_edges}
\end{figure}

\begin{figure}[!t]
  \centering
  \hspace{5em}
    \includegraphics[width=\columnwidth]{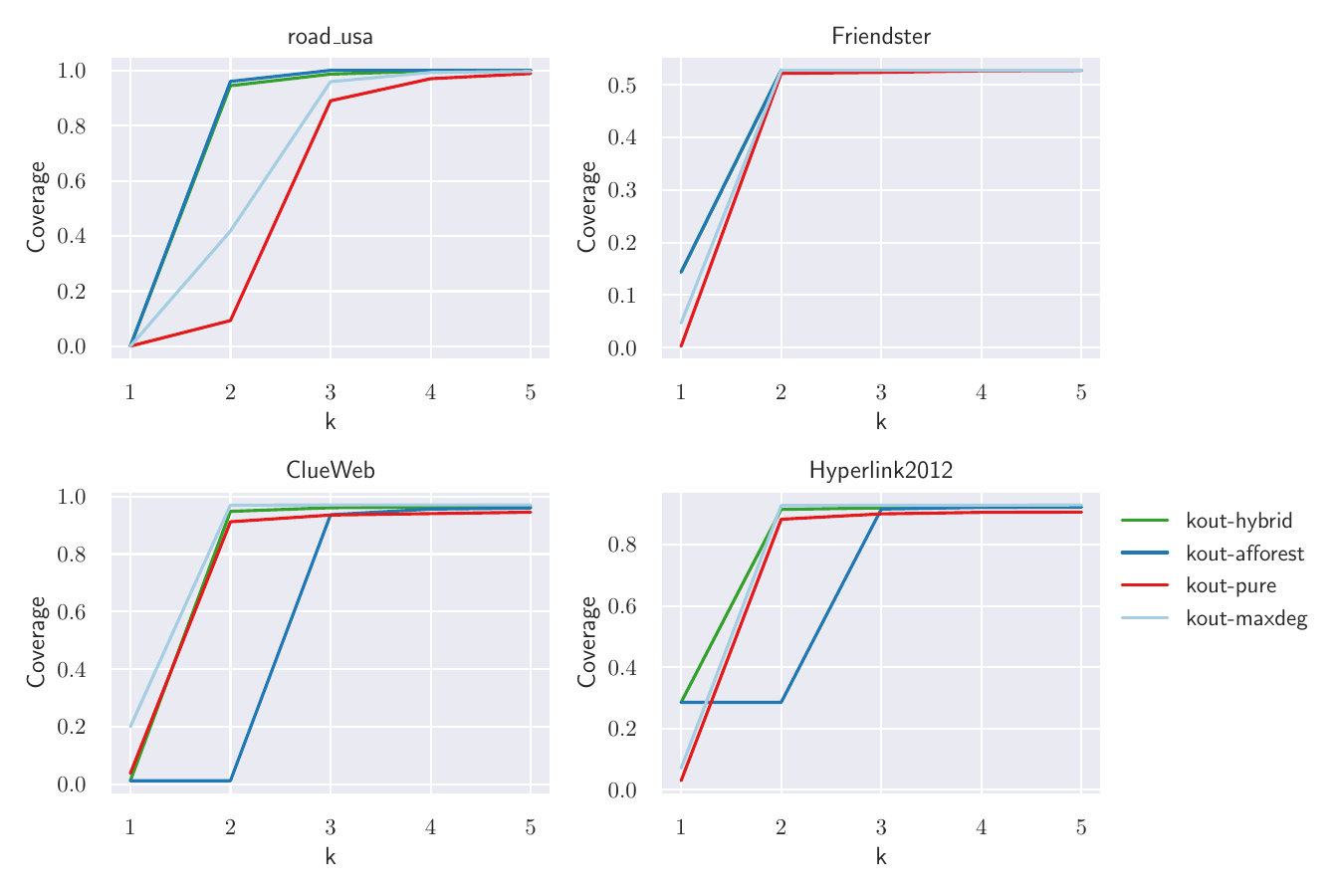}\\ 
    \caption{\small Plot of $k$ vs. percentage of vertices covered
    in the largest component for different variants of \koutsample{}. }\label{fig:kout_coverage}
\end{figure}

\myparagraph{\koutsample{}}
We now turn to \koutsample{}, which is particularly interesting due to
both its connection with recent fundamental theoretical
results~\cite{holm2019kout}. As described in
Section~\ref{subsec:sampling}, given an input $k \geq 1$, the
algorithm implemented in this paper selects the first neighbor
incident to each vertex, and selects $k-1$ other edges uniformly at
random. The sampled components are computed by contracting all
selected edges. We use a \unionfind{} algorithm to perform the
contraction, but note that other min-based algorithms can also be used
for this step.

However, there are several other design possibilities and
options in this scheme. We explore the following $k$-out algorithm
variants in this paper, which differ in how they select edges to use
in the $k$-out algorithm (Algorithm~\ref{alg:kout_sample}).

\begin{itemize}[itemsep=0pt]
  \item Afforest (\koutafforest{}): This scheme proposed by
    Sutton et al.~\cite{Sutton2018} selects \emph{the first $k$ edges}
    incident to each vertex to use.

  \item Pure Sampling (\koutpure{}): This scheme proposed by Holm et
    al.~\cite{holm2019kout} selects $k$ edges incident to each vertex
    uniformly at random to use.

  \item Hybrid Sampling (\kouthybrid{}): This scheme selects
    the first edge incident to each vertex, and $k-1$ edges incident
    to the vertex uniformly at random to use.

  \item Max-Degree Sampling (\koutmaxdeg{}): This scheme
    selects the edge corresponding to the largest degree neighbor
    incident to each vertex, and $k-1$ edges other incident to the
    vertex uniformly at random to use.
\end{itemize}

Both the \kouthybrid{} and \koutmaxdeg{} schemes are new strategies
proposed in this paper. We are not aware of any experimental
evaluation of the \koutpure{} scheme proposed by Holm et
al.~\cite{holm2019kout} prior to our work. The result Holm et
al.~provide for \koutpure{} is as follows:

\begin{theorem}\label{thm:holm}
  Let
  $G$ be an undirected $n$-vertex graph and let $k \geq c \log n$,
  where $c$ is a sufficiently large constant. Let $G'$ be a random $k$-out
  subgraph of $G$. Then the expected number of edges in $G$ that
  connect different connected components of $G'$ is $O(n/k)$.
\end{theorem}

\emph{\koutsample{} Running Time.}
Figure~\ref{fig:kout_running_times} shows the running times of
the different strategies on four of our graph inputs. We observe that
the \koutmaxdeg{} strategy consistently requires the largest running
times, since it requires reducing over all edges incident to the
vertex, and performing an indirect read on the neighboring endpoint to
find the largest degree neighbor. The other schemes all have a similar
cost, with \koutafforest{} being the fastest across all graphs, and
\koutpure{} being the second slowest across all graphs. As expected,
\kouthybrid{}'s performance interpolates that of \kouthybrid{} and
\koutpure{}, and is equal to the performance of \koutafforest{} for
$k=1$. This is expected, since the hybrid algorithm is the same as
\koutafforest{} for $k=1$, and selects $k-1$ random edges, thus
approximating the performance of \koutpure{} for larger $k$.

\emph{\koutsample{} Inter-Component Edges.}
Figure~\ref{fig:kout_ic_edges} plots the fraction of
\emph{inter-component edges} generated by different variants of the
\koutsample{} scheme as a function of $k$. We plot the fraction of
inter-component edges in log-scale to help interpret the plot.  We see
that using $k=1$ performs very poorly for all schemes, with the
exception of the \koutmaxdeg{} scheme, where it achieves reasonable
performance for the ClueWeb and Hyperlink2012 graphs (both of these
graphs are power-law degree distributed, and have several very high
degree vertices). The \koutafforest{} scheme performs very well on the
road\_usa graph, probably due to the fact that most vertices in this
graph have just a few neighbors---for $k \geq 4$ the \koutafforest{}
sampling scheme finds all components of this graph and thus the
inter-cluster edge percentage is 0 for $k \geq 4$. For the other
graphs, the \koutmaxdeg{} scheme consistently performs best.

The \koutpure{} and \kouthybrid{} schemes perform similarly, with the
\kouthybrid{} scheme usually performing slightly better when the
\koutafforest{} algorithm performs well due to the graph's vertex
ordering. On the other hand, for the ClueWeb and Hyperlink2012 graphs,
the \koutafforest{} scheme performs quite poorly---these appear to be
instances where greedily taking the first $k$ vertices based on the
input vertex order results in the graph clustering into local
clusters, with many inter-cluster edges between them. The vertex
orders of these graphs are based on the lexicographical ordering of
pages on the Web, which orders web pages within the same domain
consecutively. Therefore, one explanation for this behavior is that
the first $k$ edges are all "inter-domain" edges which result in
clusters based on domains being formed by the sampling algorithm.

\begin{table}[!t]
\footnotesize
\centering
\begin{tabular}[t]{l | c|c|c}
  \toprule
  {\bf Graph} & $\mathsf{KOut (Hybrid)}$ (s) & $\mathsf{KOut (Hybrid)}$ Cov. & $\mathsf{KOut (Hybrid)}$ IC.  \\
  \midrule
  {RU}          &0.0267   &94.4\%  &0.505\%    \\
  {LJ}        &8.82e-2  &99.9\%  &4.20e-4\%    \\ 
  {CO}          &8.574e-2 &100\%   &0\%    \\
  {TW}            &0.112    &99.9\%  &5.66e-3\%    \\
  {FR}         &0.274    &52.7\%  &8.39e-4\%    \\
  {CW}            &2.11     &94.8\%  &0.735\%    \\
  {HL14}      &3.31     &89.9\%  &0.0323\%    \\
  {HL12}      &7.79     &91.5\%  &0.0857\%    \\
  \bottomrule
\end{tabular}
\caption{Performance of the \koutsample{} schemes on
  graphs used in our experiments. The {\bf Cov.} column reports the
  percentage of vertices captured in the most frequent component captured
  by the sampling method. The {\bf IC.} column reports the fraction of
  inter-component edges remaining after applying the sampling method.
  All times reported are measured on a 72-core machine (with 2-way
  hyper-threading enabled).
  }
\label{table:kout_sampling}
\end{table}

\emph{\koutsample{} Coverage.}
Next, Figure~\ref{fig:kout_coverage} plots the fraction of vertices
contained in the most frequent component identified by sampling as a
function of $k$. The results are similar to those previously analyzed
from Figure~\ref{fig:kout_ic_edges}: on some graphs with good vertex
orderings (road\_usa, Friendster), \koutafforest{} achieves excellent
performance. However, on other graphs (ClueWeb, Hyperlink2012),
\koutafforest{} achieves very poor performance due to the ordering
containing significant local structure that prevents large components
from forming. The \koutpure{} algorithm addresses this issue, and
achieves consistent performance on all graphs. The \kouthybrid{}
algorithm further improves the performance of \koutpure{},
significantly outperforming \koutafforest{} on the ClueWeb and
Hyperlink2012 graphs for $k=2$, and achieving performance close to
\koutafforest{} on the other graphs. Table~\ref{table:kout_sampling}
shows the running times, percent coverage, and percent of
inter-component edges remaining for the \koutsample{} scheme used in
our paper ($k = 2$ using the \kouthybrid{} algorithm).

\emph{\koutsample{} Summary.}
Finally, we consider how the performance of these schemes compares
with the predictions of the theory. First, we note that the theory
rules out the possibility of Theorem~\ref{thm:holm} holding for $k=1$
(however, it is conjectured to hold for $k \geq 2$). We see that in
practice for $k \geq 2$ a huge fraction of the edges become
intracluster, with very few remaining inter-cluster edges, and far
fewer edges remain than predicted by theory. For example, for the
Hyperlink2012 graph, the \koutpure{} scheme results in 0.0009149\%
inter-component edges, or 206M inter-component edges, which is about
1/20th the number of vertices in this graph. Our results suggest that
the empirical performance of these schemes is far better than
suggested by theory. It would be interesting to explore applying these
ideas to speed up connectivity algorithms in other settings, such as
practical distributed connectivity implementations, or possibly even
in practical dynamic connectivity algorithms.

\subsubsection{Comparison with MapEdges and GatherEdges}\label{sec:edge_gather}
Next, we compare the cost of our connectivity algorithms with two
basic graph processing primitives, \textsc{MapEdges} and
\textsc{GatherEdges}. These are primitives that map over all vertices
in parallel, and perform a parallel reduction over their incident
edges, storing the result in an array proportional to the number of
vertices. \textsc{MapEdges} returns a value of $\mathsf{1}$ per edge,
which corresponds to computing the degree of each vertex after reading
every edge.  \textsc{GatherEdges} performs an indirect read into an
array at the location corresponding to the neighbor id, and so the
performance of this primitive depends on the structure and ordering of
the input graph.

Both primitives are useful baselines for understanding the performance
of more complex graph algorithms.
In particular, \textsc{MapEdges} corresponds to the cost of reading
and processing a graph input, and storing some output value for each
vertex.  \textsc{GatherEdges} serves as a useful practical lower-bound
on the performance of graph algorithms that perform indirect accesses
over every edge. This is useful for connectivity algorithms, since
every connectivity algorithm that always outputs a correct answer (Las
Vegas algorithm) must check some connectivity information for both
endpoints of an edge, which seems impossible to do without performing
an indirect read.\footnote{One exception is for $n$ vertex graphs with
  more than ${n-1}\choose{2}$ edges, in which case it is clear the
  graph is guaranteed to be connected.}
All connectivity algorithms in our framework perform an indirect read
across every edge to the parent array of the neighboring vertex.

Table~\ref{table:gathervsalg} shows the running time of the
\textsc{MapEdges} and \textsc{GatherEdges} primitives, and the time
taken by the fastest connectivity implementation in \framework{} for
that graph. First, we observe that there is an order of magnitude
difference in running times between the \mapedges{} and
\gatheredges{}: \gatheredges{} is 13.5x slower on average than
\mapedges{} across all inputs. The difference between the two is more
noticeable for small graphs, since the sequential reads used in
\mapedges{} are likely to be served by cache. On the other hand, for
\gatheredges{}, even our smallest graph has 23.9M vertices, whose
vertex data is already larger than our multicore machine's LLC, and
thus many of the indirect reads across edges will be cache misses.

\setlength{\tabcolsep}{2pt}
\begin{table}[!t]
\footnotesize
\centering
\begin{tabular}[t]{l | r|r|r|r}
  \toprule
  {\bf Graph} & \textsc{MapEdges} &\textsc{GatherEdges}  & (No Sample)
              & \framework{} (Sample) \\
  \midrule
  {road\_usa}          &5.54e-3   &1.33e-2  & 2.80e-2  & 3.77e-2   \\ 
  {LiveJournal}        &1.31e-3   &8.59e-3  & 1.27e-2  & 8.96e-3   \\ 
  {com-Orkut}          &9.77e-4   &1.65e-2  & 1.91e-2  & 8.56e-3   \\
  {Twitter}            &2.63e-2   &0.488    & 0.316    & 9.24e-2   \\
  {Friendster}         &2.77e-2   &1.50     & 0.902    & 0.183     \\
  {ClueWeb}            &0.790     &2.77     & 4.04     & 1.69      \\
  {Hyperlink2014}      &1.36      &4.14     & 6.64     & 2.83      \\
  {Hyperlink2012}      &2.96      &10.4     & 13.9     & 8.20      \\
  \bottomrule
\end{tabular}
\caption{72-core with hyper-threading running times (in seconds)
  comparing the performance of \textsc{GatherEdges} to the performance
  of the fastest connectivity implementation in \framework{} without
sampling (No Sample), and with sampling used (Sample).}
\label{table:gathervsalg}
\end{table}

\subsection{Algorithms for Massive Web Graphs}\label{sec:large_graphs}
Table~\ref{table:big_comparison} in the introduction presents the
results of comparing \framework{} algorithms with these existing
state-of-the-art results from different computational paradigms. We
can see that in all cases, \framework{} achieves significant speedup,
sometimes ranging on orders of magnitude in terms of raw running time
alone. We now provide a detailed breakdown of how \framework{}
algorithms compare with the fastest systems from different categories.

\subsubsection*{Distributed Memory Systems}
There have been various distributed-memory connected components
algorithms that report running times for Hyperlink2012, which we
compare with here. Slota et al.~\cite{Slota2016} present a
distributed-memory connectivity algorithm, and report a running time
of \emph{63 seconds} for finding the largest connected component on the
Hyperlink2012 graph using 8,192 cores (256 nodes each with two 16-core
processors) on the Blue Waters supercomputer.  Stergiou et
al.~\cite{Stergiou2018} describe a distributed-memory connected
components algorithm. Using 12,000 cores (1,000 nodes each with two
6-core CPUs) and 128TB of RAM on a proprietary cluster at Yahoo!, they report a
running time of \emph{341 seconds} for this graph.  Dathathri et
al.~\cite{Dathathri2018} report a running time of \emph{75.3 seconds} for
connected components in the D-Galois graph processing system for
Hyperlink2012 using 256 nodes (each with a KNL processor with 68 cores
and 96GB RAM) on the Stampede KNL Cluster.  Very recently, Zhang et
al.~\cite{zhang2020fastsv} present a variant of Shiloach-Vishkin for
distributed-memory and report a running time of \emph{30 seconds} on
262,000 cores of an XC40 supercomputer. Compared with these results,
our fastest times are significantly faster.
\begin{itemize}[topsep=0pt,itemsep=0pt,parsep=0pt,leftmargin=16pt]
  \item Compared to Slota et al.~\cite{Slota2016}, our results are
  7.68x faster, using a machine with 16x less memory, and 56x fewer
  hyper-threads. Additionally, our algorithm finds all connected
  components, whereas their algorithm can only find the largest
  connected component (this is trivial to do on this graph by running
  a breadth-first search from a random vertex).

  \item Compared to Stergiou et al.~\cite{Stergiou2018}, our results
  achieve 41.5x speedup over their result, using 128x less memory, and
  166x fewer hyper-threads.

  \item Compared to the recent results of Dathathri et
  al.~\cite{Dathathri2018}, our results are 9.18x faster. Furthermore,
  we use 24x less memory, and 483x fewer hyper-threads.

  \item Compared to the very recent results of Zhang et al.~\cite{zhang2020fastsv}, which is
  the fastest current distributed-memory running time to solve
  connectivity on this graph, our results are 3.65x faster, using at
  least 256x less memory (each node on their system has at least 64GB
  of memory, to up to 256GB of memory, but unfortunately the authors
  did not report these statistics), and 1,819x fewer hyper-threads.
\end{itemize}

\subsubsection*{External Memory Systems}
Mosaic~\cite{maass2017mosaic} reports in-memory running times on the
Hyperlink2014 graph. Their system is optimized for external memory
execution. They solve connected components in 700 seconds on a machine
with 24 hyper-threads and 4 Xeon-Phis (244 cores with 4 threads each)
for a total of 1,000 hyper-threads, 768GB of RAM, and 6 NVMe SSDs.
FlashGraph~\cite{da2015flashgraph} reports disk-based running times
for the Hyperlink2012 graph on a 4-socket, 32-core machine with 512GB
of memory and 15 SSDs. On 64 hyper-threads, they solve connected
components in 461 seconds.
\begin{itemize}[topsep=0pt,itemsep=0pt,parsep=0pt,leftmargin=16pt]
  \item Compared with Mosaic, \framework{} is 247x faster, using 1.3x
    more memory, but 6.94x fewer hyper-threads. Although the Xeon-Phi
    chips Mosaic uses are no longer available, from a cost
    perspective, machines like the one we use in this paper are widely
    available from existing cloud service providers for a few dollars
    on the hour, but to the best of our knowledge, one cannot easily
    rent a machine like the one used in Mosaic.

  \item Compared with FlashGraph's result for the Hyperlink2012 graph,
    \framework{} is 56x faster, but uses 1.95x more memory, and 2.25x
    more hyper-threads. Therefore, we achieve orders of magnitude more
    performance, using comparatively little extra resources, although
    the comparison is not entirely fair, since FlashGraph must read
    this graph from an array of SSDs. It would be interesting for
    future work to evaluate whether our union-find algorithms can be
    used as part of a ``semi-external'' memory connectivity algorithm
    that streams the edges of the graph from SSD, and whether such an
    approach would improve on existing external and semi-external
    memory results.
\end{itemize}

\subsubsection*{Existing Shared-Memory (DRAM and NVRAM) Systems}
Finally, we compare with several recent results showing that the
Hyperlink2012 graph can be processed in the main memory of a single
machine. To the best of our knowledge, the first such result is from
Dhulipala et al.~\cite{DhBlSh18}, who showed a running time of 38.3
seconds for this graph on a single-node system with 1TB of memory and
144 hyper-threads. The authors later presented an improved result,
showing that the algorithm runs in just over 25 seconds on the same
machine.
Subsequently, two simultaneous results achieved very high performance
for processing this graph on a system with relatively little DRAM, but
a significant amount of NVRAM~\cite{Gill2020, dhulipala2020semi}. The
result by Gill et al.~\cite{Gill2020} extends the Galois system to
support graph processing when the graph is stored on non-volatile
memory. They achieve a running time of 76 seconds on a machine with
376GB of memory, 3TB of NVRAM, and 96 hyper-threads. Dhulipala et
al.~\cite{dhulipala2020semi} also describe a system for graph
processing on NVRAMs based on a semi-asymmetric approach. Their system
runs in 36.2 seconds on the same machine (376GB of memory, 3TB of
NVRAM and 96 hyper-threads).
\begin{itemize}[topsep=0pt,itemsep=0pt,parsep=0pt,leftmargin=16pt]
  \item Compared with GBBS, our results are 3.2x faster on the same
  machine as used in~\cite{DhBlSh18}. Given that this result was the
  fastest previous result under any computational model that we are
  aware of, we improve on the state-of-the-art time for processing
  this graph by 3.2x, while using exactly the same resources as those
  used by the authors of GBBS~\cite{DhBlSh18}.

  \item Compared with the non-volatile memory results of Gill et
  al.~\cite{Gill2020} and Dhulipala et al.~\cite{dhulipala2020semi},
  our results for Hyperlink2012 are 9.26x faster, using a system with
  2.65x more memory, and 1.5x more hyper-threads. It is an interesting
  question how quickly our new algorithms solve the connectivity
  problem when running in the system used by these authors, and also
  theoretically in the model proposed by Dhulipala et
  al.~\cite{dhulipala2020semi}.

\end{itemize}


\section{\framework{} Pseudocode}\label{apx:pseudocode}
In this section, we provide
pseudocode for methods in the \framework{} framework.

\algblock{ParFor}{EndParFor}
\algnewcommand\algorithmicparfor{\textbf{parfor}}
\algnewcommand\algorithmicpardo{\textbf{do}}
\algnewcommand\algorithmicendparfor{}
\algrenewtext{ParFor}[1]{\algorithmicparfor\ #1\ \algorithmicpardo}
\algrenewtext{EndParFor}{\algorithmicendparfor}
\algtext*{EndParFor}{}

\subsection{Sampling Schemes}\label{apx:subsec:sampling_pseudocode}
In Algorithms~\ref{alg:kout_sample},~\ref{alg:bfs_sample},
and~\ref{alg:ldd_sample}, we provide the deferred pseudocode for our
sampling methods in \framework{}, and refer to the text in
Section~\ref{subsec:sampling} for the descriptions.

\begin{algorithm*}
\caption{\koutsample{}} \label{alg:kout_sample}
\small
\begin{algorithmic}[1]
\Procedure{$k$OutSample}{$G(V, E), \codevar{labels}, k = 2$}
  \State $\codevar{edges} \gets$ $\{$ First edge from each vertex $\} \cup \{$Sample $k-1$ edges uniformly at random from each vertex$\}$
  \State $\textsc{UnionFind}(\codevar{edges}, \codevar{labels})$
  \State Fully compress the components array, in parallel.\label{line:koutcompress}
  \State  \algorithmicreturn{} $\codevar{labels}$
\EndProcedure
\end{algorithmic}
\end{algorithm*}

\begin{algorithm*}\caption{Breadth-First Search Sampling} \label{alg:bfs_sample}
\small
\begin{algorithmic}[1]
\Procedure{\bfssample{}}{$G(V, E), \codevar{labels}$, $c = 5$}
  \For{$i \in [c]$}
    \State $s \gets \textsc{RandInt}() \bmod |V|$
    \State $\codevar{labels} \gets \textsc{LabelSpreadingBFS}(G, s)$
    \Comment{Runs direction-optimizing BFS from $s$.
    $\codevar{labels}[u] = s$ if $u$ is reachable from $s$, and
    $\codevar{labels}[u] = u$ otherwise.}
    \State $\codevar{freq} \gets \textsc{IdentifyFrequent}(\codevar{labels})$
    \If {$\codevar{freq}$ makes up more than 10\% of labels}
      \State  \algorithmicreturn{} $\codevar{labels}$
    \EndIf
  \EndFor
  \State \algorithmicreturn{} $\{i \rightarrow i\ |\ i \in [1, |V|]\}$
  \Comment{Return the identity labeling}
\EndProcedure
\end{algorithmic}
\end{algorithm*}

\begin{algorithm*}[!t]\caption{Low-Diameter Decomposition Sampling} \label{alg:ldd_sample}
\small
\begin{algorithmic}[1]
\Procedure{\lddsample{}}{$G(V, E), \codevar{labels}$}
  \State $\codevar{labels} \gets \textsc{LDD}(G, \beta = 0.2)$
  \State  \algorithmicreturn{} $\codevar{labels}$
\EndProcedure
\end{algorithmic}
\end{algorithm*}

\subsection{Union Find Algorithms}
Next, we provide the pseudocode for all of our union-find
implementations, with the one exception of \jayanti{}'s algorithm and
\findtwotry{}, whose pseudocode is presented in~\cite{JayantiTB19}.
Algorithm~\ref{alg:union_find_meta} shows our generic union-find
template algorithm, which takes a custom union operator, find
operator, and splice operator, and runs the resulting algorithm
combination. Note that the splice operator is only valid for Rem's
algorithms. Algorithm~\ref{alg:find_options} provides pseudocode for
the different find implementations. Algorithm~\ref{alg:splice_options}
provides pseudocode for different splice options.

\begin{itemize}[topsep=1pt,itemsep=0pt,parsep=0pt,leftmargin=20pt]
  \item \unionasync{} (Algorithm~\ref{alg:unionasync})
  \item \unionhook{} (Algorithm~\ref{alg:union_nondeterministic})
  \item \unionearly{} (Algorithm~\ref{alg:union_early})
  \item \unionremlock{} (Algorithm~\ref{alg:union_rem_lock})
  \item \unionremcas{} (Algorithm~\ref{alg:union_rem_cas})
\end{itemize}

\subsection{Shiloach-Vishkin}
 The pseudocode for our adaptation of the Shiloach-Vishkin
algorithm is presented in Algorithm~\ref{alg:shiloach_vishkin}.

\subsection{Liu-Tarjan Algorithms}

Algorithms in the Liu-Tarjan framework
consist of rounds, where each round performs a \emph{connect phase}, a
\emph{shortcut phase}, and possibly an \emph{alter phase}.
The \emph{connect phase} updates the parents of edges based on
different operations. A \connect{} uses the endpoints of an edge as
candidates for both vertices. A \parentconnect{} uses the parents of
the endpoints of an edge as candidates for both vertices. Lastly, an
\extendedconnect{} uses the parent values of an edge as candidates for
both the edge endpoints as well as the parents of the endpoints.
Within the connect phase, the operation can also choose to update the
parent value of a vertex if and only if it is a tree-root at the start
of the round, which we refer to as \rootupdate{}. In all cases, an
update occurs if and only if the candidate parent is smaller than the
current parent.  Next, the algorithms perform the \emph{shortcut phase}, which performs one step of path
compression. Some of the algorithms perform a \fullshortcut{}, which
repeats shortcutting until no further parents change.  The algorithm
can then execute an optional \emph{alter phase}.  \alter{} updates the
endpoints of an edge to be the current labels of the endpoints (this
step is required for correctness when using \connect{}).

Below we list the variants of the Liu-Tarjan framework
implemented in \framework{}. We note that only five of these variants
were explored in their original paper. 
We refer the reader to
Section~\ref{subsec:finish} for details about the different algorithm
options in this framework, as well as the description of the overall
framework.
\begin{enumerate}[label=(\arabic*),itemsep=0pt]
  \item $\mathsf{CUSA}$: $\{\connect{}, \shortcut{},\alter{}\}$ 
  \item $\mathsf{CRSA}$: $\{\connect{}, \rootupdate{}, \shortcut{},\alter{}\}$ 
  \item $\mathsf{PUSA}$: $\{\parentconnect{}, \shortcut{}, \alter{}\}$ 
  \item $\mathsf{PRSA}$: $\{\parentconnect{}, \rootupdate{}, \shortcut{}, \alter{}\}$
  \item $\mathsf{PUS}$: $\{\parentconnect{}, \shortcut{}\}$
  \item $\mathsf{PRS}$: $\{\parentconnect{}, \rootupdate{}, \shortcut{}\}$
  \item $\mathsf{EUSA}$: $\{\extendedconnect{}, \shortcut{}, \alter{}\}$
  \item $\mathsf{EUS}$: $\{\extendedconnect{}, \shortcut{}\}$

  \item $\mathsf{CUFA}$: $\{\connect{}, \fullshortcut{},\alter{}\}$
  \item $\mathsf{CRFA}$: $\{\connect{}, \rootupdate{}, \fullshortcut{},\alter{}\}$
  \item $\mathsf{PUFA}$: $\{\parentconnect{}, \fullshortcut{}, \alter{}\}$
  \item $\mathsf{PRFA}$: $\{\parentconnect{}, \rootupdate{}, \fullshortcut{}, \alter{}\}$
  \item $\mathsf{PUF}$: $\{\parentconnect{}, \fullshortcut{}\}$
  \item $\mathsf{PRF}$: $\{\parentconnect{}, \rootupdate{}, \fullshortcut{}\}$
  \item $\mathsf{EUFA}$: $\{\extendedconnect{}, \fullshortcut{}, \alter{}\}$
  \item $\mathsf{EUF}$: $\{\extendedconnect{}, \fullshortcut{}\}$
\end{enumerate}

\begin{algorithm*}\caption{\ufmetaalgorithm{}} \label{alg:union_find_meta}
\begin{algorithmic}[1]
\State $\mathsf{FindOption} = \{$ \Comment{Shared by UnionFind Algorithms}
  \State \hspace{1em}$\findnaive{}$, \Comment{No path compression}
  \State \hspace{1em}$\findsplit{}$, \Comment{Path splitting}
  \State \hspace{1em}$\findhalve{}$, \Comment{Path halving}
  \State \hspace{1em}$\findcompress{}$, \Comment{Full path compression}
\State $\}$
\State $\mathsf{SpliceOption} = \{$ \Comment{Used by Rem's algorithms when path has not yet reached a root}
  \State \hspace{1em}$\splitone{}$, \Comment{Performs one path split}
  \State \hspace{1em}$\halveone{}$, \Comment{Performs one path halve}
  \State \hspace{1em}$\splice{}$, \Comment{Performs a splice operation}
\State $\}$
\State $\mathsf{UnionOption} = \{$
  \State \hspace{1em}$\mathsf{\unionasync{}}$, \Comment{Asynchronous union-find}
  \State \hspace{1em}$\mathsf{\unionhook{}}$, \Comment{Asynchronous hook-based union-find}
  \State \hspace{1em}$\mathsf{\unionearly{}}$, \Comment{Union with early-optimization}
  \State \hspace{1em}$\mathsf{\unionremlock{}}$, \Comment{Lock-based Rem's algorithm*}
  \State \hspace{1em}$\mathsf{\unionremcas{}}$, \Comment{CAS-based Rem's algorithm*}
\State $\}$
\Procedure{Connectivity}{$G$, $\codevar{labels}$,
$\largestcomp{}$, $\mathsf{UnionOption}$, $\mathsf{FindOption}$,
$\mathsf{SpliceOption}$}
\State $\mathcal{U} = \mathsf{Union\mhyphen Find}(\mathsf{UnionOption}$, $\mathsf{FindOption}$, $\mathsf{FindOption})$
\ParFor {$\{u\ |\ \codevar{labels}[u] \neq \largestcomp{}\}$}
  \ParFor {$\{(u, v) \in d^{+}(u)\}$}
    \State $\mathcal{U}.\mathsf{Union}(u, v, \codevar{labels})$
  \EndParFor
\EndParFor
\EndProcedure
\end{algorithmic}
\end{algorithm*}

\begin{algorithm*}\caption{Find Algorithms} \label{alg:find_options}
\small
\begin{algorithmic}[1]

\Procedure{FindNaive}{$u, P$}
  \State $v \gets u$
  \While {$v \neq P[v]^{*}$}
    \State $v \gets P[v]$
  \EndWhile
  \State \algorithmicreturn{} $v$
\EndProcedure

\Procedure{FindCompress}{$u, P$}
  \State $r \gets u$
  \If {$P[r] = r$} \algorithmicreturn{} $r$
  \EndIf
  \While {$r \neq P[r]$}
    \State $r \gets P[r]$
  \EndWhile
  \While {$j \gets P[u] > r$}\label{countex:compress}
  \State $P[u] \gets r$, $u \gets j$
  \EndWhile
  \State \algorithmicreturn{} $r$
\EndProcedure

\Procedure{\findsplit{}}{$u, P$}
  \State $v \gets P[u], w \gets P[v]$
  \While{$v \neq w^{*}$}
    \State $\textsc{\cas{}}(\&P[u], v, w)$
    \State $u \gets v$
  \EndWhile
  \State \algorithmicreturn{} $v$
\EndProcedure

\Procedure{\findhalve{}}{$u, P$}
  \State $v \gets P[u], w \gets P[v]$
  \While{$v \neq w^{*}$}
    \State $\textsc{\cas{}}(\&P[u], v, w)$
    \State $u \gets P[u]$
  \EndWhile
  \State \algorithmicreturn{} $v$
\EndProcedure
\end{algorithmic}
\end{algorithm*}

\begin{algorithm*}\caption{Splice Algorithms} \label{alg:splice_options}
\small
\begin{algorithmic}[1]

\Procedure{\splitone{}}{$u, x, P$}
  \State $v \gets P[r_u], w \gets P[v]^{*}$
  \If{$v \neq w$}
    \State $\textsc{\cas{}}(\&P[u], v, w)$
  \EndIf
  \State \algorithmicreturn{} $v$
\EndProcedure

\Procedure{\halveone{}}{$u, x, P$}
  \State $v \gets P[r_u], w \gets P[v]^{*}$
  \If{$v \neq w$}
    \State $\textsc{\cas{}}(\&P[u], v, w)$
  \EndIf
  \State \algorithmicreturn{} $w$
\EndProcedure

\Procedure{\splice{}}{$u, v, P$}
\State $p_u = P[u]^{*}$
  \State $\textsc{\cas{}}(\&P[u], p_u, P[v])$
  \State \algorithmicreturn{} $p_u$
\EndProcedure
\end{algorithmic}
\end{algorithm*}

\begin{algorithm*}\caption{\unionasync{}} \label{alg:unionasync}
\small
\begin{algorithmic}[1]
  \State \textsc{Find} = \text{one of } $\{\mathsf{FindNaive}, \mathsf{\findsplit{}}, \mathsf{\findhalve{}}, \mathsf{FindCompress}\}$
\Procedure{Union}{$u, v, P$}
\State $p_u \gets \textsc{Find}(u, P), p_v \gets \textsc{Find}(v, P)$
\While{$p_u \neq p_v$}\Comment{WLOG, let $p_u > p_v$}
\If{$p_u = P[p_u]$ and $\textsc{\cas{}}(\&P[u], p_u, p_v)$}
  \State \algorithmicreturn{}
  \EndIf
  \State $p_u \gets \textsc{Find}(u, P), p_v \gets \textsc{Find}(v, P)$
\EndWhile
\EndProcedure
\end{algorithmic}
\end{algorithm*}

\begin{algorithm*}\caption{\unionhook{}} \label{alg:union_nondeterministic}
\small
\begin{algorithmic}[1]
\State \textsc{Find} = \text{one of } $\{\mathsf{FindNaive},\mathsf{\findsplit{}}, \mathsf{\findhalve{}}, \mathsf{FindCompress}\}$
\State $\codevar{H} \gets \{i \rightarrow \infty\ |\ i \in |V| \}$ \Comment{Array of hooks}
\Procedure{Union}{$u, v, P$}
\State $p_u \gets \textsc{Find}(u, P), p_v \gets \textsc{Find}(v, P)$
\While{$p_u \neq p_v$} \Comment{WLOG, let $p_u > p_v$}
\If{$p_u = P[p_u]$ and $\textsc{\cas{}}(\&H[u], \infty, p_v)$}
  \State $P[p_u] \gets p_v$
  \State \algorithmicreturn{}
  \EndIf
  \State $p_u \gets \textsc{Find}(u, P), p_v \gets \textsc{Find}(v, P)$\Comment{WLOG, let $p_u < p_v$}
\EndWhile
\EndProcedure
\end{algorithmic}
\end{algorithm*}

\begin{algorithm*}\caption{\unionearly{}} \label{alg:union_early}
\small
\begin{algorithmic}[1]
\State \textsc{Find} = \text{one of } $\{\mathsf{FindNaive},\mathsf{\findsplit{}}, \mathsf{\findhalve{}}, \mathsf{FindCompress}\}$
\Procedure{Union}{$u, v, P$}
\State $p_u \gets u, p_v \gets v$
\While{$p_u \neq p_v$} \Comment{WLOG, let $p_u > p_v$}
\If{$p_u = P[p_u]$ and $\textsc{\cas{}}(\&P[u], p_u, p_v)$}
    \State $\textsc{break}$
  \EndIf
  \State $z \gets P[p_u], w \gets P[z]$
  \State $\textsc{\cas{}}(\&P[p_u], z, w)$
  \State $p_u \gets w$
\EndWhile
\State $\textsc{Find}(u, P), \textsc{Find}(v, P)$ \Comment{Elided if the find-option is $\mathsf{FindNaive}$}
\EndProcedure
\end{algorithmic}
\end{algorithm*}

\begin{algorithm*}\caption{\unionremlock{}} \label{alg:union_rem_lock}
\small
\begin{algorithmic}[1]
\State \textsc{Compress} = \text{one of } $\{\mathsf{FindNaive}, \mathsf{\findsplit{}}, \mathsf{\findhalve{}}\}$
\State \textsc{\splice{}} = \text{one of } $\{\mathsf{\splitone{}}, \mathsf{\halveone{}}, \mathsf{\splice{}}\}$
\State $\codevar{L} \gets \mathsf{mutex}[|V|]$ \Comment{Array of locks}
\Procedure{Union}{$u, v, P$}
\State $r_u \gets u, r_v \gets v$
\While{$P[r_u] \neq P[r_v]^{*}$} \Comment{WLOG, let $P[r_u] > P[r_v]$} \label{countex:while}
  \If{$r_u = P[r_u]$}
    \State $L[r_u].\textsc{Lock}()$
    \State $p_v \gets P[r_v]$
    \If{$r_u = P[r_u]$ and $r_u > p_v^{*}$}
      \State $P[r_u] \gets p_v$
    \EndIf
    \State $L[r_u].\textsc{Unlock}()$
    \State \algorithmicreturn{} $\mathsf{True}$
  \Else
    \State $r_u \gets \textsc{\splice{}}(r_u, r_v, P)$
  \EndIf
\EndWhile
\If {$\textsc{Compress} \neq \mathsf{FindNaive}$}
  \State $\textsc{Compress}(u, P), \textsc{Compress}(v, P)$ \label{countex:remfind}
\EndIf
\State \algorithmicreturn{} $\mathsf{False}$
\EndProcedure
\end{algorithmic}
\end{algorithm*}

\begin{algorithm*}\caption{\unionremcas{}} \label{alg:union_rem_cas}
\small
\begin{algorithmic}[1]
\State \textsc{Compress} = \text{one of } $\{\mathsf{FindNaive}, \mathsf{\findsplit{}}, \mathsf{\findhalve{}}\}$
\State \textsc{\splice{}} = \text{one of } $\{\mathsf{\splitone{}}, \mathsf{\halveone{}}, \mathsf{\splice{}}\}$
\Procedure{Union}{$(u, v, P)$}
\State $r_u \gets u, r_v \gets v$
\While{$P[r_u] \neq P[r_v]^{*}$} \Comment{WLOG, let $P[r_u] > P[r_v]$}
\If{$r_u = P[r_u]$ and $\textproc{\cas{}}(\&P[r_u], r_u, P[r_v])^{*}$}
    \If {$\textsc{Compress} \neq \findnaive{}$}
      \State $\textsc{Compress}(u, P)$
      \State $\textsc{Compress}(v, P)$
    \EndIf
    \State \algorithmicreturn{} $\mathsf{True}$
  \Else
    \State $r_u \gets \textsc{\splice{}}(r_u, r_v, P)$
  \EndIf
\EndWhile
\State \algorithmicreturn{} $\mathsf{False}$
\EndProcedure
\end{algorithmic}
\end{algorithm*}

\begin{algorithm*}\caption{$\mathsf{Shiloach\mhyphen Vishkin}$}\label{alg:shiloach_vishkin}
\small
\begin{algorithmic}[1]
\Procedure{Connectivity}{$G$, $\codevar{labels}$, $\largestcomp{}$}
\State $\codevar{changed} \gets \mathsf{true}$
\State $\codevar{candidates} \gets \{ v \in V\ |\ \codevar{labels}[v] \neq \largestcomp{} \}$
\State $\codevar{prev\_labels} \gets \codevar{parents}$ \Comment{A copy}
\While {$\codevar{changed} = \mathsf{true}$}
  \State $\codevar{changed} = \mathsf{false}$
  \ParFor {$\{v \in \codevar{candidates}\}$}
    \ParFor {$\{(u, v) \in d^{+}(u)\}$}
      \State $p_u = \codevar{labels}[u], p_v = \codevar{labels}[v]$
      \State $l = \min(p_u, p_v), h = \max(p_u, p_v)$
      \If {$l \neq h$ and $h = \codevar{prev\_labels}[h]$}
        \State $\textsc{\writemin{}}(\&\codevar{labels}[h], l)$
        \State $\codevar{changed} \gets \mathsf{true}$
      \EndIf
    \EndParFor
  \EndParFor
  \ParFor {$\{v \in V\}$}
    \State $\textsc{FullShortcut}(v, \codevar{labels})$
    \State $\codevar{prev\_labels}[v] \gets \codevar{parents}[v]$
  \EndParFor
\EndWhile
\EndProcedure
\end{algorithmic}
\end{algorithm*}

\end{document}